\documentclass[12pt,a4paper,twoside]{amsart} 
\usepackage{latexsym}
\usepackage{amsmath}
\usepackage{amssymb}
\usepackage{amsfonts}
\usepackage{ifthen}
\usepackage{mathrsfs}               % \mathscr
\usepackage{bbm}                    % \mathbbm 
\usepackage{bbold}                  % \mathbb 
\usepackage{textcomp}               
\usepackage{amstext}
\usepackage{enumerate}
\usepackage{amstext}
\newcommand{\CC}{\mathbb{C}}

\newcommand{\NN}{\mathbb{N}}

\newcommand{\RR}{\mathbb{R}}

\newcommand{\ZZ}{\mathbb{Z}}
\newcommand{\supp}{\mathrm{supp}}
\newcommand{\const}{\mathrm{const}}

\newcommand{\Ran}{\mathrm{Ran}}

\newcommand{\el}{\mathrm{el}}

\newcommand{\phys}{\mathrm{phys}}
\newcommand{\id}{\mathbbm{1}}% Identity
\newcommand{\klg}{\leqslant} % greater/less or equal         
\newcommand{\grg}{\geqslant}          
\newcommand{\ve}{\varepsilon}% alternativ small greak letters
\newcommand{\vp}{\varphi}
\newcommand{\vk}{\varkappa}
\newcommand{\vr}{\varrho}
\newcommand{\vt}{\vartheta}
\newcommand{\vs}{\varsigma}
\newcommand{\vo}{\varpi}
\newcommand{\wt}[1]{\widetilde{#1}}
\newcommand{\SPn}[2]{\langle \,#1\,|\,#2\, \rangle} 
\newcommand{\SPb}[2]{\big\langle \,#1\,\big|\,#2\, \big\rangle} 
\newcommand{\SPB}[2]{\Big\langle \,#1\,\Big|\,#2\, \Big\rangle}
\newcommand{\ol}[1]{\overline{#1}} % overline
\newcommand{\ul}[1]{\underline{#1}}
\newcommand{\wh}[1]{\widehat{#1}}  % widehat
\newcommand{\bigO}{\mathcal{O}}    % bigO
\newcommand{\fourier}{\mathcal{F}} % Fourier transform
\newcommand{\V}[1]{\mathbf{#1}}
\newcommand{\valpha}{\boldsymbol{\alpha}}
\newcommand{\vxi}{\boldsymbol{\xi}}

\newcommand{\vsigma}{\boldsymbol{\sigma}}

\newcommand{\vnu}{\boldsymbol{\nu}}

\newcommand{\veps}{\boldsymbol{\varepsilon}}
\newcommand{\LO}{\mathscr{L}}      % bounded linear operators

\newcommand{\HP}{\mathscr{K}}
\newcommand{\HR}{\mathscr{H}}

\newcommand{\Fock}{\mathscr{F}_{\mathrm{b}}}
\newcommand{\core}{\mathscr{D}_4}
\newcommand{\dom}{\mathcal{D}}
\newcommand{\form}{\mathcal{Q}}
\newcommand{\spec}{\mathrm{\sigma}}

\newcommand{\specpp}{\mathrm{\sigma}_{\mathrm{pp}}}
\newcommand{\SA}{S_{\mathbf{A}}}  % sign functions
\newcommand{\SAt}{S_{\widetilde{\mathbf{A}}}}
\newcommand{\SAm}{S_{\V{A}_m}}
\newcommand{\SAmt}{S_{\widetilde{\V{A}}_m}}
\newcommand{\SAmg}{S_{\V{A}_m^>}}

\newcommand{\DA}{D_{\mathbf{A}}}                % Dirac operators
\newcommand{\DAm}{D_{\mathbf{A}_m}}  
\newcommand{\DAe}{D_{\mathbf{A}_{m,\varepsilon}}} 
\newcommand{\DAme}{D_{\mathbf{A}_{m,\varepsilon}}} 
\newcommand{\DAed}{D_{\mathbf{A}_{m,\varepsilon}^d}}  
\newcommand{\DAmg}{D_{\mathbf{A}_m^>}}  
\newcommand{\DAt}{D_{\widetilde{\mathbf{A}}}}
\newcommand{\DAmt}{D_{\widetilde{\mathbf{A}}_m}}
\newcommand{\DO}{D_{\mathbf{0}}}

\newcommand{\FD}{\widehat{D}_{m,\ve}}
\newcommand{\R}[2]{R_{#1}(#2)}                  % resolvents of the Dirac op.
\newcommand{\RA}[1]{R_{\mathbf{A}}(#1)}
\newcommand{\RAt}[1]{R_{\widetilde{\mathbf{A}}}(#1)}

\newcommand{\Hf}{H_f}                           % field energy
\newcommand{\HT}{\check{H}_f}                   % renormalized field energy 
\newcommand{\Hft}{\widetilde{H}_f}    
\newcommand{\Hfmg}{H_{f,m}^>}
\newcommand{\He}{H_{f,m,\ve}}
\newcommand{\Hed}{H_{f,m,\ve}^d}
\newcommand{\Hef}{H_{f,m,\ve}^f}
\newcommand{\PF}[1]{H_{#1}}                     % semi-rel. Pauli-Fierz op.
\newcommand{\PFt}[1]{\widetilde{H}_{#1}}
\newcommand{\PFmt}[1]{\widetilde{H}_{#1,m}}
\newcommand{\PFm}[1]{H_{#1,m}}
\newcommand{\PFme}[1]{H_{#1,m,\varepsilon}}
\newcommand{\PFo}[1]{H_{#1,0}}
\newcommand{\PFmg}[1]{H_{#1,m}^{>}}
\newcommand{\PFe}[1]{H_{\gamma,\varepsilon}}
\newcommand{\tgV}{\tfrac{\gamma}{|\V{x}|}}
\newcommand{\ad}{a^\dagger}                     % creation operator
\newcommand{\UV}{\Lambda}
\newcommand{\Th}{\Sigma}
\newcommand{\cA}{\mathcal{A}}
\newcommand{\cO}{\mathcal{O}}

\newcommand{\cH}{\mathcal{H}}\newcommand{\cT}{\mathcal{T}}

\newcommand{\cM}{\mathcal{M}}       

\newcommand{\pe}{\mathbf{p}}
\newcommand{\pf}{\mathbf{p}_f}
\newcommand{\pfe}{\mathbf{p}_{f,\ve}}
\newcommand{\pp}{\vxi}
\newcommand{\ps}{{\vxi}_\star}
\newcommand{\al}{\valpha}
\newcommand{\bts}{\mathbf{t}_{{\star}}}
\newcommand{\sC}{\mathscr{C}}
\newcommand{\sD}{\mathscr{D}}

\renewcommand{\Re}{\mathrm{Re}\,}
\newtheorem{theorem}{Theorem}[section]
\newtheorem{lemma}[theorem]{Lemma}
\newtheorem{proposition}[theorem]{Proposition}
\newtheorem{corollary}[theorem]{Corollary}
\newtheorem{hypothesis}[theorem]{Hypothesis}
\theoremstyle{remark}
\newtheorem{remark}[theorem]{Remark}

\numberwithin{equation}{section}
\title[Ground states relativistic QED]{
Existence of ground states of hydrogen-like atoms in
relativistic QED I:\\ The semi-relativistic
Pauli-Fierz operator}
\author{Martin K\"onenberg}
\address{Martin K\"onenberg\\
Fakult\"at f\"ur Mathematik und Informatik\\
FernUniversit\"at Hagen\\
L\"utzowstra{\ss}e 125\\
D-58084 Hagen, Germany.}
\email{martin.koenenberg@fernuni-hagen.de}
\author{Oliver Matte}
\address{Oliver Matte
Institut f\"ur Mathematik\\
TU Clausthal\\
Erzstra{\ss}e 1\\
D-38678 Clausthal-Zellerfeld, Germany\\
{\em On leave from:} Mathematisches Institut\\
Ludwig-Maximilians-Universit\"at\\
Theresienstra{\ss}e 39\\
D-80333 M\"unchen, Germany.}
\email{matte@math.lmu.de}
\author{Edgardo Stockmeyer}
\address{Edgardo Stockmeyer\\
Mathematisches Institut\\
Ludwig-Maximilians-Univer- sit\"at\\
Theresienstra{\ss}e 39\\
D-80333 M\"unchen, Germany.}
\email{stock@math.lmu.de}

\date{\today}
\begin{document}

\begin{abstract}
We consider a hydrogen-like atom in a quantized electromagnetic field
which is modeled by means of the semi-relativistic
Pauli-Fierz operator
and prove that the infimum of the spectrum
of the latter operator is an eigenvalue.
In particular, we verify that the bottom of its spectrum
is strictly less than its ionization threshold.
These results hold true for arbitrary values of the fine-structure
constant and the ultra-violet cut-off as long as the
Coulomb coupling constant (i.e. the product of the
fine-structure constant and the nuclear charge) is less than $2/\pi$.
\end{abstract}

\maketitle

%%%%%%%%%%%%%%%%%%%%%%%%%%%%%%%%%%%%%%%%%%%%%%%%%%%%%%%%%%%%%%%%%%%%%%%%%%
%%%%%%%%%%%%%%%%%%%%%%%%%%%%%%%%%%%%%%%%%%%%%%%%%%%%%%%%%%%%%%%%%%%%%%%%%%
%%%%%%%%%%%%%%%%%%%%%%%%%%%%%%%%%%%%%%%%%%%%%%%%%%%%%%%%%%%%%%%%%%%%%%%%%%

\section{Introduction}
\label{sec-intro}

\noindent
The existence of atoms described in the framework of
non-relativistic quantum electrodynamics (QED)
is by now a well-established fact.
The general picture is roughly that all exited bound states
of an electronic Hamiltonian modeling an atom turn into resonances
when the interaction with the quantized electromagnetic field
is taken into account. Only at the lower end of the spectrum
there remains an eigenvalue corresponding to the ground states
of the atomic system. Its analysis is particularly subtle
as the whole spectrum is continuous up to its minimum in the presence of
the quantized radiation field.
The existence of energy minimizing ground states
for atoms and molecules in non-relativistic QED has been
proven first in \cite{BFS1998b,BFS1999}, for small values
of the involved physical parameters.
The latter are Sommerfeld's fine structure constant, $e^2$,
and the ultra-violet cut-off, $\Lambda$.
The existence of ground states for a molecular
Pauli-Fierz-Hamiltonian has been shown in \cite{GLL2001},
for all values of $e^2$ and $\Lambda$, assuming a certain
binding condition, which has been verified later on
in \cite{BCV2003}, for helium-like atoms, and in
\cite{LiebLoss2002} in full generality.
In the last decade there appeared a large number of
further mathematical contributions  to non-relativistic
QED. Here we only want to mention that ground state energies
and projections have also been studied by means of
infra-red finite algorithms and renormalization group
methods 
\cite{BCFS2003,BFP2006,BFS1998b,BFS1998a,BFS1999,BachKoenenberg2006,FGS2008}.

In contrast to the situation in non-relativistic QED  
only a few mathematical works deal with models where
the quantized radiation field is coupled to 
relativistic particles. 
For instance, in \cite{LiebLoss2002b,LiebLoss2002}
the authors study a relativistic no-pair model of a molecule. 
They prove the stability of matter of the second kind
and give an upper bound on the (positive) binding energy
under certain restrictions on $e^2$, $\Lambda$, and the nuclear
charges. 
In \cite{MatteStockmeyer2009a} two of the present authors
consider a no-pair model of a hydrogenic atom and
study the exponential localization of low-lying 
spectral subspaces.
The same result is established in \cite{MatteStockmeyer2009a}
also for the following operator which is
investigated further in the present paper,
\begin{equation}\label{def-PF2}
\cH_\gamma\,:=\,\sqrt{(\vsigma\cdot(-i\nabla+\V{A}))^2+\id}\,
-\,\frac{\gamma}{|\V{x}|}\,+\,H_f\,.
\end{equation}
Here $\V{A}$ is the quantized vector potential 
in the Coulomb gauge, $H_f$ is the
radiation field energy, $\vsigma$ is a formal vector
containing the Pauli spin matrices, and $\gamma=e^2\,Z>0$ is the Coulomb 
coupling constant, $Z>0$ denoting the
nuclear charge. (The square-root of the fine structure
constant is included in the symbol $\V{A}$, which also
depends on the choice of $\Lambda$.)
Previous mathematical works dealing with this operator
include \cite{MiyaoSpohn2008} where
the fiber decomposition of $\cH_0$
with respect to different values of the total momentum
is studied. We adopt the
nomenclature of the latter paper and call $\cH_\gamma$
the semi-relativistic Pauli-Fierz operator.
It is called semi-relativistic since time and space are
certainly not treated on equal footing. 
Furthermore, the operator $\cH_\gamma$ 
appears in the mathematical analysis of Rayleigh scattering
\cite{FGS2001} which is connected to the phenomenon of
relaxation of an isolated atom to its ground state. 
(The electron spin has been neglected in \cite{FGS2001} for
notational simplicity.)
An advantageous feature of semi-relativistic
Hamiltonians in this situation is that the propagation
speed of the electron is strictly less than the speed of light
(which equals one in the units chosen in \eqref{def-PF2}). 
Moreover, it is shown in \cite{Stockmeyer2009} that $\cH_\gamma$
converges in norm resolvent sense to the non-relativistic 
Pauli-Fierz operator when the speed of light is re-introduced
in \eqref{def-PF2} and send to infinity.
  
We remark that the existence of ground states in
relativistic models of QED where all particles,
including the electrons and positrons, are described by
quantized fields is proven in \cite{BDG2004}. 
To this end the authors employ infra-red cut-offs in the
interaction part of the Hamiltonian which will not be
necessary in our analysis below. 

Thanks to \cite{MatteStockmeyer2009a} we already know that
$\cH_\gamma$ is semi-bounded below on some natural
dense domain, for all $\gamma\in[0,2/\pi]$,
and, hence, has a 
physically distinguished self-adjoint realization.
As already indicated above, it is also shown in \cite{MatteStockmeyer2009a}
that its spectral subspaces corresponding to energies below
the ionization threshold are exponentially localized
with respect to the electron coordinates.
Typically, localization estimates are important ingredients in 
the proofs of the existence of ground states.
Here the ionization threshold equals, by definition, the
infimum of the spectrum of $\cH_0$.
The first result of the present paper states that, 
for every $\gamma\in(0,2/\pi)$, the atomic system
modeled by $\cH_\gamma$ is able to bind an electron.
This means that the infimum of the
spectrum of $\cH_\gamma$ is strictly smaller than the
ionization threshold. (In \cite{MatteStockmeyer2009a}
it has been verified that binding occurs for small values
of $e^2$ and/or $\Lambda$.) 
The main theorem of this article asserts that the operator
$\cH_\gamma$ has an
energy minimizing ground state eigenvector.
This result holds true, for arbitrary values of $e^2$
and $\Lambda$ and for $\gamma\in(0,2/\pi)$.
We remark that the ground state energy 
-- in fact, every speculative eigenvalue -- of
$\cH_\gamma$ is evenly degenerate since $\cH_\gamma$ 
commutes with the time reversal operator \cite{MiyaoSpohn2008}.
In order to proof the existence of ground states 
we combine the strategies employed
in \cite{BFS1998b,BFS1999} and \cite{GLL2001}.
Roughly speaking we construct a sequence of approximating
ground state eigenvectors -- these are ground states of infra-red cut-off
Hamiltonians -- along the lines of \cite{BFS1998b,BFS1999},
and apply a compactness argument very similar to the one 
given in \cite{GLL2001}.
As in \cite{BFS1998b,BFS1999}, where the authors assumed
$e^2$ or $\Lambda$ to be small, 
we prove the existence of ground states for the infra-red cut-off
Hamiltonians by means of a discretization procedure.
A new observation based on the
localization estimates actually permits to
carry out the discretization argument, for all values
of $e^2$ and $\Lambda$.
Another key ingredient in the proofs are infra-red estimates
on the approximating ground state eigenvectors, namely,
a bound on the number of soft photons \cite{BFS1998b,GLL2001}
and a photon derivative bound \cite{GLL2001}.
In order to establish these bounds for the model treated here,
the formal gauge invariance of $\cH_\gamma$ is crucial.
In fact, the no-pair models investigated in 
\cite{LiebLoss2002b,LiebLoss2002,MatteStockmeyer2009a} are gauge
invariant also and
the present authors shall exploit this
property to prove the existence of ground states
for a no-pair model of a hydrogenic atom in the forthcoming
article \cite{KMS2009b}. 
Although the general strategies to prove the existence of ground states
in QED are fairly well-known by now, their application
to the model studied in the present paper and to no-pair
models in QED is non-trivial, mainly due to the non-locality
of the corresponding Hamiltonians. In fact, the electronic kinetic energy
and the quantized vector potential in the Hamiltonian
\eqref{def-PF2} are always linked together in a non-local
way which leads to a variety of new mathematical
problems in each of the steps in the existence proof mentioned above. 
To overcome these difficulties we employ various
commutator estimates involving sign functions of the Dirac
operator, multiplication operators, and the radiation field
energy. Some of them have already been derived
in \cite{MatteStockmeyer2009a}.       

\smallskip

\noindent
{\em This article is organized as follows.}
In the subsequent Section~\ref{sec-model}
we introduce the semi-relativistic Pauli-Fierz operator
and state our main results more precisely.
Section~\ref{sec-pre} summarizes some technical
prerequisites obtained earlier in \cite{MatteStockmeyer2009a}
and provides a number of new results on absolute values
and sign functions of the Dirac operator.  
In Section~\ref{sec-binding} we prove that binding occurs
in our model. Section~\ref{sec-ex}
is devoted to the proof of the existence of ground states
and starts with a brief outline of the strategy.
Finally, in Section~\ref{sec-IR-bounds} we prove the infra-red
bounds. The main text is followed by two appendices. 
In the first one we provide some \`{a}-priori
information on eigenvectors which is required to prove the infra-red
estimates.
In the second one we recall
some basic definitions of operators acting in Fock spaces.

%%%%%%%%%%%%%%%%%%%%%%%%%%%%%%%%%%%%%%%%%%%%%%%%%%%%%%%%%%%%%%%%%%%%%%%%%%
%%%%%%%%%%%%%%%%%%%%%%%%%%%%%%%%%%%%%%%%%%%%%%%%%%%%%%%%%%%%%%%%%%%%%%%%%%
%%%%%%%%%%%%%%%%%%%%%%%%%%%%%%%%%%%%%%%%%%%%%%%%%%%%%%%%%%%%%%%%%%%%%%%%%%

\section{Definition of the model and main results}
\label{sec-model}

\noindent
The
semi-relativistic Pauli-Fierz
operator studied in this article acts in the Hilbert space
$$
\HR_2\,:=\,L^2(\RR^3_\V{x},\CC^2)\otimes\Fock[\HP]
\,\cong\,\sum_{\vs=1,2}\int_{\RR^3_\V{x}}^\oplus
\Fock[\HP]\,d^3\V{x}\,.
$$
Here the bosonic Fock space, which is
the state space of the 
quantized photon field,
$$
\Fock[\HP]\,=\,
\bigoplus_{n=0}^\infty\Fock^{(n)}[\HP]
$$
is modeled over the one photon Hilbert space
$$
\HP\,:=\,L^2(\RR^3\times\ZZ_2,dk)\,,\quad
\int dk\,:=\,\sum_{\lambda\in\ZZ_2}\int_{\RR^3}d^3\V{k}\,.
$$
The letter
$k=(\V{k},\lambda)$ always
denotes a tuple consisting of a photon wave vector,
$\V{k}\in\RR^3$, and a polarization label, $\lambda\in\ZZ_2$.
The components of $\V{k}$ are written as $\V{k}=(k^{(1)},k^{(2)},k^{(3)})$. 
We refer the reader who is not acquainted
to the notation used here to Appendix~\ref{app-Fock},
where the basic definitions of bosonic Fock spaces
and the usual operators acting in them are briefly recalled.
The following subspace is dense in $\HR_2$,
\begin{equation}\label{def-core2-sC}
\sD_2\,:=\,
C_0^\infty(\RR^3_\V{x},\CC^2)\otimes\sC_0\,.\qquad
(\textrm{Algebraic tensor product.})
\end{equation}
Here $\sC_0\subset\Fock[\HP]$ denotes the subspace
of all elements $(\psi^{(n)})_{n=0}^\infty\in\Fock[\HP]$
such that only finitely many components $\psi^{(n)}$
are non-zero and such that each $\psi^{(n)}$ is bounded and
has a compact support.
In order to introduce the quantized vector potential
we first recall the physical choice of the
form factor with sharp ultra-violet cut-off
at $\UV>0$,
\begin{equation}\label{def-Gphys}
\V{G}_\V{x}^\phys(k)\,
\,:=\,
\,e^{-i\V{k}\cdot\V{x}}\,\V{g}(k)\,,
\qquad
\V{g}(k)\,\equiv\,\V{g}^{e,\Lambda}(k)\,:=\,-e
\,\frac{\id_{\{|\V{k}|\klg\UV\}}}{2\pi\sqrt{|\V{k}|}}\,\veps(k)\,,
\end{equation}
for every $\V{x}\in\RR^3$
and almost every $k=(\V{k},\lambda)\in\RR^3\times\ZZ_2$.
Here the square of the elementary charge, $e>0$,
is equal to Sommerfeld's fine-structure constant in our units
where Planck's constant, the speed of light, and the electron mass
are equal to one.
(Energies are measured in units of the rest energy of the electron,
$\V{x}$ is measured in units of one Compton wave length
divided by $2\pi$ and
the photon wave vectors
$\V{k}$ are measured in units of $2\pi$ times
the inverse Compton wave length; we have
$e^2\approx1/137$ in nature.)
Writing 
\begin{equation}\label{def-kbot}
\V{k}_\bot:=(k^{(2)},-k^{(1)},0)\,, \qquad 
\V{k}=(k^{(1)},k^{(2)},k^{(3)})\in\RR^3,
\end{equation}
the polarization vectors are given by 
\begin{equation}\label{pol-vec}
\veps(\V{k},0)\,=\,\frac{\V{k}_\bot}{|\V{k}_\bot|}
\,,\qquad
\veps(\V{k},1)\,=\,
\frac{\V{k}}{|\V{k}|}\,\wedge\,\veps(\V{k},0)\,,
\end{equation}
for almost every $\V{k}\in\RR^3$.
The quantized vector potential, $\V{A}\equiv\V{A}(\V{G}^\phys)$, 
is the triplet of operators given 
by the direct integral
$$
\V{A}=(A^{(1)},A^{(2)},A^{(3)})\,
:=\,\sum_{\vs=1,2}\int_{\RR^3_\V{x}}^\oplus\V{A}(\V{x})\,d^3\V{x}\,,
$$
where, for each fixed $\V{x}$,
$$
\V{A}(\V{x})\,:=\,
\ad(\V{G}_\V{x}^\phys)\,+\,a(\V{G}_\V{x}^\phys)
$$
is acting in the Fock space.
The definition of the bosonic creation and annihilation
operators, $\ad(f)$ and $a(f)$, are recalled
in Appendix~\ref{app-Fock}. For short
we write 
$a^\sharp(\V{f}):=(a^\sharp(f^{(1)}),a^\sharp(f^{(2)}),a^\sharp(f^{(3)}))$,
for a three-vector of functions $\V{f}=(f^{(1)},f^{(2)},f^{(3)})\in\HP^3$, where
$a^\sharp$ is $a$ or $\ad$.
We further set
$$
\V{p}\,:=\,-i\nabla_\V{x}\,,\qquad
\vsigma\cdot(\V{p}+\V{A})\,:=\,
\sum_{j=1}^3\sigma_j\,(-i\partial_{x_j}+A^{(j)})\,,
$$
where $\sigma_1,\sigma_2,\sigma_3$ are the Pauli spin matrices.
An application of Nelson's commutator theorem with
test operator $-\Delta+\Hf$ shows that $\vsigma\cdot(\V{p}+\V{A})$
is essentially self-adjoint on $\sD_2$. 
We denote its closure again by the same symbol and define
$$
\cT_\V{A}\,:=\,\sqrt{(\vsigma\cdot(\V{p}+\V{A}))^2+\id}
$$ 
by means of the spectral calculus.
Now the semi-relativistic Pauli-Fierz operator
is \`{a}-priori given as
\begin{equation}\label{def-PF2b}
\cH_{\gamma}\,\vp\,\equiv\,
\cH_{\gamma,\V{G}^\phys}\,\vp:=\,\big(\cT_\V{A}-\tgV+\Hf\big)\,\vp\,,\qquad
\vp\in\sD_2\,.
\end{equation}
Here the radiation field energy, $\Hf:=d\Gamma(\omega)$,
is given as the second quantization of the
dispersion relation $\omega(k)=|\V{k}|$, 
$k=(\V{k},\lambda)\in\RR^3\times\ZZ_2$; see Appendix~\ref{app-Fock}.
Moreover, we identify $\tgV\equiv\tgV\otimes\id$,
$\Hf\equiv\id\otimes\Hf$, etc. in \eqref{def-PF2b}
and henceforth.

It has been shown in \cite{MatteStockmeyer2009a}
that the quadratic form of $\cH_{\gamma}$ is bounded
from below on $\sD_2$, for $\gamma\in[0,2\pi]$ and
all values of $e,\UV>0$; compare
Inequality~\eqref{maria} below.
In particular, $\cH_{\gamma}$ has a self-adjoint
Friedrichs extension which we again denote by the
same symbol $\cH_{\gamma}$.
In what follows we denote the ground state energy
of $\cH_{\gamma}$ by
$$
E_{\gamma}\,:=\,\inf\spec[\cH_{\gamma}]\,,\qquad \gamma\in(0,2/\pi)\,,
$$
and its ionization threshold by
$$
\Th\,:=\,\inf\spec[\cH_{0}]\,.
$$
The following two theorems are the main results of 
this paper.

\begin{theorem}[{\bf Binding}]\label{thm-binding-PF}
Let $e^2,\UV>0$ and $\gamma\in(0,2/\pi)$.
Then 
$$
\Th\,-\,E_\gamma\,\grg\,|E^\el_{\mathrm{nr},\gamma}|\,,
$$
where $E^\el_{\mathrm{nr},\gamma}=-\gamma^2/2$ 
is the lowest
eigenvalue of the Schr\"odinger operator $-\frac{1}{2}\,\Delta-\tgV$
describing a non-relativistic hydrogenic atom.
\end{theorem}

\begin{proof}
This theorem is a special case of Theorem~\ref{thm-binding}
below.
\end{proof}

\begin{theorem}[{\bf Existence of ground states}]\label{thm-ex-PF}
Let $e^2,\UV>0$ and $\gamma\in(0,2/\pi)$.
Then $E_{\gamma}$ is an evenly
degenerated eigenvalue of $\cH_{\gamma}$.
\end{theorem}

\begin{proof}
It is remarked in \cite[\textsection4]{MiyaoSpohn2008} that
every eigenvalue of $\cH_{\gamma}$
is evenly degenerated. In fact, this
follows from Kramers' degeneracy theorem
since $\cH_{\gamma}$ commutes with the 
anti-unitary time reversal operator
$\vt:=\sigma_2\,C\,R$, $\vt^2=-\id$,
where $C$ denotes complex conjugation and the 
electron parity $R$ replaces $\V{x}$ by $-\V{x}$.
The fact that $E_{\gamma}$ is an eigenvalue
is proved in Section~\ref{sec-ex}.
\end{proof}

\begin{remark}
(i) 
The authors are aware of the fact that 
in Theorem~\ref{thm-binding-PF} 
it would be preferable
to have a bound in terms of the lowest eigenvalue
of $\sqrt{-\Delta+1}-\tgV-1$, which is larger in absolute value.

\smallskip

\noindent
(ii) Every ground state eigenfunction of $\cH_{\gamma}$
is exponentially localized with respect to the
electron coordinates in the $L^2$-sense \cite{MatteStockmeyer2009a};
see Proposition~\ref{prop-exp-loc} where we recall
the precise statement.

\smallskip

\noindent
(iii) Theorems~\ref{thm-binding-PF} and~\ref{thm-ex-PF} actually hold true
for arbitrary choices of the polarization vectors $\veps(\V{k},\lambda)$,
$\lambda\in\ZZ_2$, as long as $\veps(\V{k},\lambda)$ is homogeneous
of degree zero in $\V{k}$ and 
$\{\V{k}/|\V{k}|,\veps(\V{k},0),\veps(\V{k},1)\}$ 
is an orthonormal basis of $\RR^3$, for almost every $\V{k}$.
For in this case the special form \eqref{pol-vec} of
the polarization vectors can always be achieved by a suitable
unitary transformation; see the appendix to \cite{Sasaki2006}
for details. Moreover, the sharp ultra-violet cut-off
in \eqref{def-Gphys} can be replaced by a smooth cut-off
implemented by some rapidly decaying function and
Theorems~\ref{thm-binding-PF} and~\ref{thm-ex-PF} still remain valid.
This follows by inspection of the proofs below.
\end{remark}

%%%%%%%%%%%%%%%%%%%%%%%%%%%%%%%%%%%%%%%%%%%%%%%%%%%%%%%%%%%%%%%%%%%%%%%%%%%%
%%%%%%%%%%%%%%%%%%%%%%%%%%%%%%%%%%%%%%%%%%%%%%%%%%%%%%%%%%%%%%%%%%%%%%%%%%%%
%%%%%%%%%%%%%%%%%%%%%%%%%%%%%%%%%%%%%%%%%%%%%%%%%%%%%%%%%%%%%%%%%%%%%%%%%%%%

\section{The Dirac operator}
\label{sec-pre}

\subsection{Operators acting on four-spinors}

\noindent
It shall be convenient to work with a two-fold
direct sum of the operator $\cH_\gamma$ defined in \eqref{def-PF2b}.
For this permits to exploit earlier results
on sign functions of the free Dirac operator
minimally coupled to the quantized radiation field
and to have a familiar notation in the proofs.
The full Hilbert space we shall work with in the rest of
this paper
is thus given by
$$
\HR_4\,:=\HR_2\oplus\HR_2\,=\,\,L^2(\RR^3_{\V{x}},\CC^4)\otimes\Fock[\HP]\,.
$$
It contains the dense subspace
$$
\core\,:=\,C_0^\infty(\RR^3_{\V{x}},\CC^4)\otimes\sC_0\,.
\quad\textrm{(Algebraic tensor product.)}
$$
In order to introduce the Dirac operator  
we first recall that the Dirac matrices
$\alpha_1,\alpha_2,\alpha_3$, and $\beta=\alpha_0$
are hermitian $(4\times4)$-matrices
obeying the Clifford algebra relations
 \begin{equation}\label{Clifford}
\alpha_i\,\alpha_j\,+\,\alpha_j\,\alpha_i\,=\,2\,\delta_{ij}\,\id\,,
\qquad i,j\in\{0,1,2,3\}\,.
\end{equation}
In the standard representation they are given in terms of
the Pauli matrices as $\alpha_j=\sigma_1\otimes\sigma_j$,
$j\in\{1,2,3\}$, and $\beta=\sigma_3\otimes\id$.
We shall also work with generalized form factors
in what follows. For many of the technical results
stated below are applied to truncated and discretized versions
of the physical form factor \eqref{def-Gphys}.
Moreover, this permits to apply some of the technical
results of this article in our forthcoming work.
Hence, it makes sense to introduce the following
hypothesis.

\begin{hypothesis}\label{hyp-G}
Let $\cA:=\{\V{k}\in\RR^3:\,|\V{k}|\grg m\}$, for some $m\grg0$,
and let $\vo:\cA\times\ZZ_2\to[0,\infty)$ be a measurable
function that depends on $\V{k}\in\cA$ only
such that $0<\vo(k)\klg|\V{k}|$, for 
$k=(\V{k},\lambda)\in\cA\times\ZZ_2$ with $\V{k}\not=0$.
For almost every 
$k\in\cA\times\ZZ_2$ and $j\in\{1,2,3\}$, let
$G^{(j)}(k)$
be a bounded continuously differentiable function,
$\RR^3_{\V{x}}\ni\V{x}\mapsto G^{(j)}_{\V{x}}(k)$,
such that the map $(\V{x},k)\mapsto G^{(j)}_{\V{x}}(k)$
is measurable, 
\begin{equation}
\label{def-d3}
2\int\vo(k)^{\ell}\,\|\V{G}(k)\|^2_\infty\,dk
\,\klg\,d_\ell^2\,,
\qquad \ell\in\{-1,0,1,2\}\,,
\end{equation}
where $\int dk:=\sum_{\lambda\in\ZZ_2}\int_{\cA}d^3\V{k}$, and 
\begin{equation}\label{hyp-rotG}
2\int\vo(k)^{-1}\,\|\nabla_{\V{x}}\wedge\V{G}(k)\|^2_\infty\,dk
\,\klg\,d_{1}^2\,,
\end{equation}
for some $d_{-1},d_0,d_1,d_2\in(0,\infty)$, where
$\|\V{G}(k)\|_\infty:=\sup_{\V{x}}|\V{G}_{\V{x}}(k)|$.
\end{hypothesis}

\smallskip

\noindent
The generalized interaction between matter
and radiation is now given as
$$
\valpha\cdot\V{A}
\,:=\,\valpha\cdot\big(\ad(\V{G})+a(\V{G})\big)
\,:=\,\sum_{\vs=1,2,3,4}\int_{\RR^3_\V{x}}^\oplus
\valpha\cdot \V{A}(\V{x})\,d^3\V{x}\,,
$$
where
$$
\valpha\cdot\V{A}(\V{x})
\,:=\,\valpha\cdot\ad(\V{G}_\V{x})+\valpha\cdot a(\V{G}_\V{x})\,,
\quad
\valpha\cdot a^\sharp(\V{G}_\V{x})\,:=\,
\sum_{j=1}^3\,\alpha_j\,a^\sharp(G^{(j)}_\V{x})\,.
$$
Under Hypothesis~\ref{hyp-G}
we have the following well-known
relative bounds 
showing that $\valpha\cdot\V{A}$ is a symmetric operator
on $\dom(d\Gamma(\vo)^{1/2})$.
For every $\psi\in\dom(d\Gamma(\vo)^{1/2})$,
\begin{align}
\|\valpha\cdot a(\V{G})\,\psi\|^2
\,&\klg\,d_{-1}^2 \,\|d\Gamma(\vo)^{1/2}\,\psi\|^2,\label{rb-a}
\\
\|\valpha\cdot\ad(\V{G})\,\psi\|^2\,&\klg\,
d_{-1}^2\,\|d\Gamma(\vo)^{1/2}\,\psi\|^2\,+\,d_0^2\,\|\psi\|^2,\label{rb-ad}
\\
\|\valpha\cdot\V{A}\,\psi\|^2\,&\klg\,
d_*^2\,\|(d\Gamma(\vo)+1)^{1/2}\,\psi\|\,,\qquad\quad d_*^2\,:=\,
d_0^2+2d_{-1}^2\,.
\end{align}
(Notice that the $C^*$-equality and \eqref{Clifford} imply
$\|\valpha\cdot \V{u}\|=|\V{u}|$, for every $\V{u}\in\RR^3$, whence
$\|\valpha\cdot \V{z}\|^2\klg2|\V{z}|^2$, for every $\V{z}\in\CC^3$.
For this reason we put the 
factor $2$ on the left sides of \eqref{def-d3}
and \eqref{hyp-rotG}.)
The
free Dirac operator minimally coupled to $\V{A}$ is now given as
\begin{equation}\label{def-DA}
\DA\,:=\,\valpha\cdot(\pe+\V{A})+\beta
\,:=\,
\sum_{j=1}^3\alpha_j\,(-i\partial_{x_j}+A^{(j)})\,+\,\beta
\,.
\end{equation}
$\DA$ is essentially self-adjoint
on $\core$ as a
straightforward application of Nelson's commutator theorem shows 
\cite{LiebLoss2002,MiyaoSpohn2008}. We use the symbol
$\DA$ again to denote its closure starting from $\core$. 
Its spectrum is contained
in the union of two half-lines,
$
\spec(\DA)\subset(-\infty,-1]\cup[1,\infty)
$.
Next, we define the semi-relativistic Pauli-Fierz
operator acting on four-spinors, 
$\PF{\gamma}\equiv \PF{\gamma,\V{G},\vo}$, \`{a}-priori by
\begin{equation}\label{def-PF}
\PF{\gamma}\,\vp\,
:=\,\big(|\DA|-\tgV+d\Gamma(\vo)\big)\,\vp\,,\qquad
\vp\in\core\,.
\end{equation}
In \cite{MatteStockmeyer2009a} two of the present authors
proved the inequality
\begin{equation}\label{maria}
\frac{2}{\pi}\frac{1}{|\V{x}|}\,\klg\,|\DA|+\delta\,d\Gamma(\vo)
+(\delta^{-1}+\delta\,k^2)\,d_1^2\,,
\end{equation}
for some $k\in(0,\infty)$ and every $\delta>0$,
in the sense of quadratic forms on $\core$.
It implies  that, for every $\gamma\in[0,2/\pi]$, 
the operator $\PF{\gamma}$
has a self-adjoint Friedrichs extension which is henceforth
again denoted by the same symbol.
For $\gamma\in[0,2/\pi)$, the KLMN theorem further implies that 
$\form(\PF{\gamma})=\form(\PF{0})$
and that $\core$ is a form core for $\PF{\gamma}$.  
Here $\form$ denotes the form domain of an operator.
(We actually know that 
$\form(\PF{\gamma})=\form(|\DO|)\cap\form(d\Gamma(\vo))$,
for $\gamma\in[0,2/\pi)$ \cite{Stockmeyer2009}.)
We observe that in the case $\V{G}=\V{G}^\phys$ 
the operator defined in \eqref{def-PF} is a two-fold
copy of the one given in \eqref{def-PF2b} since
$$
|\DA|\,=\, 
\begin{pmatrix}
           \cT_\V{A}&0\\0&\cT_\V{A}
\end{pmatrix}\,,
\qquad \cT_\V{A}\,:=\,\sqrt{(\vsigma\cdot(\pe+\V{A}))^2+\id}\,.
$$

\subsection{A survey of earlier results}

\noindent
In what follows we collect some basic estimates
we shall need in the sequel. All of them have been derived
in \cite[\textsection3]{MatteStockmeyer2009a}.
As in \cite{MatteStockmeyer2009a} we introduce the parameter
$$
\delta_\nu^2\,\equiv\,
\delta_\nu^2(E)\,:=\,
8\int\frac{w_\nu(k,E)^2}{\vo(k)}\,\sup_{\V{x}\in\RR^3}
|\V{G}_\V{x}(k)|^2\,dk\,,\quad E,\nu>0\,,
$$
where
$$
w_\nu(k,E)\,:=\,
E^{1/2-\nu}\,
\big((E+\vo(k))^{\nu+1/2}-E^\nu\,(E+\vo(k))^{1/2}\big)\,,
$$
and note that
\begin{equation}\label{verona}
\delta_{1/2}\,\klg\,2d_1\,.
\end{equation}
Moreover, given some $E>0$ and some $\vo$ as in Hypothesis~\ref{hyp-G}, we set
$$
\HT\,:=\,d\Gamma(\vo)\,+\,E\,.
$$

\begin{lemma}\label{cor-T-Xi}
Assume that $\vo$ and $\V{G}$ fulfill Hypothesis~\ref{hyp-G}.
Then the following assertions hold true:

\smallskip

\noindent
(i) For all $\nu,E>0$, the densely defined operator
$[\HT^{-\nu}\,,\,\valpha\cdot\V{A}]\,\HT^{\nu}$ extends
to a bounded operator on $\HR_4$ which we denote
by $T_\nu$. We have
$$
\|T_\nu\|\,\klg\,\delta_\nu/E^{1/2}\,,\qquad \nu,E>0\,.
$$
(ii)
Let $z\in\CC$ and $L\in\LO(L^2(\RR^3_{\V{x}},\CC^4))$ be such that
$z\in\vr(\DA)\cap\vr(\DA+L)$ (where $L\equiv L\otimes\id$ and
$\vr$ denotes the resolvent set)
and define
\begin{equation}\label{def-HT-RAL}
\R{\V{A},L}{z}\,:=\,(\DA+L-z)^{-1}\,,\qquad \R{\V{A}}{z}\,:=\,
\R{\V{A},0}{z}\,.
\end{equation}
Assume that $\nu,E>0$ satisfy $\delta_\nu/E^{1/2}<1/\|\R{\V{A},L}{z}\|$.
Then the Neumann series 
\begin{align*}
\Xi_{\nu,L}(z)\,&:=\,\sum_{j=0}^\infty\{-\R{\V{A},L}{z}\,T_\nu\}^j\,,
\qquad
\Upsilon_{\nu,L}(z)\,:=\,\sum_{j=0}^\infty\{-T_\nu^*\,\R{\V{A},L}{z}\}^j
\,,
\end{align*}
converge absolutely,
$\Upsilon_{\nu,L}(z)=\Xi_{\nu,-L}(\ol{z})^*$, and
\begin{equation}\label{bd-Xi}
\|\Xi_{\nu,L}(z)\|\,,\|\Upsilon_{\nu,L}(z)\|\,\klg\,\big(1-\delta_\nu\,
\|\R{\V{A},L}{z}\|/E^{1/2}\big)^{-1}\,.
\end{equation}
(iii) Under the assumptions of (ii) the following operator
identities hold true on $\HR_4$,
\begin{eqnarray}
\HT^{-\nu}\,\R{\V{A},L}{z}
&=&\label{eva3}
\Xi_{\nu,L}(z)
\,\R{\V{A},L}{z}\,\HT^{-\nu}\,,
\\
\R{\V{A},L}{z}\,\HT^{-\nu}
&=&\label{eva3b}
\HT^{-\nu}\,\R{\V{A},L}{z}\,\Upsilon_{\nu,L}(z)
\,.
\end{eqnarray}
In particular, $\R{\V{A},L}{z}$ maps $\dom(\id\otimes\HT^\nu)$
into itself.
\end{lemma}

\smallskip

\noindent
In the next lemma we summarize some results
on the sign function of the Dirac operator,
$$
\SA\,:=\,\DA\,|\DA|^{-1}\,,
$$
which have 
essentially been obtained in \cite[\textsection3]{MatteStockmeyer2009a}.
We define $J:[0,1)\to\RR$ by
\begin{equation}\label{def-J(a)}
J(0)\,:=\,1\,,\qquad
J(a)\,:=\,\sqrt{6}/(1-a^2)\,,\quad a\in(0,1)\,.
\end{equation}

\begin{lemma}
Assume that $\vo$ and $\V{G}$ fulfill Hypothesis~\ref{hyp-G} and 
suppose that
$F\in C^\infty(\RR^3_\V{x},[0,\infty))$
satisfies $|\nabla F|\klg a$, for some $a\in[0,1)$.
Moreover, let $\nu\grg0$, set $\HT=d\Gamma(\vo)+E$, and assume that
$E>(\delta_\nu J(a))^2$.
Then
\begin{align}
\big\|\,e^F\,\HT^\nu\,\SA\,\HT^{-\nu}\,e^{-F}\,\big\|\,&\klg\,
\frac{1+a\,J(a)}{1-\delta_\nu J(a)/E^{1/2}}\,.\label{SA-eF-Hnu}
\end{align}
Moreover, $\SA$ maps the domain of $\HT^\nu$ into itself.
\end{lemma}

\begin{proof}
First, we assume in addition that $F$ is bounded
but allow $F$ to be either non-negative or non-positive. Then
it follows from \cite[Lemma~3.5]{MatteStockmeyer2009a}
that $\|e^{-F}\,\SA\,e^{F}\|\klg 1+a\,J(a)$.
Moreover, we have
$$
\big\|\,e^{-F}\,[\HT^{-\nu}\,,\,\SA]\,\HT^{\nu}\,e^F\,\big\|
\,\klg\,(1+a\,J(a))\,\frac{\delta_\nu\,J(a)/E^{1/2}}{
1-\delta_\nu\,J(a)/E^{1/2}}\,,
$$
due to \cite[Lemma~3.3]{MatteStockmeyer2009a}.
Writing 
$$
e^{-F}\,\HT^{-\nu}\,\SA\,\HT^\nu\,e^F
\,=\,e^{-F}\,\SA\,e^F+
e^{-F}\,[\HT^{-\nu}\,,\,\SA]\,\HT^{\nu}\,e^F
$$
and combining these two inequalities we obtain \eqref{SA-eF-Hnu},
for bounded $F$ having a fixed sign.

Let us now assume that 
$F\grg0$ is not necessarily bounded.
Then we pick a sequence of bounded smooth functions
$F_1,F_2,\ldots\in C^\infty(\RR^3,[0,\infty))$ such that
$|\nabla F_n|\klg a$ and
$F_n=F$ on $\{|\V{x}|\klg n\}$, $n\in\NN$, and
$F_n\to F$, as $n\to\infty$.
Since every $\vp\in\core$ has a compact support with respect
to $\V{x}$ we then obtain
$e^{-F_n}\,\HT^{-\nu}\,\SA\,\HT^{\nu}\,e^{F_n}\,\vp\to
e^{-F}\,\HT^{-\nu}\,\SA\,\HT^{\nu}\,e^F\,\vp$ by the
dominated convergence theorem. 
Since the operators $e^{-F_n}\,\HT^{-\nu}\,\SA\,\HT^{\nu}\,e^{F_n}$
obey the bound \eqref{SA-eF-Hnu} with $F=-F_n$ uniformly in $n\in\NN$,
we conclude that 
$e^{-F}\,\HT^{-\nu}\,\SA\,\HT^{\nu}\,e^F\!\!\upharpoonright_{\core}$
is bounded and its norm is bounded by the right side of \eqref{SA-eF-Hnu} also.
But this is true if and only if its adjoint,
$e^F\,\HT^\nu\,\SA\,\HT^{-\nu}\,e^{-F}$, belongs to
$\LO(\HR_4)$ and satisfies \eqref{SA-eF-Hnu} as well.
(Here we use the facts that $(ST)^*=T^*S^*$ when $ST$ is densely
defined and $S$ is bounded and that $\HT^\nu$ and $e^F$
commute since they act on different tensor factors.)
\end{proof}

%%%%%%%%%%%%%%%%%%%%%%%%%%%%%%%%%%%%%%%%%%%%%%%%%%%%%%%%%%%%%%%%%%%%%%%%%%%%

\subsection{Comparison between operators with different form factors}

\noindent
In the following we assume that
$\wt{G}_{\V{x}}^{(j)}(k)$, $k\in\cA\times\ZZ_2$, $j\in\{1,2,3\}$,
is another form factor fulfilling Hypothesis~\ref{hyp-G} 
with new constants $\wt{d}_{-1},\ldots,\wt{d}_2$, that is,
\begin{equation}\label{def-dtilde}
2\int\vo(k)^\ell\,\sup_{\V{x}\in\RR^3}|\wt{\V{G}}_{\V{x}}(k)|^2\,dk
\,\klg\,\wt{d}_\ell^2
\,<\,\infty\,,\qquad \ell\in\{-1,0,1,2\}\,.
\end{equation}
We write $\wt{\V{A}}=\ad(\wt{\V{G}})+a(\wt{\V{G}})$ and 
assume further that
\begin{equation}\label{def-delta(a)}
\triangle^2_{\ell}(a)\,:=\,2
\int\vo(k)^{\ell}\,\sup_{\V{x}\in\RR^3}\big\{\,e^{-a|\V{x}|}
|\V{G}_{\V{x}}(k)-\wt{\V{G}}_{\V{x}}(k)|^2\,\big\}\,
\,dk\,<\,\infty\,,
\end{equation}
for $\ell\in\{-1,0\}$ and some $a\grg0$.
Then the bounds \eqref{rb-a} and \eqref{rb-ad} still hold
true when $\V{G}$ is replaced by $\wt{\V{G}}$, 
provided at the same time $d_\ell$ is replaced by
$\wt{d}_\ell$.
Likewise we have
\begin{align}
\big\|\,e^{-a|\V{x}|}\,\valpha\cdot a(\V{G}-\wt{\V{G}})\,
\psi\,\big\|^2
\,&\klg\,\label{delta(a)1}
\triangle_{-1}^2(a)\,\|d\Gamma(\vo)^{1/2}\,\psi\|^2\,,
\\
\big\|\,e^{-a|\V{x}|}\,\valpha\cdot\ad(\V{G}-\wt{\V{G}})\,
\psi\,\big\|^2
\,&\klg\,\label{delta(a)2}
\triangle_{-1}^2(a)\,\|d\Gamma(\vo)^{1/2}\,\psi\|^2\,
+\,\triangle_0^2(a)\,\|\psi\|^2\,,
\\ \label{delta(a)3}
\big\|\,e^{-a|\V{x}|}\,\valpha\cdot(\V{A}-\wt{\V{A}})\,\psi\,\big\|^2
\,&\klg\,
\triangle_*^2(a)\,\big\|\,(d\Gamma(\vo)+1)^{1/2}\,\psi\,\big\|^2\,,
\end{align}
where $\triangle_*^2(a):=2\,\triangle_0^2(a)+4\,\triangle_{-1}^2(a)$,
for every $\psi\in\dom(d\Gamma(\vo)^{1/2})$.
Next, we state some simple facts 
which are used in the proofs of the 
lemmata below:
First, we have the following representation of the sign
function of $\DA$ as a strongly convergent principal value
\cite[Lemma~VI.5.6]{Kato}, 
\begin{equation}\label{sgn}
\SA\,\varphi
\,=\,\DA\,|\DA|^{-1}\,\vp
\,=\,\lim_{\tau\to\infty}\,
\int_{-\tau}^\tau\RA{iy}\,\varphi\,\frac{dy}{\pi}\,, \qquad 
\varphi\,\in\,\HR\,.
\end{equation}
Furthermore, since
$(-1,1)\subset\vr(\DA)$ the spectral calculus yields,
for all $y\in\RR$ and  $\kappa\in [0,1)$, 
\begin{equation}\label{ralf3}
\big\|\,|\DA|^{\kappa}\,\RA{iy}\,\big\|\,\klg\,
\frac{\id_{|y|<b(\kappa)}}{
\sqrt{1+y^2}}\:+\:
\frac{c(\kappa)\,\id_{|y|\grg b(\kappa)}}{|y|^{1-\kappa}}
\,=:\,\zeta_\kappa(y)
\,,
\end{equation} 
where $b(\kappa):=\kappa^{-1/2}(1-\kappa)^{1/2}$ ($1/0:=\infty$),
$c(\kappa):=\kappa ^{\kappa/2}(1-\kappa)^{(1-\kappa)/2}$.
We shall often encounter the constants
\begin{equation}\label{def-K(kappa)}
K(0)\,:=\,\frac{1}{2}\,,\qquad
K(\kappa)\,:=\,\int_\RR\frac{\zeta_\kappa(y)}{\sqrt{1+y^2}}\:
\frac{dy}{2\pi}\,<\,\infty\,,\quad\kappa\in(0,1)\,.
\end{equation}
The next lemma shows that
the resolvent of $\DA$ stays bounded after conjugation
with exponential weights $e^F$ acting on the electron coordinates. 
This assertion is 
well-known in the case of classical magnetic fields; see, e.g.,
\cite{BeGe1987}. 
The proof presented in \cite[Lemma~3.1]{MatteStockmeyer2008b}
for classical vector potentials
applies, however, also to
quantized fields without any change and we refrain from
repeating it here.

\begin{lemma}\label{le-marah}
Assume that $\V{G}$ fulfills Hypothesis~\ref{hyp-G}.
Let 
$y\in\RR$, $a\in[0,1)$,
and let
$F\in C^\infty(\RR_\V{x}^3,\RR)$ have a fixed sign and
satisfy $|\nabla F|\klg a$.
Then 
$iy\in\vr(\DA+i\valpha\cdot\nabla F)$,
\begin{equation}\label{marah0}
e^F\,\RA{iy}\,e^{-F}
=(\DA+i\valpha\cdot \nabla F
-iy)^{-1}\!\!\upharpoonright_{\dom(e^{-F})}\,,
\end{equation}
where $e^F\equiv e^F\otimes\id$, and
\begin{equation}\label{marah1}
\big\|\,e^F\,\RA{iy}\,e^{-F}\,\big\|\,\klg\,
\frac{J(a)}{\sqrt{1+y^2}}
\,,
\end{equation}
where $J$ is defined in \eqref{def-J(a)}.
\end{lemma}

\begin{lemma}\label{le-A-Aeps}
Assume that $\vo$, $\V{G}$, and $\wt{\V{G}}$ fulfill
Hypothesis~\ref{hyp-G} such that \eqref{def-dtilde}
and \eqref{def-delta(a)} are satisfied, for some $a\in[0,1)$.
Let $\kappa \in [0,1)$ and assume that 
$F\in C^\infty(\RR^3_\V{x},[0,\infty))$
satisfies $|\nabla F(\V{x})|\klg a$ and $F(\V{x})\grg a|\V{x}|$,
for all $\V{x}\in\RR^3$, and 
$F(\V{x})=a|\V{x}|$, for large $|\V{x}|$.
Then we have, for all $E\grg1$ with $E>(2d_1J(a))^2$,
\begin{align}
\big\|\,|\DAt|^\kappa\,(\SAt-\SA)\,
\HT^{-1/2}\,e^{-F}\,\big\|
&\,\klg\,\frac{2K(\kappa)\,\triangle_*(a)}{1-2d_1J(a)/E^{1/2}}\,.\label{peki1}
\end{align}
Here $\HT:=d\Gamma(\vo)+E$
and $\triangle_*(a)$ is defined after \eqref{delta(a)3}.
\end{lemma}

\begin{proof}
We define $L:=i\valpha\cdot\nabla F$ so that $\|L\|\klg a$.
Then
a short computation using \eqref{eva3} and \eqref{marah0} yields,
for every
$\vp\in\core$, 
\begin{align*}
&e^{-F}\,\HT^{-1/2}\,\big(\RA{-iy}-\RAt{-iy}\big)\,(\DAt+iy)\,\vp
\\
&\,=\,
\Xi_{1/2,-L}(-iy)\,\R{\V{A},-L}{-iy}\,\HT^{-1/2}\,e^{-F}\,
\valpha\cdot(\wt{\V{A}}-\V{A})\,
\RAt{-iy}\,(\DAt+iy)\,\vp\,.
\end{align*}
Now, $(\DAt+iy)\,\core$ is dense in $\HR$ 
since $\DAt$ is essentially self-adjoint on $\core$
and 
$\HT^{-1/2}\,e^{-F}\,\valpha\cdot(\V{A}-\wt{\V{A}})$ is bounded
due to \eqref{delta(a)3}.
Therefore, the previous computation 
implies an operator identity in $\LO(\HR)$
whose adjoint reads
\begin{align*}
\big(\RA{iy}-&\RAt{iy}\big)\,\HT^{-1/2}\,e^{-F}
\\
&=\,
\RAt{iy}\,\valpha\cdot(\wt{\V{A}}-\V{A})\,e^{-F}\,
\HT^{-1/2}\,\R{\V{A},L}{iy}\,\Upsilon_{1/2,L}(iy)\,.
\end{align*}
Combining this with \eqref{sgn} we find, for $\phi,\psi\in\core$,
\begin{align*}
\big|&\SPb{|\DAt|^\kappa\,\phi}{(\SA-\SAt)\,
\HT^{-1/2}\,e^{-F}\,\psi}\big|
\\
&=\,
\int_\RR\Big|\SPB{|\DAt|^\kappa\,\phi}{
\RAt{iy}\,\valpha\cdot(\wt{\V{A}}-\V{A})\,e^{-F}
\,\HT^{-1/2}\,\R{\V{A},L}{iy}\,\Upsilon_{1/2,L}(iy)
\,\psi}\Big|
\,\frac{d\eta}{\pi}
\\
&\klg\,
\int_\RR\big\|\,|\DAt|^\kappa\,\RAt{iy}\,\big\|\,\|\phi\|\,
\\
& \qquad\;\;\cdot\;
\|\,\valpha\cdot e^{-F}\,(\V{A}-\wt{\V{A}})
\,\HT^{-1/2}\,\|\,\|\R{\V{A},L}{iy}\|
\,\|\Upsilon_{1/2,L}(iy)\|
\,\|\psi\|\,\frac{dy}{\pi}\,.
\end{align*}
On account of \eqref{ralf3} and
\eqref{verona}, \eqref{bd-Xi}, and \eqref{marah1}, which
imply $\|\Upsilon_{1/2,L}(iy)\|\klg(1-2d_1J(a)/E^{1/2})^{-1}$,
the previous estimate proves \eqref{peki1}. 
\end{proof}

\begin{lemma}
Assume that $\vo$, $\V{G}$, and $\wt{\V{G}}$
fulfill Hypothesis~\ref{hyp-G} such that
\eqref{def-dtilde} and \eqref{def-delta(a)} are satisfied,
for some $a\in[0,1)$. 
Let $F\in C^\infty(\RR^3_\V{x},[0,\infty))$
satisfy $|\nabla F(\V{x})|\klg a$ and $F(\V{x})\grg a|\V{x}|$,
for all $\V{x}\in\RR^3$, and $F(\V{x})=a|\V{x}|$,
for large $|\V{x}|$. Set
$\HT:=d\Gamma(\vo)+E$.
Then, for every $E\grg1$ with $E>(2d_1)^2$, there is
some $C\equiv C(a,E,d_1)\in(0,\infty)$
such that, for all $\epsilon,\tau\in(0,1]$ and $\vp\in\core$,
\begin{align}\nonumber
\big|\SPb{\vp}{&\big(|\DA|-|\DAt|\big)\,\vp}\big|
\\
&\klg\,\label{pia0}
\epsilon\,\big\|\,|\DAt|^{1/2}\,\vp\,\big\|^2\,
+\,\tau\,\big\|\,\HT^{1/2}\,e^{F}\,\vp\,\big\|^2
\,+\,\frac{C\,\triangle_*^4(a)}{\epsilon^3\,\tau^2}\,\|\vp\|^2\,.
\end{align}
\end{lemma}

\begin{proof}
Since $\SA$ maps $\core$ into the domains of
$\DO$ and $\Hf^{1/2}$ (compare \cite[Lemma~3.4(ii)]{MatteStockmeyer2009a})
we have the following identity on $\core$,
\begin{equation}\label{pia1}
|\DA|-|\DAt|\,=\,
\DAt\,S_\Delta\,+\,\valpha\cdot(\V{A}-\wt{\V{A}})\,\SA\,,
\end{equation}
where
$
S_\Delta=\SA-\SAt
$. 
On account of \eqref{delta(a)3} and
\eqref{SA-eF-Hnu} (where $\delta_{1/2}\klg2d_1$)
the second term on the right side of \eqref{pia1} can be estimated as
\begin{align*}
&\big|\SPb{\vp}{\valpha\cdot(\V{A}-\wt{\V{A}})\,\SA\,\vp}\big|
\\
&\klg\,
\|\vp\|\,
\big\|\,e^{-F}\valpha\cdot(\V{A}-\wt{\V{A}})\,\HT^{-1/2}\big\|
\,\big\|e^F\,\HT^{1/2}\,\SA\,\HT^{-1/2}\,e^{-F}\big\|\,
\big\|e^F \,\HT^{1/2}\,\vp\big\|
\\
&\klg\,\tau\,\big\|e^F\,\HT^{1/2}\,\vp\big\|^2\,+\,
\frac{1}{4\tau}\cdot\frac{\triangle_*^2(a)\,(1+a\,J(a))^2}{(1-2d_1/E^{1/2})^2}\:
\|\vp\|^2\,,
\end{align*}
for all $\vp\in\core$ and $\tau>0$.
Next, we treat the first term on the right hand side of \eqref{pia1}.
By virtue of Lemma~\ref{le-A-Aeps} with $\kappa=3/4$ we 
find some $C_*\in(0,\infty)$ such that, for all $\vp\in\core$,
\begin{align}
\big\|\,|\DAt|^{1/2}&\,S_\Delta\,\vp\,\big\|^2\,
\klg\,\nonumber
\big\|\,|\DAt|^{1/4}\,S_\Delta\,\vp\,\big\|\,
\big\|\,|\DAt|^{3/4}\,S_\Delta\,\vp\,\big\|
\\
&\klg\,\nonumber
\big\|\,|\DAt|^{1/4}\,S_\Delta\,\vp\,\big\|\,C_*\,\triangle_*(a)\,
\big\|\,e^F\,\HT^{1/2}\,\vp\,\big\|
\\
&\klg\,\nonumber
\frac{\tau}{2}\,\big\|\,e^F\,\HT^{1/2}\,\vp\,\big\|^2\,+\,
\frac{C_*^2\,\triangle_*^2(a)}{2\tau}
\,\SPb{S_\Delta\,\vp}{|\DAt|^{1/2}\,S_\Delta\,\vp}
\\
&\klg\,\label{pia2}
\frac{\tau}{2}\,\big\|\,e^F\,\HT^{1/2}\,\vp\,\big\|^2
\,+\,\frac{1}{2}\,\big\|\,|\DAt|^{1/2}\,S_\Delta\,\vp\,\big\|^2
\,+\,\frac{C_*^4\,\triangle_*^4(a)}{8\tau^2}\cdot 2^2\,\|\vp\|^2\,.
\end{align}
In the last step we also used that $\|S_\Delta\|\klg2$.
Solving \eqref{pia2} for $\|\,|\DAt|^{1/2}\,S_\Delta\,\vp\|^2$
and replacing $\tau$ by $4\ve\,\tau$
we arrive at
\begin{align}
\big|\SPb{\vp}{\DAt\,S_\Delta\,\vp}\big|\,&\klg
\,\ve\,\big\|\,|\DAt|^{1/2}\,\vp\,\big\|^2\,+\,\frac{1}{4\ve}\,
\big\|\,|\DAt|^{1/2}\,S_\Delta\,\vp\,\big\|^2\nonumber
\\
&\klg\,\nonumber
\ve\,\big\|\,|\DAt|^{1/2}\,\vp\,\big\|^2
\,+\,\tau\,\big\|\,e^F\,\HT^{1/2}\,\vp\,\big\|^2
\,+\,
\frac{C_*^4}{64\,\ve^3\,\tau^2}\,\|\vp\|^2\,.
\end{align}
\end{proof}

\begin{corollary}\label{cor-pia}
Assume that $\vo$, $\V{G}$, and $\wt{\V{G}}$ fulfill
Hypothesis~\ref{hyp-G} 
such that \eqref{def-dtilde} and \eqref{def-delta(a)} hold true 
with $a=0$. Then,
for every $\gamma\in[0,2/\pi)$ and $\tau\in(0,1]$,
we find $\ve,C\equiv C(\ve,\gamma,\tau,d_1)\in(0,\infty)$  
such that
\begin{equation}
|\DA|-\gamma/|\V{x}|+\tau\,d\Gamma(\vo)\,\grg\,
\ve\,\big(|\DAt|+1/|\V{x}|+d\Gamma(\vo)\big)\,-\,C\,,
\end{equation}
in the sense of quadratic forms on $\core$.
\end{corollary}

\begin{proof}
We choose 
$\epsilon\in(0,1]$ such that $(\gamma+\epsilon)/(1-\epsilon)\klg2/\pi$.
Then \eqref{maria} with 
$\delta$ replaced by $\tau/(2-2\epsilon)$ implies
\begin{align*}
|\DA|-\frac{\gamma}{|\V{x}|}+\tau\,d\Gamma(\vo)\,\grg\,\epsilon\,|\DAt|
+\epsilon\,\big(\,|\DA|-|\DAt|\,\big)
+\frac{\epsilon}{|\V{x}|}
+\frac{\tau}{2}\,d\Gamma(\vo)-C\,,
\end{align*}
for some $C\in(0,\infty)$.
Applying \eqref{pia0} with $F=0$ and $\tau$ replaced by $\tau/4$
we obtain
$$
|\DA|-\frac{\gamma}{|\V{x}|}+\tau\,d\Gamma(\vo)\,
\grg\,(\epsilon-\epsilon^2)\,|\DAt|\,
+\,\frac{\epsilon}{|\V{x}|}\,+\,\frac{\tau}{4}\,d\Gamma(\vo)\,
-\,C'\,,
$$
for some $C'\in(0,\infty)$,
which implies the statement of the corollary.
\end{proof}

\begin{corollary}\label{cor-bd-on-Th}
Assume that $\vo$ and $\V{G}$ fulfill Hypothesis~\ref{hyp-G}.
Then we find some constant, $c(d_1)\in(0,\infty)$, depending only on
the value of $d_1$, such that
$\inf\spec[|\DA|+d\Gamma(\vo)]\klg c(d_1)$.
\end{corollary}

\begin{proof}
By virtue of Corollary~\ref{cor-pia} 
we find some
$c\equiv c(d_1)$ such that
$
|\DA|+d\Gamma(\vo)\klg c\,(|\DO|+d\Gamma(\vo))
$.
Picking a minimizing sequence for the quadratic form
on the right hand side we conclude that
$\inf\spec[|\DA|+d\Gamma(\vo)]\klg c$.
\end{proof}

%%%%%%%%%%%%%%%%%%%%%%%%%%%%%%%%%%%%%%%%%%%%%%%%%%%%%%%%%%%%%%%%%%%%%%%%%%
%%%%%%%%%%%%%%%%%%%%%%%%%%%%%%%%%%%%%%%%%%%%%%%%%%%%%%%%%%%%%%%%%%%%%%%%%%
%%%%%%%%%%%%%%%%%%%%%%%%%%%%%%%%%%%%%%%%%%%%%%%%%%%%%%%%%%%%%%%%%%%%%%%%%%

\section{Existence of binding}
\label{sec-binding}

\noindent
As a first step towards the proof of the existence of ground
states we need to show that binding occurs in the atomic
system defined by $\PF{\gamma}\equiv\PF{\gamma,\V{G}^\phys,\omega}$. 
That is, we need to show that
$\inf\spec[\PF{\gamma}]+C_\gamma\klg\inf\spec[\PF{0}]$, for all
$\gamma\in(0,2/\pi)$, where $C_\gamma>0$.
This information will be exploited mathematically
when we apply a bound on the spatial localization 
of low-lying spectral subspaces of $\PF{\gamma}$
from \cite{MatteStockmeyer2009a}.
The localization estimate in turn enters into the
proof of the existence of ground states at various 
places, for instance, into the derivation of the
infra-red estimates.
We shall obtain a ground state of $\PF{\gamma}$ 
as a limit of ground states of 
infra-red cut-off Hamiltonians.
The existence of the latter ground states in turn is proved by
means of a discretization in the photonic degrees of freedom. 
Therefore, it is actually
necessary to have a bound on the constant $C_\gamma$
which is uniform in the infra-red cut-off and in the discretization
parameter. 

In the first two subsections below we introduce the infra-red cut-off
and discretized semi-relativistic Pauli-Fierz operators.
After that we introduce a fiber integral representation of these
operators with vanishing exterior potentials.
For the translation invariant non-discretized
operators this corresponds to the 
decomposition with respect to different values
of the total momentum operator.
This representation is a key ingredient in the proof of
the binding condition which is presented in the last of
the four subsequent subsections.

%%%%%%%%%%%%%%%%%%%%%%%%%%%%%%%%%%%%%%%%%%%%%%%%%%%%%%%%%%%%%%%%%%%%%%%%%%%%

\subsection{The infra-red cut-off operator $\PFm{\gamma}$}\label{ssec-PFm}
The infra-red cut-off Hamiltonians, 
$\PFm{\gamma}$, $m>0$, are given by
\begin{equation}\label{def-cAm}
\cA_m\,:=\,
\big\{\,\V{k}\in\RR^3\,\big|\:|\V{k}|\grg m\,\big\}\,,
\quad m>0\,,
\end{equation}
and
\begin{eqnarray}
\valpha\cdot\V{A}_m(\V{x})&:=&\label{def-Am}
\valpha\cdot\ad(\id_{\cA_m}\,e^{-i\V{k}\cdot\V{x}}\,\V{g})\,+\,
\valpha\cdot a(\id_{\cA_m}\,e^{-i\V{k}\cdot\V{x}}\,\V{g})\,,
\\ \label{def-Hm}
\PFm{\gamma}&:=&|\DAm|\,-\,\tgV\,+\,\Hf\,.
\end{eqnarray}
Here $e^{-i\V{k}\cdot\V{x}}\,\V{g}=\V{G}^\phys_\V{x}$
is the physical choice of the form factor 
defined by \eqref{def-Gphys} and $\Hf=d\Gamma(\omega)$.
From the remarks below \eqref{maria} we know that
$\PFm{\gamma}$ is well-defined as a self-adjoint
Friedrichs extension starting from $\core$.
To have a unified notation we further set
$$
\cA_0\,:=\,\RR^3\,,\qquad \V{A}_0\,:=\,\V{A}(\V{G}^\phys)\,,\qquad
\PFo{\gamma}\,:=\,\PF{\gamma,\V{G}^\phys,\omega}\,.
$$

%%%%%%%%%%%%%%%%%%%%%%%%%%%%%%%%%%%%%%%%%%%%%%%%%%%%%%%%%%%%%%%%%%%%%%%%%%

\subsection{The discretized operator $\PFme{\gamma}$}\label{ssec-NPe}
Next, we define a discretized version of $\PFm{\gamma}$.
It is considered as an operator acting in a subspace 
of the truncated Hilbert space
\begin{equation}\label{def-HRm}
\HR_m^>\,:=\,L^2(\RR^3_\V{x},\CC^4)\otimes\Fock[\HP_m^>]\,,
\qquad \HP_m^>\,:=\,L^2(\cA_m\times\ZZ_2)\,.
\end{equation}
On this Hilbert space we introduce a discretization 
in the photon momenta:
We decompose $\cA_m=\{\V{k}\in\RR^3:\,|\V{k}|\grg m\}$ into a disjoint
union of ``cubes'' with side length $\ve>0$,
$$
\cA_m\,=\,\bigcup_{\vnu\in(\ve\ZZ)^3}Q_m^\ve(\vnu)\,,
\qquad Q_m^{\ve}(\vnu)\,:=\,\big(\vnu+[-\ve/2\,,\,\ve/2)^3\big)
\cap\cA_m\,,\quad
\vnu\in(\ve\ZZ)^3\,.
$$
Of course, for every $\V{k}\in\cA_m$, we find a unique
vector, $\vnu_\ve(\V{k})\in(\ve\ZZ)^3$, such that
$\V{k}\in Q_m^\ve(\vnu_\ve(\V{k}))$.
In this way we obtain a map
\begin{equation}\label{def-nu}
\vnu_\ve\::\;\cA_m\times\ZZ_2\,\longrightarrow\,\RR^3\,,
\qquad
k=(\V{k},\lambda)\,\longmapsto\,\vnu_\ve(k)\,:=\,
\vnu_\ve(\V{k})\,.
\end{equation}
We define the $\ve$-average of a locally integrable
function, $f$, on $\cA_m\times\ZZ_2$ by
$$
f_\ve(k)\,:=\,
\frac{1}{|Q_m^\ve(\vnu(\V{k}))|}\int\limits_{Q_m^\ve(\vnu(\V{k}))}f(\V{p},\lambda)
\,d^3\V{p}\,,
\qquad k=(\V{k},\lambda)\in\cA_m\times\ZZ_2\,.
$$
Alternatively, we may write, for every $f\in\HP_m^>=L^2(\cA_m\times\ZZ_2)$,
\begin{equation}\label{def-Pve}
f_\ve\,=\,P_\ve\,f\,:=\,\sum_{{\vnu\in(\ve\ZZ)^3:\atop
Q^\ve_m(\vnu)\not=\varnothing}}
\SPn{\wt{\id}_{Q_m^\ve(\vnu)}}{f}\,\wt{\id}_{Q_m^\ve(\vnu)}\,,
\end{equation}
where $\wt{\id}_{Q_m^\ve(\vnu)}$ denotes the normalized characteristic
function of the set $Q_m^\ve(\vnu)$. (Thus 
$\wt{\id}_{Q_m^\ve(\vnu)}=\ve^{-3/2}\,\id_{Q_m^\ve(\vnu)}$, 
provided $|\vnu|>m+\sqrt{3}\,\ve/2$.)
Of course, $P_\ve$ is an orthogonal projection in 
$\HP_m^>$.
The discretized vector potential is now given as
\begin{equation}\label{def-Aeps}
\V{A}_{m,\ve}(\V{x})\,:=\,\ad(e^{-i\vnu_\ve\cdot\V{x}}\,\V{g}_{m,\ve})+
a(e^{-i\vnu_\ve\cdot\V{x}}\,\V{g}_{m,\ve})\,,\qquad
\V{g}_{m,\ve}\,=\,P_\ve\big[\id_{\cA_m}\,\V{g}\big]\,,
\end{equation}
The dispersion relation is discretized in a 
slightly different way, namely
$$
\omega_\ve(k)
\,:=\,\inf\big\{\,|\V{p}|\::\;\V{p}\in
Q_m^\ve(\vnu_\ve(\V{k}))
\,\big\}\,,\qquad 
k=(\V{k},\lambda)\in\cA_m\times\ZZ^2\,.
$$
For this definition of $\omega_\ve$ has the 
following trivial consequences
which shall be useful later on,
\begin{eqnarray}\label{omega-omegaeps1}
\max\big\{\,m\,,\,(1-\sqrt{3}\,\ve/m)\,\omega\,\big\}&\klg&
\omega_\ve\,\klg\,\omega\qquad\textrm{on}\;\;\cA_m\,,
\\ \label{omega-omegaeps2}
\He\,:=\,d\Gamma(\omega_\ve)&\klg&
\Hfmg\,:=\,d\Gamma(\omega\!\!\upharpoonright_{\cA_m\times\ZZ_2})\,.
\end{eqnarray}
Here the operators in the last line are acting in
$\Fock[\HP_m^>]$.
Given $\gamma\in(0,2/\pi)$
and some function $c:(0,1)\to(0,1)$, $\ve\mapsto c(\ve)$,
possibly $m$-dependent also, but tending to zero
uniformly in $m\grg0$, as 
$\ve\searrow0$, we write
\begin{equation}\label{def-gammaeps}
\gamma_\ve\,:=\,\gamma/(1-c(\ve))
\end{equation}
in what follows. (Later on we choose $c$ essentially equal to 
$m^{1/4}\,\triangle^{1/2}(\ve)$ where $\triangle(\ve)$
is defined in \eqref{def-treps}.)
We shall always restrict our attention to sufficiently small
values of $\ve$ satisfying $\gamma_\ve<2/\pi$.
For such $\gamma$ and $\ve$ and every $m>0$,
we define a discretized semi-relativistic Pauli-Fierz operator, 
$\PFme{\gamma}$, 
acting in 
$L^2(\RR^3,\CC^4)\otimes\Fock[\HP^>_m]$,
\begin{align}
\PFme{\gamma}\,&:=\,|\DAme|-\gamma_\ve/|\V{x}|+H_{f,m,\ve}
\,.\label{def-PFe}
\end{align}
Notice that, by definition, 
the above operators 
act in the truncated Hilbert space modeled by
means of $\HP_m^>$
and that the Coulomb coupling constant has been
changed to $\gamma_\ve$ in $\PFme{\gamma}$.
Once more, the remarks succeeding
\eqref{maria} apply to $\PFme{\gamma}$ which is thus 
well-defined as a Friedrichs extension starting from
the algebraic tensor product
\begin{equation}\label{def-core>}
\core^>\,:=\,C_0^\infty(\RR^3,\CC^4)\otimes\sC_0^>\,.
\end{equation}
Here $\sC_0^>\subset\Fock[\HP_m^>]$ is the dense subspace of all
$(\psi^{(n)})_{n=0}^\infty\in\Fock[\HP_m^>]$ such that
$\psi^{(n)}\not=0$, for only finitely many $n$, and such that each
$\psi^{(n)}$, $n>0$, is bounded and
has a compact support in $(\cA_m\times\ZZ_2)^n$.

%%%%%%%%%%%%%%%%%%%%%%%%%%%%%%%%%%%%%%%%%%%%%%%%%%%%%%%%%%%%%%%%%%%%%%%%%%%%%%

\subsection{Fiber decompositions of the free operators ($\gamma=0$)}\label{ssec-fiber}
Our bound on the binding energy for $\PFm{\gamma}$, $m\grg0$,
is based on a direct fiber decomposition of $\HR_4$
with respect to fixed values of the total
momentum $\pe\otimes\id+\id\otimes\pf$, where
\begin{equation}
  \label{eq:1} 
\pf\, := \,d \Gamma (\V{k})\,:=\,\big(d\Gamma(k^{(1)})\,,\,
d\Gamma(k^{(2)})\,,\,d\Gamma(k^{(3)})\big)\,,
\end{equation}
is the photon momentum operator.
A conjugation of the Dirac operator with
the unitary operator $e^{i\pf\cdot\V{x}}$
-- which is simply a multiplication
with the phase $e^{i(\V{k}_1+\dots+\V{k}_n)\cdot\V{x}}$
in each Fock space sector $\Fock^{(n)}[\HP]$ -- 
yields
$$
e^{i\pf\cdot\V{x}}\,\DAm\,e^{-i\pf\cdot\V{x}}
\,=\,\valpha\cdot(\V{p}-\pf+\V{A}_m(\V{0}))+\beta\,.
$$
When we deal with the discretized operator $\PFme{\gamma}$, $m,\ve>0$,
then we replace the photon momentum operator by 
$$
\pfe\,:=\,d\Gamma(\vnu_\ve)\qquad
(\,\textrm{all three components acting in}\,\Fock[\HP_m^>]\,)\,,
$$
where $\vnu_\ve$ is defined in \eqref{def-nu}.
Then it is again easy to check that
\begin{align*}
e^{i\pfe\cdot\V{x}}\,\DAe\,e^{-i\pfe\cdot\V{x}}
\,&=\,\valpha\cdot(\V{p}-\pfe+\V{A}_{m,\ve}(\V{0}))+\beta\,.
\end{align*}
We unify our notation by setting
\begin{equation}\label{clarissa}
\gamma_0\,:=\,\gamma\,,\quad
\V{p}_{m,0}\,:=\,\pf\,,\quad\V{A}_{m,0}\,:=\,\V{A}_m\,,\quad
H_{f,m,0}\,:=\,\Hf\,,
\end{equation}
and we always assume that $\ve=0$ when $m=0$ in the sequel.
Then a further conjugation with the Fourier transform, 
$\fourier:L^2(\RR^3_\V{x})\to L^2(\RR^3_{\vxi})$,
with respect to the variable 
$\V{x}$ 
turns the transformed Dirac operators into
\begin{equation}\label{tf-Dirac}
(\fourier\otimes\id)\,e^{i\pfe\cdot\V{x}}\,\DAme\,
e^{-i\pfe\cdot\V{x}}\,(\fourier^{-1}\otimes\id)
\,=\,\int_{\RR^3}^\oplus \wh{D}_{m,\ve}(\vxi)\,d^3\vxi\,.
\end{equation}
Here the operators
$$
\FD(\vxi)\,:=\,\valpha\cdot(\vxi-\pfe+\V{A}_{m,\ve}(\V{0}))+\beta\,,
\qquad \vxi\in\RR^3\,,
$$
acting in $\CC^4\otimes\Fock[\HP]$, for $\ve=0$, and
in $\CC^4\otimes\Fock[\HP_m^>]$, for $\ve>0$,
are fiber Hamiltonians of
the transformed Dirac operator in \eqref{tf-Dirac}
with respect to the isomorphisms
\begin{equation}
  \label{eq:2} 
\HR\,\cong\,
\int_{\RR^3}^\oplus\CC^4\otimes\Fock[\HP]\,d^3\vxi\,,
\qquad \HR_m^>\,\cong\,
\int_{\RR^3}^\oplus\CC^4\otimes\Fock[\HP_m^>]\,d^3\vxi
\,,
\end{equation}
respectively.
(In particular, the transformed Dirac operator in \eqref{tf-Dirac}
again acts in $\HR_4$ or $\HR_m^>$, respectively,
where, for $\ve=0$, the variable in the first tensor factor, $\vxi$,
is now interpreted as the total momentum of the
combined electron-photon system.)
Corresponding to \eqref{eq:2} we 
then have the direct integral representation
(compare, e.g., \cite[Theorem~XIII.85]{ReedSimonIV})
\begin{equation}
  \label{eq:3} 
(\fourier\otimes\id)\,e^{i\pfe\cdot\V{x}}\,\PFme{0}\,
e^{-i\pfe\cdot\V{x}}\,(\fourier^{-1}\otimes\id)
\,=\,
\int_{\RR^3}^\oplus\PFme{0}(\vxi)
\,d^3\vxi \,,
\end{equation}
where
$$
\PFme{0}(\vxi)\,:=\,|\FD(\vxi)|+H_{f,m,\ve}\,.
$$

%%%%%%%%%%%%%%%%%%%%%%%%%%%%%%%%%%%%%%%%%%%%%%%%%%%%%%%%%%%%%%%%%%%%%%%

\subsection{Proof of the binding condition}\label{ssec-binding-proof}

For $m\grg0$ and $\ve>0$ (in case that $m>0$), we set
\begin{equation}\label{clarissa2}
E_{\gamma,m,\ve}\,:=\,\inf\spec[\PFme{\gamma}]\,,\quad \gamma_\ve\in(0,2/\pi)\,,
\qquad \Th_{m,\ve}\,:=\,\inf\spec[\PFme{0}]\,,
\end{equation}
and we fix some $\rho>0$ in what follows.
In view of the fiber decomposition \eqref{eq:3}
we know that the Lebesgue measure
of the set of all $\vxi\in\RR^3$ satisfying
$\spec[\PFme{0}(\vxi)]\cap(\Th_{m,\ve}-\rho,\Th_{m,\ve}+\rho)\not=\varnothing$
is strictly positive \cite[Theorem~XIII.85]{ReedSimonIV}. In particular, 
we find some $\ps\in\RR^3$ and some
normalized 
$\vp_\star\in\form(\PFme{0}(\ps))$
such that
\begin{align}
  \label{eq:6}
    &\SPb{\vp_{\star}}{\PFme{0}(\ps) \, 
    \vp_{\star}}_{\CC^4\otimes\Fock} 
    \,<\, \Th_{m,\ve} + \rho\,.
\end{align}
We define the unitary transformation
\begin{equation*}
 %\label{eq:7} 
U\,\equiv\,U_{m,\ve}\, : =\, e^{i ( \pfe - \ps) \cdot \mathbf{x}}
\end{equation*}
and observe as above that
\begin{equation*}
  %\label{eq:8}
\FD(\pe,\ps)\,:=\,  
U\,\DAme\,U^*\,=\,\valpha\cdot(\pe+\bts)+\beta\,,
\end{equation*}
where
$$
\bts\,:=\,\ps - \pfe + \V{A}_{m,\ve}(\V{0})\,.
$$
It suffices to prove the binding condition for the 
unitarily equivalent operator
\begin{equation}
  \label{eq:9} U \PFme{\gamma} U^{\ast}\,=\,
\sqrt{\big( \valpha \cdot ( \pe + \bts) \big)^2 + 1} 
\,-\,\frac{\gamma_\ve}{| \mathbf{x} |}\,+\,H_{f,m,\ve}\,.
\end{equation}

\begin{theorem}\label{thm-binding}
For all $e^2,\UV>0$, $\gamma\in(0,2/\pi)$, $m\grg0$,
and $\ve>0$ (provided that $m>0$),
\begin{equation}
    \label{eq:18} \Th_{m,\ve} - E_{\gamma,m,\ve}  
\,\grg\,|E^\el_{\mathrm{nr},\gamma_\ve}|\,,
  \end{equation}
where 
$E^\el_{\mathrm{nr},\gamma_\ve}
=\inf\spec\big[\frac{1}{2}\,\pe^2-\frac{\gamma_\ve}{|\V{x}|}\big]
=-\gamma^2_\ve/2$.
\end{theorem}

\begin{proof}
 Let $\rho > 0$ and $\ps$ be as in the paragraphs
preceding the statement. For $\eta \grg 0$, we abbreviate
\begin{align*}  
R_1({\eta}):=
\big((\valpha\cdot\bts)^2+{\eta}+1\big)^{-1},\qquad
R_2({\eta}):=\big((\valpha\cdot(\pe+\bts))^2+{\eta}+1\big)^{-1}.
\end{align*}  
Since the anti-commutator of $\valpha\cdot\pe$ and 
$\valpha\cdot\bts$ is equal to $2\,\pe\cdot\bts$
it holds 
$(\valpha\cdot(\pe+\bts))^2=(\valpha\cdot\bts)^2+2\,\pe\cdot\bts+\pe^2$
on $\dom(\pe^2)\cap\dom(\He^2)$.
In Lemma~\ref{le-aischa} below we verify that 
$R_1(\eta)$ maps $\dom(\He^\nu)$ into itself, for every $\nu>0$. 
We deduce that, for any $\varphi \in \core$ (respectively
$\vp\in\core^>$), 
\begin{align}
    -R_2({\eta})\,\vp \nonumber
\,&=\,
-R_2({\eta})\,[\,(\al\cdot\bts)^2
+1+{\eta}\,]\,R_1({\eta})\,\vp
\\ \label{eq:10} 
&=\,
R_2({\eta})\,[\,2\,\pe\cdot\bts
+{\pe}^2\,]\,R_1({\eta})\,\vp\,-\,R_1({\eta})\,\vp\,.
\end{align}
We use the following formula, valid for any self-adjoint operator $T > 0$,
\begin{align*}  
{\sqrt{T}}\,\psi\,=\,
\int_0^{{\infty}}\Big(1-{\frac{{\eta}}{T+{\eta}}}\Big)
\,\psi\,{\frac{d{\eta}}{{\pi}{\sqrt{{\eta}}}}}\:,
\qquad \psi\,{\in}\,{\dom}(T)\,,
\end{align*}
and the resolvent identity \eqref{eq:10}
to obtain
\begin{align}
      \SPb{&\vp}{\big(\sqrt{( \al \cdot(\pe+\nonumber
      \bts))^2 + 1}\,-\, \sqrt{( \al \cdot \bts)^2 + 1} \,\big)\, \vp}
      \\
      &= \,\nonumber
      \int_0^{\infty} \SPb{\vp}{\big( R_1 (\eta) - R_2 (\eta) \big)
      \, \vp} \sqrt{\eta} \: \frac{d \eta}{\pi}\\
      &= \,\nonumber
      \int_0^{\infty} \SPb{R_2 (\eta)\,\vp}{\left[ \,2\,\pe\cdot\bts 
      + \pe^2 \right] R_1 (\eta)
      \,\vp} \sqrt{\eta} \: \frac{d \eta}{\pi}\\
      &= \,\nonumber
      \int_0^{\infty} \SPb{\vp}{R_1 (\eta) \left[ \,2\,\pe\cdot\bts
      + \pe^2 \right] R_1 (\eta)
      \, \vp} \sqrt{\eta} \: \frac{d \eta}{\pi}
\\
&\;\;\nonumber
- \int_0^{\infty} \SPb{\left[ \,2\,\pe\cdot\bts
        + \pe^2 \right] R_1 (\eta)\,\vp}{
          R_2 (\eta)\left[ \,2\,\pe\cdot\bts+
        \pe^2 \right] R_1 (\eta) \, \vp
      } \sqrt{\eta} \: \frac{d \eta}{\pi}
\\
&\klg\,\label{eq:7000}
      \int_0^{\infty} \SPb{\vp}{R_1 (\eta) \left[\,2\,\pe\cdot\bts 
      + \pe^2 \right] R_1 (\eta)
      \, \vp} \sqrt{\eta} \, \frac{d \eta}{\pi}
      \,.
\end{align}
In the last step we used the positivity of $R_2(\eta)$.
  We consider now 
$\vp : = \vp_1 \otimes \vp_2$ where 
$\vp_1 \in C_0^{\infty}( \RR^3, \RR)$ 
and $\vp_2 \in \CC^4 \otimes \sC_0$
(respectively $\vp_2\in\CC^4\otimes\sC_0^>$) 
with $\| \vp_j \|= 1$, $j=1,2$. 
(Here the dense subspaces $\sC_0\subset\Fock[\HP]$ 
and $\sC_0^>\subset\Fock[\HP_m^>]$ are defined below 
\eqref{def-core2-sC} and \eqref{def-core>}.)
We find that
\begin{align}\nonumber
\SPb{\vp}{&R_1 (\eta) \, \pe\cdot\bts\, R_1 (\eta)\, \vp}  
\\
& = \,\label{eq:14}
   \sum_{j=1}^3 \SPb{\vp_1}{-i\partial_{x_j}\, \vp_1}
      \SPb{\vp_2}{R_1 (\eta)\, \bts^{(j)}\, R_1 (\eta)\,
      \vp_2}\,=\,0\,,
\end{align}
  due to the fact that $\vp_1$ is real and, hence,
  $2\SPb{\vp_1}{\partial_{x_j} \vp_1}=\int\partial_{x_j}\vp_1^2 = 0$. 
On the other hand the functional
  calculus implies
 \begin{align*}
   \nonumber
      \int_0^{\infty} \SPb{\vp_2}{R_1 (\eta)^2 \vp_2} \sqrt{\eta}
      \: \frac{d \eta}{\pi} &\, =\, \int_0^{\infty}
      \frac{\sqrt{\eta}}{(1 + \eta)^2} \, \frac{d \eta}{\pi}\:
      \SPb{\vp_2}{\left( ( \al \cdot \bts)^2 + 1 \right)^{- 1 / 2}
      \, \vp_2}
\\
      & = \,  \frac{1}{2}\,\SPb{\vp_2}{\big( ( \al \cdot \bts)^2 + 1
      \big)^{- 1 / 2} \, \vp_2}\,\klg\,\frac{1}{2} \,,
\end{align*}
which permits to get
  \begin{equation}
    \label{eq:15} \int_0^{\infty} \SPb{\vp}{R_1 (\eta)\,\pe^2\,R_1(\eta)\,\vp}
    \sqrt{\eta} \: \frac{d \eta}{\pi} \, \klg
    \,  \frac{1}{2}\,\SPb{\vp_1}{\pe^2 \, \vp_1} \,.
  \end{equation}
Combining \eqref{eq:7000}, \eqref{eq:14}, and \eqref{eq:15} we arrive at
  \begin{align*}
    \SPb{\vp}{&U \PFme{\gamma} U^{\ast} \, \vp}
   \\
    &\klg \, \SPb{\vp_2}{\big( \sqrt{( \al \cdot
    \bts)^2 + 1} + H_{f,m,\ve} \big)\, \vp_2}\, +\, 
    \SPb{\vp_1}{\big( \tfrac{1}{2}\,\pe^2 -\tfrac{\gamma_\ve}{|\V{x}|}\big) 
    \, \vp_1}
    \, ,
  \end{align*}
where 
$\sqrt{( \valpha \cdot \bts)^2 + 1} + H_{f,m,\ve} = \PF{0,m,\ve}(\ps)$.
  By a limiting argument the previous inequality extends to any 
  real-valued $\vp_1\in \form ( \pe^2)$ and 
  $\vp_2 \in \form ( \PFme{0} ( \ps))$. 
  We choose $\vp_1$ to be
%= (\gamma^3/8\pi)^{1/2}\,e^{-\gamma|\V{x}|/2}$, which is 
  the normalized, strictly positive 
  eigenfunction of 
  $\frac{1}{2}\,\pe^2 - \frac{\gamma_\ve}{| \mathbf{x} |}$
  corresponding to its lowest eigenvalue $-\gamma^2_\ve/2$, 
  and $\vp_2 =\vp_{\star}$. 
  By the choice of $\vp_\star$ in \eqref{eq:6}, where
  $\rho>0$ is arbitrary, this proves the assertion.
\end{proof}

\begin{lemma}\label{le-aischa}
Let $\nu,\eta>0$. Then the resolvent $R_1(\eta)$
defined in the previous proof maps $\dom(H_{f,m,\ve}^\nu)$
into itself.
\end{lemma}

\begin{proof}
In view of the representation
$$
R_1(\eta)\,=\,\frac{1}{2i\sqrt{\eta}}\big\{\,
\big(\FD(\ps)-i\sqrt{\eta}\big)^{-1}\,
-\,\big(\FD(\ps)+i\sqrt{\eta}\big)^{-1}
\,\big\}
$$
it suffices to show that $\wh{R}(y):=(\FD(\pp)-iy)^{-1}$
maps $\dom(\He^\nu)$ into itself, for all $y\in\RR$ and $\pp\in\RR^3$.
Now, an application of Nelson's commutator theorem
with test operator $\HT:=\He+E$, $E\grg1$, 
shows that $\FD(\pp)$
is essentially self-adjoint on $\CC^4\otimes\sC_0$.
This property together with Lemma~\ref{cor-T-Xi} permits
to derive the identity
$[\wh{R}(y)\,,\,\HT^{-\nu}]=\wh{R}(y)\,T_\nu\,\HT^{-\nu}\,\wh{R}(y)$,
where $T_\nu\in\LO(\CC^4\otimes\Fock[\HP])$ is the closure
of $[\HT^{-\nu},\valpha\cdot\V{A}_{m,\ve}(\V{0})]\,\HT^{\nu}$;
compare \cite[Proof of Corollary~3.2]{MatteStockmeyer2009a}.
This identity in turn implies 
\begin{equation}\label{aischa}
\wh{R}(y)\,\HT^{-\nu}\,=\,\HT^{-\nu}\,\wh{R}(y)\,\wh{\Upsilon}_\nu(y)
\end{equation}
on $\CC^4\otimes\Fock[\HP]$, 
where the Neumann series
$\wh{\Upsilon}_\nu(y):=\sum_{j=0}^\infty\{-T_\nu^*\,\wh{R}(y)\}^j$
converges, provided $E\grg1$ is chosen sufficiently large.
The operator identity \eqref{aischa} shows that
$\wh{R}(y)$ maps $\dom(\He^\nu)=\Ran(\HT^{-\nu})$
into itself. (For $\ve>0$, we have to replace $\HP$ 
by $\HP_m^>$ and $\sC_0$ by $\sC_0^>$
in the argument above.)
\end{proof}

%%%%%%%%%%%%%%%%%%%%%%%%%%%%%%%%%%%%%%%%%%%%%%%%%%%%%%%%%%%%%%%%%%%%%
%%%%%%%%%%%%%%%%%%%%%%%%%%%%%%%%%%%%%%%%%%%%%%%%%%%%%%%%%%%%%%%%%%%%%
%%%%%%%%%%%%%%%%%%%%%%%%%%%%%%%%%%%%%%%%%%%%%%%%%%%%%%%%%%%%%%%%%%%%%

\section{Existence of ground states}
\label{sec-ex}

\subsection{Outline of the proof}\label{ssec-ex-outline}

\noindent
In this section we prove our main Theorem~\ref{thm-ex-PF}.
As in \cite{BFS1998b,BFS1999} 
(see also \cite{GLL2001} where a photon mass is introduced
in a slightly different way)  
we first show that the infra-red cut-off
Hamiltonians, $\PFm{\gamma}$, $m>0$, 
defined in \eqref{def-cAm}--\eqref{def-Hm}
possess ground state eigenfunctions,
provided $m>0$ is sufficiently small.
This is done in Subsection~\ref{ssec-ex-m-PF}
by means of a discretization argument
similar to the one in \cite{BFS1999}.
The implementation of the discretization procedure in
\cite{BFS1998b,BFS1999} requires a small coupling
condition. By a modification of the argument
we observe, however, that this is actually not necessary.
Before we turn to these issues
we explain in Subsection~\ref{ssec-approx-m}
how to infer the existence of ground states for
the limit operator $\PF{\gamma}$ from the fact that
the $\PFm{\gamma}$ have ground state eigenfunctions.
Here we benefit from a result from \cite{GLL2001}
saying that the spatial localization,
a bound on the number of soft photons,
and a photon derivative bound introduced in \cite{GLL2001}
allow to use standard imbedding theorems for
Sobolev spaces to ensure the compactness of the
set of approximating ground states.
The first of the latter key ingredients, the
exponential localization estimate for
low-lying spectral subspaces of the operators
$\PFm{\gamma}$, has been proven
in \cite{MatteStockmeyer2009a}.
The proofs of the two infra-red bounds 
are postponed to Section~\ref{sec-IR-bounds}.

We close this subsection by a general lemma
which allows to prove the existence of imbedded eigenvalues
by means of approximating sequences of operators and
eigenvectors.
It is a modified version of a result we learned from \cite{BFS1998b}
and its assertion is actually stronger than necessary
for our application.

\begin{lemma}\label{le-ex-allg}
Let $T,T_1,T_2,\ldots$ be self-adjoint operators acting in some
separable Hilbert space, $\mathscr{X}$, such that $\{T_j\}_{j\in\NN}$
converges to $T$ in the strong resolvent sense.
Assume that $E_j$ is an eigenvalue of $T_j$ with corresponding
eigenvector $\phi_j\in\dom(T_j)$.
Assume further that $\{\phi_j\}_{j\in\NN}$ converges weakly to
some $0\not=\phi\in\mathscr{X}$. Then $E:=\lim_{j\to\infty}E_j$
exists and is an eigenvalue of $T$. If $E_j=\inf\spec[T_j]$,
then $T$ is semi-bounded below and $E=\inf\spec[T]$.
\end{lemma}

\begin{proof}
In what follows we abbreviate $f:=\arctan$.
Then $f(T_j)\,\psi\to f(T)\,\psi$, $j\to\infty$,
for every $\psi\in\mathscr{X}$,
since $T_j\to T$ 
in the strong resolvent sense.
Let us assume for the moment that $\psi\in\mathscr{X}$ fulfills
$\SPn{\psi}{\phi}\not=0$. Then we find some $j_0\in\NN$
such that $\SPn{\psi}{\phi_{j}}\not=0$, for $j\grg j_0$,
and we may write
$$
f(E_{j})\,=\,
\frac{\SPn{f(T)\,\psi}{\phi_{j}}\,
+\,\SPn{f(T_j)\,\psi-f(T)\,\psi}{\phi_{j}}}{
\SPn{\psi}{\phi_{j}}}\,,\qquad j\grg j_0\,.
$$
The sequence $\{f(E_{j})\}_{j\in\NN}$ thus has a limit
$$
f(E)\,:=\,
\lim_{j\to\infty}f(E_{j})\,=\,
\frac{\SPn{\psi}{f(T)\,\phi}}{\SPn{\psi}{\phi}}
\,.
$$
In the case $\SPn{\psi}{\phi}=0$ we may replace $\psi$ by
$\wt{\psi}=\psi+\phi$ in the above argument
to see that $\SPn{f(T)\,\psi}{\phi}=0$
also.
The equality $\SPn{\psi}{f(T)\,\phi}=\SPn{\psi}{f(E)\,\phi}$ thus
holds, for every $\psi\in\mathscr{X}$, 
whence $f(T)\,\phi\,=f(E)\,\phi$.
It follows that $E:=\tan(f(E))\grg\inf\spec[T]$ is an eigenvalue
of $T$ since 
$u(\specpp[f(T)])\subset\specpp[u(f(T))]=\specpp[(u\circ f)(T))]$, for 
every Borel measurable function $u$. 

Now assume that $E_j=\inf\spec[T_j]$.
Set $\lambda:=\inf\spec[T]$, when $T$ is semi-bounded below,
and pick some $\lambda\in\spec[T]\cap(-\infty,E-1)$, when
$\inf\spec[T]=-\infty$.
Since $T_j$ converges to $T$ 
in the strong resolvent sense
there is a sequence $\{E_j'\}_{j\in\NN}$ with 
$E_{j}\klg E_j'\in\spec[T_j]$ 
and $E_j'\to\lambda$.
So $E\grg\lambda=\lim_{j\to\infty}E_j'\grg\lim_{j\to\infty}E_j=E$.
If $\inf\spec[T]=-\infty$ the first inequality is strict
and we get a contradiction.
\end{proof}

%%%%%%%%%%%%%%%%%%%%%%%%%%%%%%%%%%%%%%%%%%%%%%%%%%%%%%%%%%%%%%%%%%%%%%%%%%%

\subsection{Approximation by infra-red cut-off electromagnetic fields}
\label{ssec-approx-m}

\noindent
In order to prove that $\PF{\gamma}$ has a ground state
provided this holds true for $\PFm{\gamma}$
with sufficiently small $m>0$
we first show that $\PFm{\gamma}$ converges to $\PF{\gamma}$ in norm
resolvent sense. 
If we choose $a=0$, $\vo=\omega$, 
$\V{G}=\V{G}^\phys=e^{-i\V{k}\cdot\V{x}}\,\V{g}$, 
and $\wt{\V{G}}=\id_{\cA_m}\,\V{G}^\phys$, $m>0$, then
the parameter defined below \eqref{delta(a)3} is equal to
\begin{equation}\label{def-Deltam}
\wh{\triangle}^2_{*}(m)\,:=\,
\int_{\{|\V{k}|\klg m\}}\Big(2+\frac{4}{\omega(k)}\Big)\,
|\V{g}(k)|^2\,dk\,\longrightarrow\,0\,,\qquad m\searrow0\,.
\end{equation}

\begin{lemma}\label{le-sres-conv-PF}
Let $e^2,\UV,m>0$ and $\gamma\in[0,2/\pi)$. Then 
$\PFm{\gamma}$ and $\PF{\gamma}$ have the same form domain,
$\form(\PFm{\gamma})=\form(\PF{\gamma})=\form(|\DO|)\cap\form(\Hf)$,
and the form norms associated to $\PFm{\gamma}$ and
$\PF{\gamma}$ are equivalent. Moreover,
$\PFm{\gamma}$ converges to
$\PF{\gamma}$ in the norm resolvent sense, as $m\searrow0$.
\end{lemma}

\begin{proof}
The first assertion, which 
has already been observed in \cite{Stockmeyer2009},
follows from Corollary~\ref{cor-pia} 
(with $\V{A}=\V{0}$ or $\wt{\V{A}}=\V{0}$).
Moreover,
we know that $\core$ is a common form core of
$\PF{\gamma}$ and $\PFm{\gamma}$, $m>0$, and on $\core$
we have 
$$
\PF{\gamma}-\PFm{\gamma}\,=\,
\big(\SA-\SAm\big)\,\DA\,
+\,\SAm\,\valpha\cdot(\V{A}-\V{A}_m)\,.
$$
By virtue of Lemma~\ref{le-A-Aeps} and \eqref{delta(a)3}
we thus find some $C\in(0,\infty)$ such that, for all $m>0$
and $\vp\in\core$,
\begin{align}
\big|\SPb{\vp}{(\PF{\gamma}-\PFm{\gamma})\,\vp}\big|\,
&\klg\,\nonumber
\bigO(\wh{\triangle}_*(m))\,\big(\,\|\HT^{1/2}\,\vp\|\,
\|\,|\DA|^{1/2}\,\vp\|+\|\vp\|\,\|\HT^{1/2}\,\vp\|\,\big)
\\
&\klg\,\label{verona2}
\bigO(\wh{\triangle}_*(m))\,\SPb{\vp}{(\PF{\gamma}+C)\,\vp}\,.
\end{align}
Here $\HT=\Hf+E$, for some sufficiently large $E>0$, and
in the last step we used Corollary~\ref{cor-pia}. 
Since we may replace $\vp$ in \eqref{verona} by any element
of $\form(\PFm{\gamma})=\form(\PF{\gamma})$
the second assertion follows from \cite[Theorem~VIII.25]{ReedSimonI}.
\end{proof}

\smallskip

\noindent
For every $m>0$, we split the one-photon Hilbert space
into two mutually orthogonal subspaces
$$
\HP\,=\,L^2(\RR^3\times\ZZ_2)\,=\,\HP_m^>\oplus\HP_m^<\,,
$$
where
$\HP_m^{>}=L^2(\cA_m\times\ZZ_2)$ has already been introduced
in \eqref{def-HRm}.
It is well-known that $\Fock[\HP]=\Fock[\HP_m^>]\otimes\Fock[\HP_m^<]$
and we observe that
all the operators $a(\id_{\cA_m}\,e^{-i\V{k}\cdot\V{x}}\,g^{(j)}(k))$, 
$\ad(\id_{\cA_m}\,e^{-i\V{k}\cdot\V{x}}\,g^{(j)}(k))$,
$j\in\{1,2,3\}$, $\V{x}\in\RR^3$, and $\Hf$
leave the Fock space factors associated
to the subspaces $\HP_m^{\lessgtr}$ invariant.
Hence, the same holds true also for $\DAm$, $\SAm$, and $|\DAm|$.
We shall designate operators acting in the Fock space
factors $\Fock[\HP_m^{>}]$ or $\Fock[\HP_m^{<}]$ by a superscript
$>$ or $<$, respectively. 
Then we have $\DAm\cong\DAmg\otimes\id$ and $\SAm\cong\SAmg\otimes\id$ 
under the isomorphism
\begin{equation}\label{isom-m<>}
\HR_4\,\cong\,\big(L^2(\RR^3,\CC^4)\otimes
\Fock[\HP_m^>]\big)\otimes\Fock[\HP_m^<]\,.
\end{equation}
The semi-relativistic Pauli-Fierz operator
decomposes under the isomorphism \eqref{isom-m<>}
as
\begin{equation*}
\PFm{\gamma}\,=\,\ol{\PFmg{\gamma}\otimes\id\,+\,\id\otimes H_{f}^<}\,,
\qquad
\PFmg{\gamma}\,:=\,|\DAmg|-\tgV+ H_{f}^>\,.
\end{equation*}
By the remarks below \eqref{maria} the operator $\PFmg{\gamma}$
is well-defined as a Friedrichs extension starting from
the the dense subspace $\core^>$ defined in
\eqref{def-core>}.
We let $\Omega^<$ denote the vacuum
in $\Fock[\HP^<_m]$.
In view of $H_{f}^<\,\Omega^<=0$ we then observe that
\begin{equation}\label{vera1}
E_m\,:=\,
E_{\gamma,m}\,=\,\inf\spec[\PFm{\gamma}]\,=\,\inf\spec[\PFmg{\gamma}]\,,
\qquad m>0\,.
\end{equation}
Thus, if $\phi^>_m$ is a ground state eigenvector of
$\PFmg{\gamma}$ then $\phi^>_m\otimes\Omega^<$
is a ground state eigenvector of $\PFm{\gamma}$.
Likewise, we have
\begin{equation}\label{vera2}
\Th_m\,=\,\inf\spec[\PFm{0}]\,=\,\inf\spec[\PFmg{0}]\,,
\qquad m>0\,,
\end{equation}
since we can tensor-multiply minimizing sequences
for $\PFmg{0}$ with $\Omega^<$.
In Subsection~\ref{ssec-ex-m-PF} we prove the following
proposition, where
$[\,\cdot\,]_-:\RR\to(-\infty,0]$ denotes the negative part
$$
[\,t\,]_-\,:=\,\min\{\,t\,,\,0\,\}\,,\qquad t\in\RR\,.
$$

\begin{proposition}[{\bf Ground states with mass}]\label{prop-gs-PFm}
Let $e^2,\UV>0$, $\gamma\in(0,2/\pi)$.
Then there exists some $m_0>0$ such that, for every $m\in(0,m_0)$,
the operator $[\PFmg{\gamma}-E_m-\frac{m}{4}]_-$ has finite rank.
In particular, $E_m$ is an eigenvalue of both $\PFmg{\gamma}$ and
$\PFm{\gamma}$.
\end{proposition}

\smallskip

\noindent
To benefit from this proposition we also
need the following results. The first one on the exponential
localization of low-lying spectral subspaces is also
applied to the discretized operator $\PFme{\gamma}$
defined in \eqref{def-gammaeps} and \eqref{def-PFe}
later on. (Recall our convention \eqref{clarissa} and \eqref{clarissa2}.)

\begin{proposition}[{\bf Exponential localization}]\label{prop-exp-loc}
There exist $k_0,k_1\in(0,\infty)$
such that the following holds:
Let $e^2,\UV>0$, $m,\ve\grg0$, $\gamma\in(0,2/\pi)$,
and let $I\subset(-\infty,\Th_{m,\ve})$ be some compact interval.  
Pick some
$a\in(0,1)$ such that
$\vr:=\Th_{m,\ve}-\max I-6a^2/(1-a^2)>0$ and $\vr\klg1$.
Then
\begin{equation}\label{exp-loc-EI}
\big\|\,e^{a|\V{x}|}\,\id_I(\PFme{\gamma})\,\big\|
\,\klg\,
(k_1/\vr^2)(1+|I|)(\Th_{m,\ve}+k_0\,e^2\,\UV^3)\,
e^{c(\gamma_\ve)\,a\,(\Th_{m,\ve}+k_0e^2\UV^3)/\vr}\,,
\end{equation}
where $|I|$ denotes the length of $I$ and 
$c:(0,2/\pi)\to(0,\infty)$ is some universal increasing function.
In particular, we find $\ve_1,m_1,\delta_1>0$, and $a_1\in(0,1)$ such that
\begin{equation}\label{exp-loc-sup}
\sup\big\{\,\|e^{a_1|\V{x}|}\,\id_{J_1(m,\ve)}(\PFme{\gamma}))\|\,:\:
m\in[0,m_1],\,\ve\in[0,\ve_1]\,\big\}
<\infty\,,
\end{equation}
where $J_1(m,\ve):=[E_{\gamma,m,\ve}\,,\,E_{\gamma,m,\ve}+\delta_1]$.
The same estimates \eqref{exp-loc-EI} and \eqref{exp-loc-sup}
hold true with $\PFme{\gamma}$ replaced by $\PFmg{\gamma}$,
for $m\in(0,m_1]$ and $\ve=0$.
\end{proposition}

\begin{proof}
The bound \eqref{exp-loc-EI} with $k_0\,e^2\,\UV^3$ replaced
by some constant times $d_1^2$
is stated in \cite[Theorem~2.5]{MatteStockmeyer2009a}
for dispersion relations $\vo$ and form factors $\V{G}$
fulfilling Hypothesis~\ref{hyp-G}.
In the cases $\vo=\omega$, $\V{G}=e^{-i\V{k}\cdot\V{x}}\,\V{g}\,\id_{\cA_m}$,
or $\vo=\omega_\ve$, $\V{G}=e^{-i\vnu_\ve\cdot\V{x}}\,\V{g}_{m,\ve}$
we can clearly choose
$d_1^2=\const\,e^2\,\UV^3$ uniformly in $m\grg0$ and $\ve\in[0,\ve_1]$,
for some $\ve_1>0$.
These
remarks apply to both $\PFme{\gamma}$ and $\PFmg{\gamma}$ and
on account of \eqref{vera1} and \eqref{vera2} we obtain
the same right hand side in \eqref{exp-loc-EI} when $\ve=0$.
To prove \eqref{exp-loc-sup} we pick
some $0<\delta_1<\gamma\klg\gamma_\ve$, choose $I=J_1(m,\ve)$,
and observe that, by Theorem~\ref{thm-binding},  
$\vr\grg \gamma^2/4-\delta_1-6\,a_1^2/(1-a_1^2)\grg\const>0$, 
uniformly in $m\grg0$ and $\ve\in[0,\ve_1]$,
provided that $a_1\in(0,1)$ is sufficiently small.
%Moreover, we know from Lemma~\ref{le-sres-conv-PF} 
%that $\Th_{m}\to\Th_0$, as $m\searrow0$. In particular, we find
%some $m_1>0$ such that the function
%$[0,m_1]\ni m\mapsto\Th_m$ is bounded.
Finally, we know from Corollary~\ref{cor-bd-on-Th}
that all threshold energies $\Th_{m,\ve}$, $m\grg0$, $\ve\in[0,\ve_1]$,
are bounded from above by some constant 
that depends only on the value of $d_1$.
Since $d_1$ has been chosen uniformly in $m\grg0$
and $\ve\in[0,\ve_1]$ this concludes the proof
of this proposition. 
\end{proof}

\smallskip

\noindent
The following proposition is proved in Section~\ref{sec-IR-bounds}.

\begin{proposition}[{\bf Soft photon bound}]\label{prop-spb}
Let $e^2,\UV>0$ and $\gamma\in(0,2/\pi)$.
Then there exist constants, $m_2,C\in(0,\infty)$,
such that, for all $m\in(0,m_2]$ and every normalized
ground state eigenfunction, $\phi_m$, of $\PFm{\gamma}$,
we have
\begin{equation}\label{eq-spb}
\big\|\,a(k)\,\phi_m\,\big\|^2\,\klg\,\id_{\{m\klg|\V{k}|\klg\UV\}}
\:\frac{C}{|\V{k}|}\,,
\end{equation}
for almost every $k=(\V{k},\lambda)\in\RR^3\times\ZZ_2$.
\end{proposition}

\smallskip

\noindent
The next proposition, which is also
proved in Section~\ref{sec-IR-bounds}, is the only place in the whole article
where the special choice of the polarization vectors \eqref{pol-vec}
enters into the analysis.

\begin{proposition}[{\bf Photon derivative bound}]\label{prop-pdb}
Let $e^2,\UV>0$, $\gamma\in(0,2/\pi)$.
Then there exist constants, $m_3,C\in(0,\infty)$,
such that, for all $m\in(0,m_3]$ and every normalized
ground state eigenfunction, $\phi_m$, of $\PFm{\gamma}$,
we have
\begin{equation}\label{pdb-kp}
\big\|\,a(k)\,\phi_m-a(p)\,\phi_m\,\big\|
\,\klg\,
C\,|\V{k}-\V{p}|\,
\Big(\frac{1}{|\V{k}|^{1/2}|\V{k}_\bot|}\,+\,
\frac{1}{|\V{p}|^{1/2}|\V{p}_\bot|}\Big)\,,
\end{equation}
for almost every $k=(\V{k},\lambda),p=(\V{p},\mu)\in\RR^3\times\ZZ_2$
with $m<|\V{k}|<\UV$ and $m<|\V{p}|<\UV$.
(Here we use the notation introduced in \eqref{def-kbot}.)
\end{proposition}

\smallskip

\noindent
We have now collected all prerequisites to show that
$\inf\spec[\PF{\gamma}]$ is an eigenvalue of $\PF{\gamma}$.

\bigskip

\begin{proof}[Proof of Theorem~\ref{thm-ex-PF}
by means of Propositions~\ref{prop-gs-PFm}--\ref{prop-pdb}]
Let $\phi_m$ denote a normalized
ground state of $\PFm{\gamma}$, for $m\in(0,m_\star]$,
where $m_\star>0$ is the minimum of the constants
$m_0,m_1,m_2,m_3$ appearing in Propositions~\ref{prop-gs-PFm}--\ref{prop-pdb}.
Then the family $\{\phi_m\}_{m\in(0,m_\star]}$ contains a weakly
convergent sequence, $\{\phi_{m_j}\}_{j\in\NN}$.
We denote the weak limit of the latter by $\phi$ and verify that $\phi\not=0$
in the following. 
The assertion of Theorem~\ref{thm-ex-PF} will then follow
from Lemma~\ref{le-ex-allg}. (In fact, we shall show that
$\phi_{m_j}\to\phi$ strongly in $\HR_4$ along a subsequence.)

To verify that $\phi\not=0$ one can argue
as in \cite{GLL2001}. Essentially, we only have to replace
the Rellich-Kondrashov theorem applied there by a
suitable imbedding theorem for spaces of functions with fractional
derivatives. (In the non-relativistic case the ground states $\phi_m$
possess weak derivatives with respect to the electron coordinates,
whereas in our case we only have Inequality~\eqref{tina} below.)
For the convenience of the reader we present the complete argument.
Writing 
$\phi_m=(\phi_m^{(n)})_{n=0}^\infty\in\bigoplus_{n=0}^\infty\Fock^{(n)}[\HP]$
we infer from the soft photon bound
that
$$
\sum_{n=n_0}^\infty\|\phi_m^{(n)}\|^2\,
\klg\,\frac{1}{n_0}\,\sum_{n=0}^\infty n\,\|\phi_m^{(n)}\|^2
\,=\,
\frac{1}{n_0}\int\|a(k)\,\phi_m\|^2\,dk\,
\klg\,
\frac{C}{n_0}\,,
$$
for $m\in(0,m_\star]$ and some $m$-independent constant $C\in(0,\infty)$.
Given some $\ve>0$
we fix $n_0\in\NN$ so large that 
\begin{equation}
C/n_0\,<\,\ve\,.\label{sabine1}
\end{equation}
By virtue of \eqref{exp-loc-sup} we further find
some $R>0$ such that, for all $m\in(0,m_\star]$,
\begin{equation}\label{sabine2}
\int_{|\V{x}|\grg R/2}\|\phi_m\|_{\CC^4\otimes\Fock}^2(\V{x})\,d^3\V{x}\,<\,\ve\,.
\end{equation}
In addition,
the soft photon bound ensures that
$\phi_m^{(n)}(\V{x},\vs,k_1,\dots,k_n)=0$, for almost every
$(\V{x},\vs,k_1,\dots,k_n)\in\RR^3\times\{1,2,3,4\}\times(\RR^3\times\ZZ_2)^n$,
$k_j=(\V{k}_j,\lambda_j)$,
such that $|\V{k}_j|>\Lambda$, for some $j\in\{1,\dots,n\}$.
(Here and henceforth $\vs$ labels the four spinor components.)
For $0<n<n_0$ and some fixed
$\ul{\theta}
=(\vs,\lambda_1,\dots,\lambda_n)\in\{1,2,3,4\}\times\ZZ_2^n$
we set 
$$
\phi_{m,\ul{\theta}}^{(n)}(\V{x},\V{k}_1,\dots,\V{k}_n)
\,:=\,\phi_m^{(n)}(\V{x},\vs,\V{k}_1,\lambda_1,\ldots,\V{k}_n,\lambda_n)
$$
and similarly for $\phi$.
Moreover, we set, for every $\delta\grg0$, 
$$
Q_{n,\delta}\,:=\,
\big\{\,(\V{x},\V{k}_1,\dots,\V{k}_n)\::\;
|\V{x}|< R-\delta\,,
\,\delta<|\V{k}_j|<\Lambda-\delta\,,\;j=1,\ldots,n\,\big\}\,.
$$
Fixing some small $0<\delta<\min\{m_\star,R/2,\UV/4\}$ we
pick some cut-off function $\chi\in C_0^\infty(\RR^{3(n+1)},[0,1])$ such that
$\chi\equiv1$ on $Q_{n,2\delta}$ and $\supp(\chi)\subset Q_{n,\delta}$
and define $\psi^{(n)}_{m,\ul{\theta}}:=\chi\,\phi^{(n)}_{m,\ul{\theta}}$.
As a next step the photon derivative bound is used to
show that $\{\psi_{m,\ul{\theta}}^{(n)}\}_{m\in(0,\delta]}$
is a bounded family in the anisotropic Nikol'ski{\u\i} space\footnote{
For $r_1,\ldots,r_d\in[0,1]$, $q_1,\ldots,q_d\grg1$, we have
$H^{(r_1,\ldots,r_d)}_{q_1,\ldots,q_d}(\RR^d)
:=\bigcap_{i=1}^d H^{r_i}_{q_ix_i}(\RR^d)$.
For $r_i\in[0,1)$, a measurable function
$f:\RR^d\to\CC$ belongs to the class $H^{r_i}_{q_ix_i}(\RR^d)$,
if $f\in L^{q_i}(\RR^d)$ and there is some $M\in(0,\infty)$
such that
\begin{equation}\label{Nik1}
\|f(\cdot+h\,{\sf e}_i)-f\|_{L^{q_i}(\RR^d)}
\,\klg\,M\,|h|^{r_i}\,,\qquad h\in\RR\,,
\end{equation}
where ${\sf e}_i$ is the $i$-th canonical unit vector in $\RR^d$.
If $r_i=1$ then \eqref{Nik1} is replaced by
\begin{equation}\label{Nik2}
\|f(\cdot+h\,{\sf e}_i)-2f+f(\cdot-h\,{\sf e}_i)\|_{L^{q_i}(\RR^d)}\,
\klg\,M\,|h|\,,\qquad h\in\RR\,.
\end{equation}
$H^{(r_1,\ldots,r_d)}_{q_1,\ldots,q_d}(\RR^d)$ is a Banach space
with norm
$$
\|f\|^{(r_1,\ldots,r_d)}_{q_1,\ldots,q_d}\,:=\,
\max_{1\klg i\klg d}\|f\|_{L^{q_i}(\RR^d)}\,+\,
\max_{1\klg i\klg d} M_i\,,
$$ 
where $M_i$ is the infimum of all constants $M>0$
satisfying \eqref{Nik1} or \eqref{Nik2}, respectively.
Finally, we abbreviate 
$H^{(r_1,\ldots,r_d)}_q(\RR^d):=H^{(r_1,\ldots,r_d)}_{q,\ldots,q}(\RR^d)$.
}
$H^{\V{s}}_{\V{q}}(\RR^{3(n+1)})$, 
where $\V{s}=(1/2,1/2,1/2,1,\dots,1)$ and $\V{q}=(2,2,2,p,\dots,p)$
with
$p\in[1,2)$.
In fact, employing the H\"older
inequality (w.r.t. $d^3\V{x}\,d^3\V{k}_2\dots d^3\V{k}_n$) and
the photon derivative bound \eqref{pdb-kp}, 
we obtain as in \cite{GLL2001}, for 
$p\in[1,2)$ and $m\in(0,\delta]$,
\begin{align*}
&
\int\limits_{{Q_{n,\delta}\cap\atop\{\delta<|\V{k}_1+\V{h}|<\UV\}}}\!\!\!\!\!\!\!
\big|\phi_{m,\ul{\theta}}^{(n)}(\V{x},\V{k}_1+\V{h},\V{k}_2,\ldots,\V{k}_n)
-\phi_{m,\ul{\theta}}^{(n)}(\V{x},\V{k}_1,\dots,\V{k}_n)\big|^p
\,d^3\V{x}\,d^3\V{k}_1\dots d^3\V{k}_n
\\
&\klg\,
\frac{([4\pi/3]^{n}R^3\,\UV^{3(n-1)})^{\frac{2-p}{2}}}{n^{p/2}}
\sum_{\lambda\in\ZZ_2}\int\limits_{{m<|\V{k}|<\UV,\atop m<|\V{k}+\V{h}|<\UV}}
\big\|\,a(\V{k}+\V{h},\lambda)\,\phi_m
-a(\V{k},\lambda)\,\phi_m\,\big\|^p
d^3\V{k}
\\
&\klg\,
C\,|\V{h}|^p\!\!\!\int\limits_{|(k^{(1)},k^{(2)})|<\UV}\!\!\!
\bigg\{\!
\int\limits_0^{|(k^{(1)},k^{(2)})|}\!\frac{dk^{(3)}}{|(k^{(1)},k^{(2)})|^{p/2}}\,+
\!\!\!
\int\limits_{|(k^{(1)},k^{(2)})|}^\UV\frac{dk^{(3)}}{|k^{(3)}|^{p/2}}
\bigg\}\frac{dk^{(1)}dk^{(2)}}{|(k^{(1)},k^{(2)})|^p}
\\
&=\,C'\,|\V{h}|^p\,,
\end{align*}
where the constants $C,C'\in(0,\infty)$ do not depend on $m\in(0,\delta]$.
Since $\phi_m^{(n)}$ is symmetric in the photon variables the
previous estimate implies \cite[\textsection4.8]{Nikolskii}
that the weak first order partial derivatives of $\phi_{m,\ul{\theta}}^{(n)}$ 
with respect to its last $3n$ variables exist on 
$Q_{n,\delta}$ and that
$$
\|\phi_{m,\ul{\theta}}^{(n)}\|_{W^{\V{r}}_p(Q_{n,\delta})}^p
\,=\,
\|\phi_{m,\ul{\theta}}^{(n)}\|^p_{L^p(Q_{n,\delta})}+
\sum_{j=1}^n\sum_{i=1}^3
\|\partial_{k_j^{(i)}}\phi_{m,\ul{\theta}}^{(n)}\|^p_{L^p(Q_{n,\delta})}
\,\klg\,C''\,,
$$
for $m\in(0,\delta]$ and some $m$-independent $C''\in(0,\infty)$,
with $\V{r}:=(0,0,0,1,\dots,1)$.
The previous estimate implies 
$\|\psi_{m,\ul{\theta}}^{(n)}\|_{W_p^\V{r}(\RR^{3(n+1)})}\klg C'''$,
for some $C'''\in(0,\infty)$ which does not depend on $m\in(0,\delta]$. 
Moreover, the anisotropic Sobolev space $W^{\V{r}}_p(\RR^{3(n+1)})$
is continuously imbedded into $H^{\V{r}}_p(\RR^{3(n+1)})$; see, e.g.,
\cite[\textsection6.2]{Nikolskii}.
Furthermore, since $\core$ is a form core of
$\PFm{\gamma}$, $m>0$, Corollary~\ref{cor-pia}
shows that
\begin{equation}\label{tina}
\SPn{\phi_m^{(n)}}{|\DO|\,\phi_m^{(n)}}\,\klg\,c^{-1}\,
\SPn{\phi_m}{\PFm{\gamma}\,\phi_m}+c\,=\,
c^{-1}\,E_m+c\,,\quad n\in\NN\,,
\end{equation}
for some $m$-independent $c\in(0,\infty)$.
Therefore, $\{\phi_{m,\ul{\theta}}^{(n)}\}_{m\in(0,m_\star]}$ and, hence, 
$\{\psi_{m,\ul{\theta}}^{(n)}\}_{m\in(0,m_\star]}$
are bounded families in the Bessel potential,
or, Liouville
space $L^{\V{r}'}_2(\RR^{3(n+1)})$, $\V{r}':=(1/2,1/2,1/2,0,\dots,0)$,
where the fractional derivatives are defined by means of the Fourier
transform. The imbedding $L^{\V{r}'}_2(\RR^{3(n+1)})\to H^{\V{r}'}_2(\RR^{3(n+1)})$
is continuous, too
\cite[\textsection9.3]{Nikolskii}.
Altogether it follows that $\{\psi_{m,\ul{\theta}}^{(n)}\}_{m\in(0,\delta]}$
is a bounded family in $H^\V{s}_\V{q}(\RR^{3(n+1)})$.
Now
we may apply the compactness theorem
\cite[Theorem~3.2]{Nikolskii1958}. The latter ensures
that $\{\psi_{m,\ul{\theta}}^{(n)}\}_{m\in(0,\delta]}$
contains a sequence which is strongly convergent
in $L^2(Q_{n,2\delta})$ provided $1-3n(p^{-1}-2^{-1})>0$.
Of course, we can choose $p<2$ large enough such that
the latter condition is fulfilled, for
all $n=1,\ldots,n_0-1$.
By finitely many repeated selections of subsequences 
we may hence assume without loss of generality
that $\{\phi_{m_j,\ul{\theta}}^{(n)}\}_{j\in\NN}$ converges strongly
in $L^2(Q_{n,2\delta})$ to $\phi^{(n)}_{\ul{\theta}}$, for $0\klg n<n_0$.
In particular, by the choice of $n_0$ and $R$ 
in \eqref{sabine1} and \eqref{sabine2},
$$
\|\phi\|^2\grg
\lim_{j\to\infty}
\sum_{n=0}^{n_0-1}\sum_{\ul{\theta}}
\|\phi^{(n)}_{m_j,\ul{\theta}}\|_{L^2(Q_{n,2\delta})}^2
\grg\lim_{j\to\infty}\|\phi_{m_j}\|^2-2\ve-c(\delta)=1-2\ve-c(\delta),
$$
where we use the soft photon bound to estimate
\begin{align*}
&\sum_{n=1}^{n_0-1}\sum_{\ul{\theta}}
\big\|\,
\phi^{(n)}_{m_j,\ul{\theta}}\,\id{\{\exists\,i:\,
|\V{k}_i|\klg2\delta\,\vee\,|\V{k}_i|\grg\UV-2\delta\}}
\,\big\|^2
\\
&\klg\,
\sum_{\lambda\in\ZZ_2}\int\limits_{{\{|\V{k}|\klg2\delta\}\cup\atop
\{|\V{k}|\grg\UV-2\delta\}}}
\!\!\!\!\|a(\V{k},\lambda)\,\phi_{m_j}\|^2\,d^3\V{k}
\,\klg\,C\Big(\int_0^{2\delta}+\int_{\UV-2\delta}^\UV\Big)
\frac{r^2\,dr}{r}\,=:\,c(\delta)\,\to\,0\,,
\end{align*}
as $\delta\searrow0$.
Since $\delta>0$ and $\ve>0$ are arbitrary we get $\|\phi\|=1$, whence
$\phi_{m_j}\to\phi$ strongly in $\HR_4$.
\end{proof}

%%%%%%%%%%%%%%%%%%%%%%%%%%%%%%%%%%%%%%%%%%%%%%%%%%%%%%%%%%%%%%%%%%%%%%%%%%

\subsection{Existence of ground states with infra-red cut-off }
\label{ssec-ex-m-PF}

At the end of this subsection we prove Proposition~\ref{prop-gs-PFm}.
In order to do so we first have to extend our results on
the spatial localization of low-lying spectral subspaces
a little bit. This extension requires the following
inequality.

\begin{lemma}
Let $e^2,\UV>0$, $m,\ve\grg0$, 
$\gamma\in[0,2/\pi)$, and $a\in[0,1)$.
Moreover, let $F\in C^\infty(\RR^3_{\V{x}},[0,\infty))\cap L^\infty$
satisfy $|\nabla F|\klg a$.
Then 
\begin{equation}\label{Re-conj-op}
\big|\,\Re\SPb{\vp}{[\PFme{\gamma}\,,\,e^F]\,e^{-F}\,\vp}\,\big|
\,\klg\,2\,a^2\,J(a)\,\|\vp\|^2\,,\qquad \vp\in\core\,.
\end{equation}
For $m>0$ and $\ve=0$,
the same estimate holds true with $\PFm{\gamma}$ replaced by $\PFmg{\gamma}$
and $\core$ replaced by the dense subspace 
$\core^>$ defined in \eqref{def-core>}.
\end{lemma}

\begin{proof}
This lemma is a special case of \cite[Lemma~5.7]{MatteStockmeyer2009a}.
\end{proof}

\smallskip

\noindent
The last assertion of the next lemma is also used in the proof
of the infra-red bounds.

\begin{lemma}\label{le-HfmeF}
Let $e^2,\UV>0$, $\gamma\in[0,2/\pi)$, 
and let $m_1$, $\ve_1$, $a_1$, $\delta_1$,
and $J_1(m,\ve)$ be as in 
Proposition~\ref{prop-exp-loc}.
Assume that $F\in C^\infty(\RR^3_{\V{x}},[0,\infty))$ satisfies
$|\nabla F|\klg a_1/2$ and $F(\V{x})=a_1|\V{x}|$,
for large $|\V{x}|$. Then 
there is some $C\in(0,\infty)$ such that,
for all $m\in[0,m_1]$, $\ve\in[0,\ve_1]$, and
$\psi\in\Ran(\id_{J_1(m,\ve)}(\PFme{\gamma}))$,
we have $e^F\,\psi\in\form(\PFme{\gamma})$ and 
\begin{equation}\label{est-HfneF1}
\big\|\,(\PFme{\gamma}-E_{\gamma,m,\ve})^{1/2}\,e^F\,\psi\,\big\|^2\klg
\|e^{2F}\,\psi\|\,\|(\PFme{\gamma}-E_{\gamma,m,\ve})\,\psi\|
+2a^2J(a)\,\|e^{F}\,\psi\|^2.
\end{equation}
In particular,
\begin{equation}\label{est-HfmeF2}
\big\|\,(\PFme{\gamma}-E_{\gamma,m,\ve})^{1/2}\,e^F\,
\id_{J_1(m,\ve)}(\PFme{\gamma})\,\big\|\,\klg\,
C'\,,
\end{equation}
where the constant $C'\in(0,\infty)$ neither depends 
on $m\in[0,m_1]$ nor on $\ve\in[0,\ve_1]$.
Moreover, for $\cO\in\{|\DO|,|\DAme|,\He\}$ and 
$\psi\in\id_{J_1(m,\ve)}(\PFme{\gamma})$,
we have 
$e^F\,\psi\in\dom(\cO^{1/2})$
and 
\begin{equation}\label{veronika}
\sup\big\{\,\|\cO^{1/2}\,e^F\,\id_{J_1(m,\ve)}(\PFme{\gamma})\|\,:\:m\in[0,m_1]\,,
\ve\in[0,\ve_1]\,\big\}
<\,\infty\,.
\end{equation}
For $\ve=0$ and $m>0$,
the same assertions hold true with $\PFme{\gamma}$, $|\DAme|$, and $\He$
replaced by $\PFmg{\gamma}$, $|\DAmg|$, and $\Hfmg$, respectively.
\end{lemma}

\begin{proof}
First, let $\vp\in\core$ and let
$\wt{F}\in C^\infty(\RR^3_{\V{x}},[0,\infty))\cap L^\infty$
such that $|\nabla \wt{F}|\klg a_1/2$.
Applying
\eqref{Re-conj-op} with $\vp$ replaced by $e^{\wt{F}}\,\vp$
we obtain
\begin{align}
\big\|\,&(\PFme{\gamma}-E_{\gamma,m,\ve})^{1/2}
\,e^{\wt{F}}\,\vp\,\big\|^2\nonumber
\\
&=\,\nonumber
\Re\Big[\,\SPb{e^{2\wt{F}}\,\vp}{(\PFme{\gamma}-E_{\gamma,m,\ve})\,\vp}
+\SPb{e^{\wt{F}}\,\vp}{
[\PFme{\gamma}\,,\,e^{\wt{F}}]\,e^{-\wt{F}}\,(e^{\wt{F}}\,\vp)}\,\Big]
\\
&\klg\,\label{klaus1001}
\big|\SPb{e^{2\wt{F}}\,\vp}{(\PFme{\gamma}-E_{\gamma,m,\ve})\,\vp}\big|\,
+\,c(a)\,\|e^{\wt{F}}\,\vp\|^2\,,
\end{align}
where $c(a)=2a^2\,J(a)$.
In \cite[Lemma~5.8]{MatteStockmeyer2009a} 
we proved the following inequality,
$$
\SPb{e^{G}\,\vp}{\PFme{\gamma}\,e^{G}\,\vp}\,\klg\,
c_1\,\|e^{G}\|^2\,\SPn{\vp}{\PFme{\gamma}\,\vp}\,+\,
c_2\,\|e^{G}\|^2\,\|\vp\|^2\,,\qquad
\vp\in\core\,,
$$ 
for every $G\in C^\infty(\RR^3,[0,\infty))\cap L^\infty$ with
$\|\nabla G\|_\infty<1$.
In fact, we stated this inequality 
only for $\gamma=0$. Since $\tgV$ is relatively
form bounded with respect to $\PFme{0}$ with relative bound
less than one it is clear, however, that it holds true for
$\gamma_\ve\in(0,2/\pi)$ also, 
with new constants $c_1,c_2\in(0,\infty)$
of course. In particular, if $\psi\in\form(\PFme{\gamma})$,
$\vp_n\in\core$, $n\in\NN$, and $\vp_n\to\psi$ with respect to
the form norm of $\PFme{\gamma}$, then 
$e^{\wt{F}}\,\psi,e^{2\wt{F}}\,\psi\in\form(\PFme{\gamma})$ and
$e^{\wt{F}}\,\vp_n\to e^{\wt{F}}\,\psi$ and
$e^{2\wt{F}}\,\vp_n\to e^{2\wt{F}}\,\psi$ with respect to
the form norm of $\PFme{\gamma}$ also.
We may thus replace $\vp$ by any $\psi\in\form(\PFme{\gamma})$ in 
\eqref{klaus1001}.
We fix some $\psi\in\Ran(\id_I(\PFme{\gamma}))$ in what follows. 
Then we
additionally know that $\psi\in\dom(\PFme{\gamma})$ since
$I$ is bounded and we arrive at
\begin{align*}
\big\|\,&(\PFme{\gamma}-E_{\gamma,m,\ve})^{1/2}
\,e^{\wt{F}}\,\psi\,\big\|^2\nonumber
\,\klg\,
\|e^{2\wt{F}}\,\psi\|\,\|(\PFme{\gamma}-E_{\gamma,m,\ve})\,\psi\|+
c(a)\,\|e^{\wt{F}}\,\psi\|^2\,.
\end{align*}
On the other hand we know from \eqref{exp-loc-EI} that 
$\int e^{4F(\V{x})}\,\|\psi\|^2_{\Fock}(\V{x})\,d^3\V{x}<\infty$.
We pick a sequence of bounded functions
$F_n\in C^\infty(\RR^3,[0,\infty))$ such that
$|\nabla F_n|\klg a_1/2$, $n\in\NN$, and $F_n\nearrow F$.
Then $e^{F_n}\,\psi\to e^F\,\psi$ and $e^{2 F_n}\to e^{2F}\,\psi$
in $\HR_4$ by dominated convergence. 
Inserting $F_n$ for $\wt{F}$ in the previous
estimate we conclude that
the densely defined linear functional
$$
f(\eta)\,:=\,
\SPb{e^F\,\psi}{(\PFme{\gamma}-E_{\gamma,m,\ve})^{1/2}\,\eta}=
\lim_{n\to\infty}\SPb{(\PFme{\gamma}-E_{\gamma,m,\ve})^{1/2}
\,e^{F_n}\,\psi}{\eta}\,,
$$
for all $\eta\in\form(\PFme{\gamma})$,
is bounded,
$$
|f(\eta)|\,\klg\,\big(\|e^{2F}\,\psi\|\,
\|(\PFme{\gamma}-E_{\gamma,m,\ve})\,\psi\|
+c(a)\,\|e^F\,\psi\|^2\big)^{1/2}\,\|\eta\|\,,
$$
for $\eta\in\form(\PFme{\gamma})$. 
Since $(\PFme{\gamma}-E_{\gamma,m,\ve})^{1/2}$
is self-adjoint with domain $\form(\PFme{\gamma})$ it follows that
$e^F\,\psi\in\form(\PFme{\gamma})$ and
$\|(\PFme{\gamma}-E_{\gamma,m,\ve})^{1/2}\,e^F\,\psi\|=\|f\|$.

The inequality \eqref{est-HfmeF2} follows
from \eqref{est-HfneF1} and Proposition~\ref{prop-exp-loc}.
The bound \eqref{veronika} follows from
\eqref{est-HfmeF2} and Corollary~\ref{cor-pia}
where the constants can be chosen uniformly in $m\in[0,m_1]$
and $\ve\in[0,\ve_1]$.
In all the arguments above we can replace
$\PFm{\gamma}$ by $\PFmg{\gamma}$, when $\ve=0$,  
and all subspaces of $\HR_4$
by the corresponding truncated spaces without any further changes.
So it is clear that the last assertion is valid, too. 
\end{proof}

\smallskip

\noindent
Next, we show that the condition \eqref{def-delta(a)} is
fulfilled with $\V{G}=\V{G}^\phys=e^{-i\V{k}\cdot\V{x}}\,\V{g}$,
$\wt{\V{G}}=e^{-i\vnu_\ve\cdot\V{x}}\,\V{g}_{m,\ve}$, 
and $\vo=\omega$, for every $a>0$.
To this end we set
\begin{equation}\label{def-treps}
\triangle^2(\ve):=\!
\int\!\id_{\cA_m}(\V{k})
\Big(1+\frac{1}{\omega_\ve(k)}\Big)
\sup_{\V{x}\in\RR^3}\big\{
e^{-a|\V{x}|}\,|e^{-i\V{k}\cdot\V{x}}\,\V{g}(k)-e^{-i\vnu_\ve(k)\cdot\V{x}}\,
\V{g}_{m,\ve}(k)|^2
\big\}dk.
\end{equation}

\begin{lemma}\label{le-nina1}
Let $e^2,\UV,a>0$, and $m\in(0,1]$. As $\ve>0$ tends to zero, we have 
\begin{equation}\label{heiko}
\triangle^2(\ve)\,=\,
\frac{o(\ve^0)}{m}\,,
\end{equation}
where the little $o$-symbol is uniform in $m\in(0,1]$.
\end{lemma}

\begin{proof}
Of course, we have $1+\omega_\ve(k)^{-1}\klg2/m$ on $\cA_m\times\ZZ_2$,
whence
\begin{align}\nonumber
&\sum_{\lambda\in\ZZ_2}\int_{\cA_m}
\Big(1+\frac{1}{\omega_\ve(k)}\Big)
\,\sup_{\V{x}\in\RR^3}\big\{
e^{-a|\V{x}|}\,|e^{-i\V{k}\cdot\V{x}}\,\V{g}(k)-e^{-i\vnu_\ve(k)\cdot\V{x}}\,
\V{g}_{m,\ve}(k)|^2
\big\}\,d^3\V{k}
\\
&\klg\,\nonumber
\frac{4}{m}\,\|\V{g}-P_\ve\V{g}\|_{\HP_m\otimes\CC^3}^2
\\
&\;\;
\,+\,
\frac{4}{m}\sum_{\lambda\in\ZZ_2}\int\limits_{\cA_m}\sup_{\V{x}\in\RR^3}e^{-a|\V{x}|}
\Big|
\int\limits_{Q_m^\ve(\vnu_\ve(\V{k}))}
\frac{e^{-i\V{k}\cdot\V{x}}-e^{-i\vnu_\ve(k)\cdot\V{x}}}{|Q_m^\ve(\vnu_\ve(\V{k}))|}
\,\V{g}(\V{p},\lambda)\,d^3\V{p}\,
\Big|^2\,d^3\V{k}\,,\label{uschi1}
\end{align}
where $P_\ve$ is defined in \eqref{def-Pve}.
Now, let $\eta>0$.
We choose some 
$\V{h}\in C_0^\infty(\RR^3\times\ZZ_2,\CC^3)$ 
such that
$\|\V{g}-\V{h}\|_{\HP\otimes\CC^3}\klg \eta$.
Applying $P_\ve$ to all three components of $\V{g}$ and $\V{h}$
we get
$\|P_\ve \V{g}-P_\ve\V{h}\|_{\HP_m\otimes\CC^3}\klg\eta$,
for all $m>0$,
and since $\V{h}$ is uniformly continuous on its
compact support it is clear that
$\|P_\ve\V{h}-\V{h}\|_{\HP_m\otimes\CC^3}\to0$, as $\ve\searrow0$,
uniformly in $m>0$.
Since $\eta>0$ is arbitrarily small it follows that
$\|\V{g}-P_\ve\V{g}\|_{\HP_m\otimes\CC^3}\to0$,
$\ve\searrow0$, uniformly in $m>0$.
Furthermore,
since $|\V{k}-\vnu_\ve(\V{k})|\klg\sqrt{3}\,\ve/2$, for all
$\V{k}\in\cA_m$, and
$|e^{-i\V{p}\cdot\V{x}}-e^{-i\V{k}\cdot\V{x}}|\klg|\V{k}-\V{p}|\,|\V{x}|$,
we deduce by means of the Cauchy-Schwarz inequality that 
the term in the last line of \eqref{uschi1} is bounded
from above by
\begin{align*}
\frac{3\,\ve^2}{m}\,\sup_{\V{x}\in\RR^3}\{e^{-a|\V{x}|}\,|\V{x}|^2\}
\sum_{\lambda\in\ZZ_2}\sum_{\vnu}
\int\limits_{Q_m^\ve(\vnu)}
\frac{1}{|Q_m^\ve(\vnu_\ve(\V{k}))|}
\int\limits_{Q_m^\ve(\vnu_\ve(\V{k}))}|\V{g}(\V{p},\lambda)|^2\,d^3\V{p}
\,d^3\V{k}\,,
\end{align*}
which is less than or equal to $(6\,\ve^2/[ma^2])\,\|\V{g}\|^2_\HP$.
(The second sum runs over all $\vnu\in(\ve\ZZ)^3$ such that
$Q_\ve^m(\vnu)\not=\varnothing$ and we recall that $\vnu_\ve(\V{k})=\vnu$
when $\V{k}\in Q_\ve^m(\vnu)$.)
\end{proof}

\smallskip

\noindent
In the next lemma we compare the ground state energies
\begin{align}
E_m\,&=\,\inf\spec[\PFmg{\gamma}]\qquad\textrm{and}\qquad
E_{m,\ve}\,:=\,E_{\gamma,m,\ve}\,=\,\inf\spec[\PFme{\gamma}]\,.
\label{def-wtE}
\end{align}
We recall that the Coulomb coupling constant in $\PFme{\gamma}$
has been slightly changed to $\gamma_\ve=\gamma/(1-c(\ve))$,
where the function
$c:(0,1)\to(0,1)$
has not yet been specified.
From now on we choose
\begin{equation}\label{def-c(eps)}
c(\ve)\,:=\,\min\{1/2\,,\,m^{1/4}\,\triangle^{1/2}(\ve)\}\,,
\qquad \ve\in(0,1)\,,
\end{equation}
so that $c(\ve)\to0$ uniformly in $m$, as $\ve\searrow0$.

\begin{lemma}\label{le-Em-Eme}
Let $e^2,\UV>0$, $m\in(0,m_1]$, and $\gamma\in(0,2/\pi)$.
Then 
$$
E_m\,\klg\,E_{m,\ve}\,+\,o(\ve^0)/m\,,
$$
where the little $o$-symbol is uniform in $m$.
\end{lemma}

\begin{proof}
By virtue of \eqref{omega-omegaeps1} and Corollary~\ref{cor-pia}
we know that $\form(\PFmg{\gamma})=\form(\PFme{\gamma})$.
In particular, we may pick some $\rho\in(0,\delta_1]$ and try some
normalized
$\phi_\ve^\rho\in\Ran(\id_{[E_{m,\ve},E_{m,\ve}+\rho)}(\PFme{\gamma}))$
as a test function for $\PFmg{\gamma}$. 
Here $\delta_1$ is the parameter
appearing in Proposition~\ref{prop-exp-loc} and we shall also
employ the parameters $a_1,\ve_1,m_1$, and the interval
$J_1(m,\ve)$ introduced there. 
We obtain
\begin{align}
E_m\,&\klg\,\nonumber
\SPb{\phi_\ve^\rho}{\PFmg{\gamma}\,\phi_\ve^\rho}
\\
&\klg\,\nonumber
E_{m,\ve}\,+\,\rho\,+\,
\SPb{\phi_\ve^\rho}{\big(|\DAmg|-|\DAme|\big)\,\phi_\ve^\rho}
\\
&\quad+\,\nonumber
\frac{c(\ve)\,\gamma}{1-c(\ve)}\,
\SPb{\phi_\ve^\rho}{|\V{x}|^{-1}\,\phi_\ve^\rho}
\,+\,\SPb{\phi_\ve^\rho}{\big(\Hfmg-\He\big)\,\phi_\ve^\rho}
\\
&\klg\,\nonumber
E_{m,\ve}+\rho+
\triangle^{1/2}(\ve)\,\SPn{\phi_\ve^\rho}{|\DAme|\,\phi_\ve^\rho}
\\
&\quad\nonumber
+\,\triangle^{1/2}(\ve)\,\big\|(\He+E)^{1/2}e^F\,\phi_\ve^\rho\big\|^2
+\,C\,\triangle^{3/2}(\ve)\,\|\phi_\ve^\rho\|^2
\\
&\quad\nonumber
+\,\frac{c(\ve)\,\gamma}{1-c(\ve)}\,
\SPb{\phi_\ve^\rho}{|\V{x}|^{-1}\,\phi_\ve^\rho}+\,
\frac{\sqrt{3}\,\ve/m}{1-\sqrt{3}\,\ve/m}\,
\SPn{\phi_\ve^\rho}{\He\,\phi_\ve^\rho}
\,,
\end{align}
where $E\equiv E(e^2,\UV)\in(0,\infty)$ and
$F$ is chosen as in Lemma~\ref{le-HfmeF}.
In the second step we used \eqref{pia0}  
with $\epsilon=\tau=\triangle^{1/2}(\ve)$ and \eqref{omega-omegaeps1}
which implies 
$\Hfmg-\He\klg\frac{\sqrt{3}\,\ve}{m}\,(1-\frac{\sqrt{3}\,\ve}{m})^{-1}\,\He$.
By virtue of \eqref{veronika} we have
$$
\big\|(\He+E)^{1/2}e^F\,\phi_\ve^\rho\big\|\,\klg\,
\big\|(\He+E)^{1/2}e^F\,\id_{J_1(m,\ve)}(\PFme{\gamma})\big\|
\,\klg\,C\,,
$$ 
where the constant $C\in(0,\infty)$ neither depends on $m\in(0,m_1]$
nor $\ve\in(0,\ve_1]$.
Employing Corollary~\ref{cor-pia}
once more we conclude that
\begin{align*}
E_m\,\klg\,
E_{m,\ve}+\rho+C\,\triangle^{1/2}(\ve)
+(o(\ve^0)/m)\,\SPn{\phi_\ve^\rho}{(\PFme{\gamma}+1)\,\phi_\ve^\rho}\,,
\end{align*}
where the little $o$-symbol is uniform in $m\in(0,m_1]$
and $\rho$ is arbitrarily small.
\end{proof}

\smallskip

\begin{proof}[Proof of Proposition~\ref{prop-gs-PFm}]
Let $m_1$, $\ve_1$, $a_1$, $\delta_1$, and $F$ be as in the statement
of Proposition~\ref{prop-exp-loc} and set
\begin{equation}\label{def-chi-PF}
\chi\,:=\,\id_{(-\infty\,,\,E_m+m/4]}(\PFmg{\gamma})\,.
\end{equation}
We always assume that $m\klg m_1$, $m/4\klg\delta_1$, and $\ve\klg\ve_1$
in the following so that \eqref{exp-loc-sup} can be applied to $\chi$.
On account \eqref{omega-omegaeps2}, Lemma~\ref{le-nina1},
and \eqref{pia0} with $\epsilon=\tau=c(\ve)$, where $c(\ve)$
is given by \eqref{def-treps} and \eqref{def-c(eps)},
we have
\begin{align*}
&\chi\,\{\,\PFmg{\gamma}-E_m-m/2\,\}\,\chi
\\
&\;\;\grg\,
(1-c(\ve))\,\chi\,\big\{\,|\DAe|
-\gamma_\ve/|\V{x}|+\He-E_m-m/2\,\big\}\,\chi
-c(\ve)\,T_1\,,
\end{align*}
where $\gamma_\ve=\gamma/(1-c(\ve))$ and the norm of 
$$
T_1\,:=\,\chi\,\big\{\,e^F\,(H_{f,m}^>+E+E_m+m/2)\,e^F\,\big\}\,\chi
$$
is bounded
uniformly in $m\in(0,m_1]$ due to \eqref{exp-loc-EI} and Lemma~\ref{le-HfmeF}.
(The constant $E$ appears when we apply \eqref{pia0}. It
depends on $e^2$ and $\UV$ and is proportional to
$\triangle^{3/2}(\ve)/m^{5/4}$.)
To proceed further we introduce the subspaces of discrete and fluctuating
photon states,
\begin{equation*}
\HP_m^d\,:=\,P_\ve\,\HP_m^>\,,\qquad \HP_m^f\,:=\,\HP_m^>\ominus\HP_m^d\,,
\end{equation*}
where $P_\ve$ is defined in \eqref{def-Pve}.
The splitting $\HP_m^>=\HP_m^d\oplus\HP_m^f$
gives rise to an isomorphism
$$
L^2(\RR^3,\CC^4)\otimes\Fock(\HP_m^>)\,\cong\,
\big(L^2(\RR^3,\CC^4)\otimes\Fock[\HP_m^d]\big)\otimes\Fock[\HP_m^f]
$$
and we observe that the Dirac operator and the field energy
decompose under the above isomorphism as 
\begin{equation}\label{valerie0}
\DAe\,\cong\,\DAed\otimes\id^f\,,\qquad
\He\,=\,\Hed\otimes\id^f\,+\,\id^d\otimes\Hef\,.
\end{equation}
Here and in the following we designate 
operators acting in the Fock space factors 
$\Fock[\HP_m^\ell]$, $\ell\in\{d,f\}$, by 
the corresponding superscript 
$\ell\in\{d,f\}$.
In fact, the discretized vector potential $\V{A}_{m,\ve}$
acts on the various $n$-particle sectors in $\Fock[\HP_m^>]$
by tensor-multiplying or taking scalar products with elements
from $\HP_m^d$ (apart from symmetrization and a normalization constant). 
Denoting the projection onto the vacuum sector
in $\Fock[\HP_m^\ell]$ by $P_{\Omega^\ell}$, 
writing $P_{\Omega^\ell}^\bot:=\id^\ell-P_{\Omega^\ell}$, $\ell\in\{d,f\}$,
and using $\Hef\,P_{\Omega^f}=0$,
we thus obtain
\begin{align}
\nonumber
\chi\,\big\{&\,\PFmg{\gamma}-E_m-m/2\,\big\}\,\chi\,+\,c(\ve)\,T_1
\\
&\grg\,\label{valerie1}
(1-c(\ve))\,\chi\,\big\{\,\big[
\,|\DAed|-\gamma_\ve/|\V{x}|+\Hed-E_m-m/2\,
\big]\otimes P_{\Omega^f}\,\big\}\,\chi
\\
& 
\quad+\,\label{valerie2}
(1-c(\ve))\,\chi\,\big\{\,\big[\,
|\DAed|-\gamma_\ve/|\V{x}|+\Hed-E_{m,\ve}
\,\big]\otimes P_{\Omega^f}^\bot\,\big\}\,\chi
\\
& 
\quad+\,\label{valerie3}
(1-c(\ve))\,\chi\,\big\{\,\id_\el\otimes\id_d\otimes(\Hef
-E_m+E_{m,\ve}-m/2)\,
P_{\Omega^f}^\bot\,\big\}\,\chi\,.
\end{align}
Here $E_{m,\ve}$ is defined in \eqref{def-wtE}.
Setting
$$
X_\ve^d\,:=\,|\DAed|-\gamma_\ve/|\V{x}|+\Hed
$$
we observe that
$
X_\ve^d-E_{m,\ve}\,\id_d\grg0
$
so that the term in \eqref{valerie2} is non-negative.
In fact, let $\rho>0$ and pick some $\phi_d\in\form(X_\ve^d)$,
$\|\phi_d\|=1$, satisfying
$\SPn{\phi_d}{X_\ve^d\,\phi_d}<\inf\spec(X_\ve^d)+\rho$.
Then 
$$
\SPb{\phi_d\otimes\Omega^f}{(|\DAe|-\gamma_\ve/|\V{x}|+\He)
\,\phi_d\otimes\Omega^f}\,=\,\SPn{\phi_d}{X_\ve^d\,\phi_d}
\,\klg\,\inf\spec(X_{\ve}^d)+\rho
$$
because of \eqref{valerie0}.
Moreover, we know from Lemma~\ref{le-Em-Eme}
that $E_{m,\ve}-E_m\grg o(\ve^0)/m$, $\ve\searrow0$.
Since $\Hef\,P_{\Omega^f}^\bot\grg m\,P_{\Omega^f}^\bot$
this implies that the term in \eqref{valerie3}
is non-negative also, provided that $\ve>0$ is
sufficiently small.

 In order to bound the remaining term in \eqref{valerie1}
from below we employ Corollary~\ref{cor-pia}
(with $\wt{\V{A}}=\V{0}$, $a=0$, $\vo=\omega_\ve$) 
and \eqref{valerie0} together with $\Hef\,P_{\Omega^f}=0$ to get
\begin{align*}
\big[\,|\DAed|&-\gamma_\ve/|\V{x}|+\Hed\,\big]\otimes P_{\Omega^f}
\\
&=\,
(\id\otimes P_{\Omega^f})\,
\big\{\,|\DAe|-\gamma_\ve/|\V{x}|+\He\,\big\}\,
(\id\otimes P_{\Omega^f})
\\
&\grg\,
\ve\,\big[\,|\DO|+|\V{x}|^2+\Hed\,\big]\otimes P_{\Omega^f}\,
-\,\big(\,C(\ve,\gamma,e^2,\UV)+\ve\,|\V{x}|^2\,\big)\otimes P_{\Omega^f}\,,
\end{align*}
for all sufficiently small values of $\ve>0$.
Since $\chi$ is exponentially localized we further know
that $T_2:=\chi\,\{|\V{x}|^2\otimes P_{\Omega_f}\}\,\chi$ is a bounded
operator.
Therefore, we arrive at
\begin{eqnarray}
\lefteqn{\nonumber
\chi\,\big\{\,\PFmg{\gamma}-E_m-m/2\,\big\}\,\chi\,+\,c'(\ve)\,(T_1+T_2)
}
\\
&\grg&\nonumber
\chi\,\big\{\,\big[\,\ve\,|\DO|+\ve\,|\V{x}|^2+\ve\,\Hed\,
-\,C'(\ve,\gamma,e^2,\UV)\,\big]
\otimes P_{\Omega^f}\,\big\}\,\chi
\\
&\grg&\label{valerie4}
\chi\,\big\{\,\big[\,\ve\,|\DO|+\ve\,|\V{x}|^2
+\ve\,\Hed-C'(\ve,\gamma,e^2,\UV)\,\big]_-\otimes P_{\Omega^f}\,\big\}\,\chi
\,,
\end{eqnarray}
where $[\cdots]_-\klg0$ denotes the negative part.
Now, both $|\DO|+|\V{x}|^2$ and $\Hed$ have purely discrete
spectrum as operators on the electron and photon Hilbert
spaces and $P_{\Omega^f}$, of course, has rank one. 
(Recall that $\Hed$ is the restriction of the discretized
field energy to the Fock space modeled over the
``$\ell^2$-space'' $\HP_m^d$.)
In particular, we observe that
$$
W^-_{m,\ve}\,:=\,
\big[\,\ve\,|\DO|+\ve\,|\V{x}|^2
+\ve\,\Hed-C'(\ve,\gamma,e^2,\UV)\,\big]_-\otimes P_{\Omega^f}
$$
is a finite rank operator,
for every sufficiently small $\ve>0$.

We can now conclude the proof as follows:
Given some sufficiently small $m>0$ 
we choose $\ve>0$ small enough
such that, in particular, the terms in \eqref{valerie2}\&\eqref{valerie3}
are non-negative, $\gamma_\ve/(1-c(\ve))<2/\pi$,
and $c'(\ve)\,(\|T_1\|+\|T_2\|)\klg m/8$.
Since by definition \eqref{def-chi-PF}
it holds $\chi\,\{\PFmg{\gamma}-E_m-m/2\}\,\chi\klg-(m/4)\,\chi$,
we see that the left hand side of \eqref{valerie4}
is bounded from above by $-(m/8)\,\chi$, whence
$$
-(m/8)\,\id_{E_m+m/4}(\PFmg{\gamma})\,\grg\,\chi\,W^-_{m,\ve}\,\chi\,.
$$
In particular, $\id_{(-\infty\,,\,E_m+m/4]}(\PFmg{\gamma})$ is a finite rank
projection.
\end{proof}

%%%%%%%%%%%%%%%%%%%%%%%%%%%%%%%%%%%%%%%%%%%%%%%%%%%%%%%%%%%%%%%%%%%%%
%%%%%%%%%%%%%%%%%%%%%%%%%%%%%%%%%%%%%%%%%%%%%%%%%%%%%%%%%%%%%%%%%%%%%
%%%%%%%%%%%%%%%%%%%%%%%%%%%%%%%%%%%%%%%%%%%%%%%%%%%%%%%%%%%%%%%%%%%%%

\section{Infra-red bounds}
\label{sec-IR-bounds}

\noindent
In this section we derive two key ingredients
we have used to prove the existence of ground states,
namely the soft photon and photon derivative bounds.
Soft photon bounds without infra-red regularization
have been derived in non-relativistic
QED first in \cite{BFS1999}.
We establish a soft photon bound for our
non-local model
by adapting an alternative argument from \cite{GLL2001}
where also the photon derivative bounds have been introduced.
In order to obtain these two infra-red bounds it is crucial
that the Hamiltonians $\PFm{\gamma}$ are gauge invariant.
For it has been observed in \cite{BFS1999}
that a suitable gauge transformation results in a
better infra-red behavior of the transformed vector potential.
(More precisely, it has been pointed out in \cite{GLL2001} that
the procedure from \cite{BFS1999} implicitly makes use
of a gauge transformation.)
Without the gauge transformation one would end up
with a bound in terms of infra-red divergent integrals.

This section is divided into four subsections.
In the first one we introduce the gauge transformation
mentioned above and prove some preparatory lemmata.
In Subsections~\ref{ssec-spb} and~\ref{ssec-pdb}
we prove the soft photon and photon derivative bounds,
respectively.
Some technical lemmata used in these two subsections 
are postponed to Subsection~\ref{ssec-katja}.

\subsection{The gauge transformed operator}\label{ssec-gauge}

\noindent
To start with we recall that,
for $i,j\in\{1,2,3\}$, the components $A^{(i)}_m(\V{x})$ and $A^{(j)}_m(\V{y})$
of the magnetic vector potential at $\V{x},\V{y}\in\RR^3$
commute in the sense that all their spectral projections
commute; see, e.g., 
\cite[Theorem~X.43]{ReedSimonII}. 
Therefore, it makes sense to introduce the following
operator-valued gauge transformation as in \cite{GLL2001},
\begin{equation*}
U\,:=\,\sum_{\vs=1,2,3,4}\int_{\RR^3}^\oplus U_\V{x}\,d\V{x}\,,\qquad
U_\V{x}\,:=\,\prod_{j=1}^3 e^{ix_jA_m^{(j)}(\V{0})}\,,\;\;\;
\V{x}=(x_1,x_2,x_3)\in\RR^3\,,
\end{equation*}
so that
\begin{equation}\label{laura0}
[U\,,\,\valpha\cdot\V{A}_m]\,=\,0\,.
\end{equation}
The gauge transformed vector potential is given by
\begin{equation*}
\wt{\V{A}}_m\,:=\,\V{A}_m\,-\,\id\otimes\V{A}_m(\V{0})
\,=\,\sum_{\vs=1,2,3,4}\int_{\RR^3}^\oplus
\valpha\cdot\big(\ad(\wt{\V{g}}_{\V{x}})+a(\wt{\V{g}}_{\V{x}})\big)\,d^3\V{x}\,,
\end{equation*}
where
\begin{equation*}
\wt{\V{g}}_{\V{x}}(k)\,:=\,\id_{\cA_m}(k)\,(e^{i\V{k}\cdot\V{x}}-1)\,\V{g}(k)\,,
\qquad \V{x}\in\RR^3\,,\;\textrm{a.e.}\;k=(\V{k},\lambda)\in\RR^3\times\ZZ_2\,,
\end{equation*}
and $\V{g}$ is defined in \eqref{def-Gphys}.
In fact, using \eqref{laura0} we deduce that
\begin{equation*}
U\,\DAm\,U^*\,=\,\DAmt\,,\qquad U\,\SAm\,U^*\,=\,\SAmt\,,\qquad
U\,|\DAm|\,U^*\,=\,|\DAmt|\,.
\end{equation*}
Then the key observation \cite{BFS1999} is that since
\begin{equation}\label{laura1}
|\wt{\V{g}}_{\V{x}}(k)|\,\klg\,
\id_{|\V{k}|>m}\,|\V{k}|\,|\V{x}|\,|\V{g}(k)|\,,
\qquad \V{x}\in\RR^3\,,\;\textrm{a.e.}\;k\in\RR^3\times\ZZ_2\,,
\end{equation}
the transformed vector potential $\wt{\V{A}}_m$
has a better infra-red behavior than $\V{A}_m$. 
In particular, infra-red divergent (for $m\searrow0$)
integrals appearing in the derivation of the soft photon bound
are avoided when we work with $\wt{\V{A}}_m$ instead of $\V{A}_m$.
It is needless to say that the gauge invariance of 
$\PFm{\gamma}$ is crucial at this point. 
The price to pay is that we have to control the unbounded
multiplication operator $|\V{x}|$ in \eqref{laura1}.
This is, however, possible thanks to the localization estimates
recalled in Proposition~\ref{prop-exp-loc}.

Below we shall use the following simple observations.
We pick some
orthonormal basis, $\{e_\ell:\,\ell\in\NN\}$,
of $\HP$ and some $q\in C_0((\RR^3\setminus\{0\})\times\ZZ_2)$.
Then
Fubini's theorem, Parseval's formula,
and the inequality 
$\|\valpha\cdot \V{z}\|^2\,\klg2|\V{z}|^2$, $\V{z}\in\CC^3$, 
imply
\begin{align}
\sum_{\ell\in\NN}&\nonumber
\big\|\,\valpha\cdot
\SPn{\wt{\V{g}}_\V{x}\,e^{-F}}{\omega^{-\nu}\,q\,e_\ell}\,\psi\,\big\|^2
\\
&\klg\,\nonumber
2\int_{\RR^3_\V{x}}\sum_{\ell\in\NN}
|\SPn{\omega^{-\nu}\,\ol{q}\,\wt{\V{g}}_\V{x}\,e^{-F(\V{x})}}{e_\ell}|^2_{\CC^3}
\,\|\psi\|_{\CC^4\otimes\Fock}^2(\V{x})\,d\V{x}
\\
&\klg\,\label{katja1}
C_F\,\|\omega^{1-\nu}\,q\,\V{g}\|^2\,\|\psi\|^2\,,\qquad \psi\in\HR_4\,,
\end{align}
where $\nu\in\RR$, and $F\in C^\infty(\RR^3,[0,\infty))$
is equal to $a|\V{x}|$, for large values of $|\V{x}|$
and some $a>0$.
In \eqref{katja1} and henceforth $C_F$ denotes some constant
which only depends on the choice of $F$ and whose value
might change from one estimate to another.
Moreover, we used that 
$|\omega^{-\nu}\,\wt{\V{g}}_{\V{x}}|\klg\omega^{1-\nu}\,|\V{g}|\,|\V{x}|$
and we simply wrote $\valpha\cdot
\SPn{\wt{\V{g}}_\V{x}\,e^{-F}}{\omega^{-\nu}\,q\,e_\ell}$ instead
of $\sum_{\vs=1,2,3,4}\int_{\RR^3}^\oplus\valpha\cdot
\SPn{\wt{\V{g}}_\V{x}\,e^{-F(\V{x})}}{\omega^{-\nu}\,q\,e_\ell}\,d^3\V{x}$ in the
first line. This slight abuse of notation will be maintained throughout
the whole section and should cause no confusion.
Setting
\begin{equation}\label{def-Deltah}
(\Delta_{\V{h}}f)(\V{k},\lambda)\,:=\,
f(\V{k}+\V{h},\lambda)\,-\,f(\V{k},\lambda)\,,
\qquad
\V{k},\V{h}\in\RR^3\,,\;\lambda\in\ZZ_2\,,
\end{equation}
for every $f\in\HP$,
so that
\begin{equation}\label{Green-Deltah}
\SPn{\Delta_{\V{h}}f_1}{f_2}\,=\,
\SPn{f_1}{\Delta_{-\V{h}}f_2}\,,\qquad f_1,f_2\in\HP\,,
\end{equation}
we further have
\begin{equation}\label{katja1D}
\sum_{\ell\in\NN}\big\|\,\valpha\cdot
\SPn{
\wt{\V{g}}_\V{x}\,e^{-F}}{\Delta_{\V{h}}(\omega^{-\nu}\,q\,e_\ell)}
\,\psi\,\big\|^2
\,\klg\,J_q^\nu(\V{h})
\,\|\psi\|^2\,,\qquad \psi\in\HR_4\,,
\end{equation}
where $\nu\in\RR$ and
\begin{equation}\label{def-Jqvkh}
J_q^\nu(\V{h})\,
:=\,2\int\frac{|q(k)|^2}{|\V{k}|^{2\nu}}\,\sup_{\V{x}\in\RR^3}
\big\{|\Delta_{-\V{h}}\wt{\V{g}}_\V{x}(k)|^2
\,e^{-2F(\V{x})}\big\}\,dk\,.
\end{equation}

\begin{lemma}\label{le-a(f)-comm}
Let $m\grg0$, $y\in\RR$, and let $f\in\HP$ such that $\omega^{-1/2}\,f\in\HP$. 
Let $F\in C^\infty(\RR^3_\V{x},[0,\infty))$ satisfy 
$F(\V{x})=a|\V{x}|$, for large $|\V{x}|$ and some $a\in(0,1)$,
and $|\nabla F|\klg a$
and set $L:=i\valpha\cdot\nabla F$.
Then (recall the notation \eqref{def-HT-RAL})
\begin{align}
[a(f)\,,\,\valpha\cdot\wt{\V{A}}_m]\,\phi\label{fritz1}
\,&=\,\valpha\cdot\SPn{f}{\wt{\V{g}}_\V{x}}\,\phi\,,
\\\label{fritz1b}
[\ad(f)\,,\,\valpha\cdot\wt{\V{A}}_m]\,\phi
\,&=\,-\valpha\cdot\SPn{\wt{\V{g}}_\V{x}}{f}\,\phi\,,
\\ \label{fritz2}
[\R{\wt{\V{A}}_m,L}{iy}\,,\,a(f)]\,\psi
\,&=\,
\R{\wt{\V{A}}_m,L}{iy}\,
\valpha\cdot\SPn{f}{\wt{\V{g}}_\V{x}}\,e^{-F}\,
\R{\wt{\V{A}}_m,2L}{iy}\,e^F\,\psi\,,
\\\label{fritz2b}
[\ad(f)\,,\,\R{\wt{\V{A}}_m,L}{iy}]\,\psi
\,&=\,
\R{\wt{\V{A}}_m,L}{iy}\,
\valpha\cdot\SPn{\wt{\V{g}}_\V{x}}{f}\,
e^{-F}\,\R{\wt{\V{A}}_m,2L}{iy}\,e^F\,\psi\,,
\end{align}
for all $\phi\in\dom(\Hf)$ and $\psi\in\dom(\Hf^{1/2})$.
Moreover, let $\{e_\ell:\,\ell\in J\}$ be an
orthonormal system in $\HP$.
Then we have, for all $q\in C_0((\RR^3\setminus\{0\})\times\ZZ_2)$, 
$\kappa\in[0,1)$,
and $\nu\in\RR$,
\begin{align}
\sum_{\ell\in J}
\big\|\,|\DAmt|^{\kappa}\,
[\SAmt\,,\,a^\sharp(\omega^{-\nu}\,q\,e_\ell)]\,e^{-F}\,\big\|^2
\,&\klg\,C_{F,\kappa}\,\label{fritz3}
\|\omega^{1-\nu}\,\ol{q}\,\V{g}\|^2\,,
\end{align}
and
\begin{align}
\sum_{\ell\in J}
\big\|\,|\DAmt|^{\kappa}\,
\big[\SAmt\,,\,a^\sharp
\,\big(\Delta_{\V{h}}(\omega^{-\nu}\,q\,e_\ell)\big)\,\big]\,e^{-F}\,\big\|^2
\,&\klg\,\label{fritz3D}
C_\kappa\,J_q^{\nu}(\V{h})\,,\quad \V{h}\in\RR^3\,.
\end{align}
\end{lemma}

\begin{proof}
We drop the subscript $m$ in this proof.
Of course \eqref{fritz1} and \eqref{fritz1b} follow immediately
from the canonical commutation relations and \eqref{fritz2}
and \eqref{fritz2b} are easy consequences.
\eqref{fritz2}
and \eqref{fritz2b}
together with \eqref{sgn} 
permit to get (the superscript $\sharp$ denotes complex
conjugation when $a^\sharp$ is $\ad$ and has to be ignored
when $a^\sharp$ is $a$)
\begin{align*}
\big|&\SPb{|\DAt|^{\kappa}\,\vp}{[\SAt\,,\,a^\sharp(f)]
\,e^{-F}\,\eta}\big|
\\
&\klg\,
\int_\RR\Big|\SPB{|\DAt|^\kappa\,\vp}{\RAt{iy}\,
\valpha\cdot\SPn{f}{\wt{\V{g}}_{\V{x}}}^\sharp\,e^{-F}\,\R{\wt{\V{A}},L}{iy}\,
\eta}\Big|\,\frac{dy}{\pi}
\\
&\klg\,
\|\vp\|\,
\Big(\int_\RR\frac{\|\,|\DA|^\kappa\,\RA{iy}\|^2}{(1+y^2)^{\kappa/2}}
\,\frac{dy}{\pi}\Big)^{1/2}
\\
&\qquad\cdot
\Big(\int_\RR (1+y^2)^{\kappa/2}\,\big\|\,
\valpha\cdot\SPn{f}{\wt{\V{g}}_\V{x}\,e^{-F}}^\sharp
\,\R{\wt{\V{A}},L}{iy}\,\eta
\,\big\|^2\,\frac{dy}{\pi}\Big)^{1/2}
\,,
\end{align*}
for all $\vp,\eta\in\core$.
Inserting $f=\omega^{-\nu}\,q\,e_\ell$ and summing 
the squares of the resulting inequalities
with respect
to $\ell$ we thus obtain
\begin{align*}
\sum_{\ell\in J}&
\big\|\,|\DAt|^{\kappa}\,
[\SAt\,,\,a^\sharp(\omega^{-\nu}\,q\,e_\ell)]\,e^{-F}\,\eta\,\big\|^2
\\
&\klg
\,C_\kappa'\,
\int_\RR\sum_{\ell\in J}
\big\|\,\valpha\cdot\SPn{\omega^{-\nu}\,q\,e_\ell}{
\wt{\V{g}}_{\V{x}}\,e^{-F}}^\sharp\,
\R{\wt{\V{A}},L}{iy}
\,\eta\,\big\|^2\,\frac{(1+y^2)^{\kappa/2}\,dy}{\pi}\,,
\end{align*}
for every $\eta\in\core$. Now \eqref{fritz3} follows from the
previous estimate
in combination with \eqref{katja1} and \eqref{marah1}.
If we replace $\omega^{-\nu}\,q\,e_\ell$ by
$\Delta_\V{h}(\omega^{-\nu}\,q\,e_\ell)$ in the above
argument and apply \eqref{katja1D} instead of \eqref{katja1}
then we also obtain \eqref{fritz3D}.
\end{proof}

%%%%%%%%%%%%%%%%%%%%%%%%%%%%%%%%%%%%%%%%%%%%%%%%%%%%%%%%%%%%%%%%%%%%%%%%%%%%%%

\subsection{Soft photon bound}\label{ssec-spb}

\noindent
Now assume that $\phi_m$ is a normalized ground state eigenvector
of the semi-relativistic Pauli-Fierz operator $\PFm{\gamma}$
and set
\begin{equation*}
E_m\,:=\,\inf\spec[\PFm{\gamma}]\,,\qquad
\wt{\phi}_m\,:=\,U\,\phi_m\,,\qquad \Hft\,\equiv\,
\wt{H}_{f,m}\,:=\,U\,\Hf\,U^*\,,
\end{equation*}
and
\begin{equation*}
\PFt{\gamma}\,\equiv\,\wt{H}_{\gamma,m}\,:=\,U\,\PF{\gamma}\,U^*
\,=\,
|\DAmt|-\tgV+\Hft\,.
\end{equation*}
Differentiating with respect to $\V{x}$ we verify that
($\V{g}_m:=\id_{\cA_m}\,\V{g}$)
\begin{equation}\label{klaus1}
[U_\V{x}\,,\,a(f)]\,=\,-i\SPn{f}{\V{g}_m\cdot\V{x}}\,U_\V{x}\,,\qquad
[U_\V{x}^*\,,\,a(f)]\,=\,i\SPn{f}{\V{g}_m\cdot\V{x}}\,U_\V{x}^*\,.
\end{equation}
For instance, both sides of the left identity are solutions of
the initial value problem 
$\nabla_\V{x}T(\V{x})=i\V{A}_m(\V{0})\,T(\V{x})
-i\SPn{f}{\V{g}_m}\,U_\V{x}$, $T(\V{0})=0$.
Moreover, we observe that \eqref{klaus1} gives
\begin{equation}\label{klaus2}
[\Hft\,,\,a(f)]\,=\,-a(\omega\,f)\,+\,i\SPn{\omega\,f}{\V{g}_m\cdot\V{x}}\,.
\end{equation}

\smallskip

\begin{proof}[Proof of Proposition~\ref{prop-spb}]
We drop all subscripts $m$ in this proof.
We proceed along the lines of the proof presented in
\cite[Appendix~B]{GLL2001}. The new complication comes
from the terms involving the non-local operator $|\DA|$.
To begin with we recall that,
by Fubini's theorem and Parceval's identity,
\begin{equation}\label{Parceval-a}
\sum_{\ell\in\NN}\SPb{a(f\,e_\ell)\,\psi}{a(h\,e_{\ell})\,\eta}\,=\,
\int f(k)\,\ol{h(k)}\,\SPb{a(k)\,\psi}{a(k)\,\eta}\,dk\,,
\end{equation}
for $f,h\in\HP\cap L^\infty$ and $\psi,\eta\in\dom(\Hf^{1/2})$.

Let $q\in C_0((\RR^3\setminus\{0\})\times\ZZ_2)$. (In the end we shall
insert a family of approximate delta-functions for $q$.)
Moreover, we
let $\{e_\ell:\,\ell\in\NN\}$ denote some orthonormal basis
of $\HP$ and assume that the weight function
$F\in C^\infty(\RR^3_\V{x},[0,\infty))$ is equal to
$a|\V{x}|$, for all $\V{x}\in\RR^3$, $|\V{x}|\grg R$, and
some $R>0$ such that $|\nabla F|\klg a$ on $\RR^3$.
Here $a\in(0,1/2)$ is assumed to be so small that the bound
\eqref{exp-loc-sup} is available with $a$ replaced by $2a$.
Accordingly we shall always assume that $m\in(0,m_1]$, where
$m_1$ is the parameter appearing in Proposition~\ref{prop-exp-loc}.
Together with \eqref{klaus1}
the identity \eqref{Parceval-a} then implies
\begin{align}
\int|&q(k)|^2\,\|a(k)\,\phi\|^2\,dk\,\nonumber
=\,\sum_{\ell\in\NN}\big\|\,U\,a(q\,e_\ell)\,\phi\,\big\|^2
\\
&\klg\,\nonumber
2\sum_{\ell\in\NN}
\big\|\,a(q\,e_\ell)\,\wt{\phi}\,\big\|^2
\,+\,
2\sum_{\ell\in\NN} \sup_{\V{x}\in\RR^3}
\big|\SPb{q\,e_\ell}{\V{g}\cdot\V{x}\,e^{-F(\V{x})}}\big|^2
\,\|e^F\,\phi\|^2
\\
&\klg\,\label{nora-88}
2\int|q(k)|^2\,\|a(k)\,\wt{\phi}\|^2\,dk
\,+\,
2C_F\int\id_{|\V{k}|\klg\UV}\frac{|q(k)|^2}{|\V{k}|}\,dk\,,
\end{align}
for $m\in(0,m_1]$. Here $C_F\in(0,\infty)$ depends
on $F$ and on the quantity in \eqref{exp-loc-sup},
but not on $m$.
In what follows we derive a bound on the left term
in the last line of \eqref{nora-88}. 
To this end we pick some
$\psi\in U\,\core$ and some $f\in\HP$ with compact support
in $(\RR^3\setminus\{0\})\times\ZZ_2$. 
Writing $|\DAt|=\DAt\,\SAt$ and employing
the eigenvalue equation for $\wt{\phi}$ we deduce that
\begin{eqnarray*}
\SPb{(\PFmt{\gamma}-E_m)\,\psi}{a(f)\,\wt{\phi}}
&=&
\SPb{[\ad(f)\,,\,\PFmt{\gamma}-E_m]\,\psi}{\wt{\phi}}
\\
&=&
\SPb{[\ad(f)\,,\,\valpha\cdot\wt{\V{A}}]\,\SAt\,\psi}{\wt{\phi}}
\\
& &
\;+\,
\SPb{|\DAt|^{1/2}\,[\ad(f)\,,\,\SAt]\,\psi}{\SAt\,|\DAt|^{1/2}\,\wt{\phi}}
\\
& &
\;+\,
\SPb{[\ad(f)\,,\,\Hft]\,\psi}{\wt{\phi}}
\,.
\end{eqnarray*}
By Lemma~\ref{le-a(f)-comm} and \eqref{klaus2}
we may replace $\psi\in U\,\core$
on the right and left hand sides of the previous identity
by any element of $\form(\PFmt{\gamma})$ and
in Appendix~\ref{app-a(f)} 
we verify that $a(f)\,\wt{\phi}\in\form(\PFmt{\gamma})$.
On account of \eqref{klaus2} and $\PFmt{\gamma}-E_m\grg0$ we thus get
\begin{eqnarray}
\SPb{a(f)\,\wt{\phi}}{a(\omega\,f)\,\wt{\phi}}&\klg&
-\SPb{[\SAt\,,\,a(f)]\,\wt{\phi}}{\valpha\cdot\SPn{f}{\wt{\V{g}}_{\V{x}}}\,
\wt{\phi}}\nonumber
\\
& &\;-\,\SPb{a(f)\,\SAt\,\wt{\phi}}{\valpha\cdot\SPn{f}{\wt{\V{g}}_\V{x}}
\,\wt{\phi}}\nonumber
\\
& &
\;+\,\nonumber
\SPb{|\DAt|^{1/2}\,
[\ad(f)\,,\,\SAt]\,a(f)\,\wt{\phi}}{\SAt\,|\DAt|^{1/2}\,\wt{\phi}}
\\
& &\label{nora-99}
\;+\,i
\SPb{a(f)\,\wt{\phi}}{\SPn{\omega\,f}{\V{g}\cdot\V{x}}\,\wt{\phi}}
\,.
\end{eqnarray}
Next, we substitute $f$ by $f_\ell:=\omega^{-1/2}\,q\,e_\ell$ and sum
with respect to $\ell$.
On account of \eqref{Parceval-a}
this results in
\begin{align}
\int|q(k)|^2\,\|&a(k)\,\wt{\phi}\|^2\,dk
\,\klg\,
\Big|\,\sum_{\ell\in\NN}
\SPb{[\SAt\,,\,a(f_\ell)]\,\wt{\phi}}{
\valpha\cdot\SPn{f_\ell}{\wt{\V{g}}_{\V{x}}}\,\wt{\phi}}
\,\Big|\label{franz0}
\\
& \label{franz0b}
\quad+\,\Big|\,\sum_{\ell\in\NN}
\SPb{a(f_\ell)\,\SAt\,\wt{\phi}}{
\valpha\cdot\SPn{f_\ell}{\wt{\V{g}}_{\V{x}}}\,\wt{\phi}}
\,\Big|
\\
& \label{franz1}
\quad+\,
\Big|\,\sum_{\ell\in\NN}
\SPb{|\DAt|^{1/2}\,[\ad(f_\ell)\,,\,\SAt]\,a(f_\ell)\,\wt{\phi}}{
\SAt\,|\DAt|^{1/2}\,\wt{\phi}}\,\Big|
\\
& \label{franz2}
\quad+\,
\Big|\,\sum_{\ell\in\NN}
\SPb{a(f_\ell)\,\wt{\phi}}{\SPn{\omega\,f_\ell}{\V{g}\cdot\V{x}}\,\wt{\phi}}
\,\Big|
\,.
\end{align}
By means of the Cauchy-Schwarz inequality, \eqref{katja1},
and \eqref{fritz3} we deduce the following bound on
the term on the right side of \eqref{franz0},
\begin{align}
\Big|\sum_{\ell\in\NN}
\SPb{[\SAt\,,\,a(f_\ell)]\,\wt{\phi}}{
\valpha\cdot\SPn{f_\ell}{\wt{\V{g}}_\V{x}\,e^{-F}}\,e^F\,
\wt{\phi}}\,\Big|\,\klg\,
\label{franz4}C_F\,\|\omega^{1/2}\,\ol{q}\,\V{g}\|^2
\|e^F\,\wt{\phi}\|^2.
\end{align}
Furthermore, we observe that
\begin{equation}\label{katja2b}
\sum_{\ell\in\NN}\SPn{\V{g}\cdot\V{x}\,e^{-F}}{\omega\,f_\ell}\,a(f_\ell)\,\psi
\,=\,e^{-F}\,\V{x}\cdot a\big(|q|^2\,\V{g}\big)\,\psi\,,
\qquad \psi\in\dom(\Hft^{1/2})\,.
\end{equation}
Similarly as in \cite{GLL2001} we employ \eqref{katja2b}
to estimate the term
in \eqref{franz2} as
\begin{align}
\Big|\sum_{\ell\in\NN}
\SPb{&a(f_\ell)\,\wt{\phi}}{\SPn{\omega\,f_\ell}{\V{g}\cdot\V{x}\,e^{-F}}
\,e^F\,\wt{\phi}}\,\Big|\nonumber
\\
&\klg\,
\Big|
\int|q(k)|^2\,
\SPb{a(k)\,\wt{\phi}}{(\V{g}(k)\cdot\V{x}\,e^{-F})\,e^F\,\wt{\phi}}\,dk\,
\Big|\nonumber
\\
&\klg\,
\frac{\delta}{2}\int |q(k)|^2\,\|a(k)\,\wt{\phi}\|^2\,dk
\,+\,
\frac{C_F}{\delta}
\int\id_{|\V{k}|<\UV}\,\frac{|q(k)|^2}{|\V{k}|}\,dk\,\|e^F\,\wt{\phi}\|^2
\,,\label{franz44}
\end{align}
for every $\delta\in(0,1]$.
Here we also used that $|\V{g}(k)\cdot\V{x}|^2\klg|\V{x}|^2/|\V{k}|$.
The terms in \eqref{franz0b} and \eqref{franz1} are treated
in Lemmata~\ref{le-lora} and~\ref{le-katja1} below, where we show
that their sum is bounded from above by
$$
\frac{\delta}{2}\,\int|q(k)|^2\,\|a(k)\,\wt{\phi}\|^2\,dk
\,+\,\frac{C''}{\delta}\,\int
\Big(|\V{k}|+\frac{1}{|\V{k}|}\Big)\,\id_{|\V{k}|\klg\UV}
\,|q(k)|^2\,dk\,,
$$
for some $C''\in(0,\infty)$ and every $\delta\in(0,1/2]$.
Putting all the estimates above together, we arrive at
\begin{align}\label{spb-wtphi}
(1-\delta)\int|q(k)|^2\,\|a(k)\,\wt{\phi}\|^2\,dk
&\,\klg\,
\frac{C'''}{\delta}\int\Big(|\V{k}|+\frac{1}{|\V{k}|}\Big)
\,\id_{|\V{k}|\klg\UV}\,|q(k)|^2\,dk\,,
\end{align}
for every $\delta\in(0,1/2]$.
Here the constant $C'''\in(0,\infty)$  
does not depend on $m\in(0,m_1]$.
Combining \eqref{spb-wtphi} with \eqref{nora-88} and
peaking at some $k$ by inserting an appropriate
family of approximate delta-functions for $q$,
we obtain the asserted estimate \eqref{eq-spb}.
\end{proof}

%%%%%%%%%%%%%%%%%%%%%%%%%%%%%%%%%%%%%%%%%%%%%%%%%%%%%%%%%%%%%%%%%%%%%%%%%%

\subsection{Photon derivative bound}\label{ssec-pdb}

\noindent
In this subsection we make use of the particular choice \eqref{pol-vec}
of the polarization vectors.
In the following we use the abbreviations
\begin{eqnarray*}
k+\V{h}&:=&(\V{k}+\V{h},\lambda)\,,\qquad
(\Delta_{\V{h}} f)(k)\,:=\,
f(k+\V{h})-f(k)\,,
\\
(\Delta_{-\V{h}}a)(f)&:=&a(\Delta_\V{h}f)\,,
\qquad\;\;(\Delta_{-\V{h}}a)(k)\,:=\,
a(k-\V{h})-a(k)\,,
\end{eqnarray*}
where $\V{h}\in\RR^3$, $k=(\V{k},\lambda)\in\RR^3\times\ZZ_2$, 
and $f\in\HP$,
so that
\begin{eqnarray*}
\SPn{\Delta_{\V{h}}f_1}{f_2}&=&
\SPn{f_1}{\Delta_{-\V{h}}f_2}\,,
\\
(\Delta_{-\V{h}}a)(f)&=&\int\ol{f(k)} \,(\Delta_{-\V{h}}a)(k)\,dk
\,.
\end{eqnarray*}

\smallskip

\begin{proof}[Proof of Proposition~\ref{prop-pdb}]
Most subscripts $m$ are dropped in this proof
so that $\phi\equiv\phi_m$, $\wt{\V{A}}\equiv\wt{\V{A}}_m$, etc.
Again,
we carry through a procedure presented in 
\cite[Appendix~B]{GLL2001} and the new difficulty
is how to deal with the non-local term in $\PFmt{\gamma}$.

First, we pick some orthonormal basis, $\{e_\ell:\,\ell\in\NN\}$,
of $\HP$ and observe that
\begin{equation}\label{Parceval-Deltaha}
\sum_{\ell\in\NN}\SPb{\Delta_{-\V{h}}a(f\,e_\ell)\,\psi}{
\Delta_{-\V{h}}a(h\,e_{\ell})\,\eta}=
\int \!f(k)\,\ol{h(k)}\,\SPb{\Delta_{-\V{h}}a(k)\,\psi}{
\Delta_{-\V{h}}a(k)\,\eta}\,dk,
\end{equation}
for all $f,h\in \HP\cap L^\infty$,
in analogy to \eqref{Parceval-a}.
Similarly to \eqref{nora-88} we employ
\eqref{klaus1} and \eqref{Parceval-Deltaha} to get
\begin{align}
\int|&q(k)|^2\,\big\|\,(\Delta_{-\V{h}}a)(k)\,\phi\,\big\|^2\,dk
\,=\,
\sum_{\ell\in\NN}\big\|\,U\,a\big(\Delta_{\V{h}}(q\,e_\ell)\big)\,\big\|^2
\nonumber
\\
&\klg\,\nonumber
2\sum_{\ell\in\NN}\big\|a\big(\Delta_{\V{h}}(
q\,e_\ell)\big)\,\wt{\phi}\big\|^2
+
2\sum_{\ell\in\NN}
\sup_{\V{x}\in\RR^3}\big|\SPb{e_\ell}{\ol{q}\,\Delta_{-\V{h}}\V{g}\cdot\V{x}\,
e^{-F(\V{x})}}\big|^2\,\|e^F\,\phi\|^2
\\
&\klg\label{nora-88b}
\,2\int|q(k)|^2\,\big\|\,(\Delta_{-\V{h}}a)(k)\,\wt{\phi}\,\big\|^2\,dk
+
2C_F^2\int|q(k)|^2\,|\Delta_{-\V{h}}\V{g}(k)|^2\,dk\,,
\end{align}
for every $q\in C_0((\RR^3\setminus\{0\})\times\ZZ_2)$.
In the following we seek for a bound on
$$
I_q(\V{h})\,:=\,\int|q(k)|^2\,
\big\|\,(\Delta_{-\V{h}}a)(k)\,\wt{\phi}\,\big\|^2\,dk
$$
and pick some $f\in\HP$ with compact support in 
$(\RR^3\setminus\{0\})\times\ZZ_2$.
Then we clearly have
\begin{align}
\nonumber
\SPb{\Delta_{-\V{h}}&a(f)\,\wt{\phi}}{\Delta_{-\V{h}}a(\omega\,f)\,\wt{\phi}}
\\
&=\,
\SPb{a(\Delta_\V{h}f)\,\wt{\phi}}{a(\omega\,\Delta_\V{h}f)\,\wt{\phi}}
\,+\,
\SPb{\Delta_{-\V{h}}a(f)\,\wt{\phi}}{a\big((\Delta_\V{h}\omega)\,
f(\cdot+\V{h})\big)\,\wt{\phi}}\,.\label{nora1}
\end{align}
Moreover, we use 
the eigenvalue equation for $\wt{\phi}$ and an argument
analogous to the one leading to \eqref{nora-99} to infer that
\begin{eqnarray}
0&\klg&\SPb{a(\Delta_\V{h}f)\wt{\phi}}{(\PFmt{\gamma}-E_m)
a(\Delta_\V{h}f)\,\wt{\phi}}\nonumber
\\
&=&
-\SPb{a(\Delta_\V{h}f)\,\wt{\phi}}{
\SAt\,\valpha\cdot\SPn{\Delta_\V{h}f}{\wt{\V{g}}_\V{x}}\,
\wt{\phi}}\nonumber
\\
& &\;+\,\nonumber
\SPb{|\DAt|^{1/2}\,[\ad(\Delta_\V{h}f)\,,\,\SAt]\,a(\Delta_\V{h}f)\,\wt{\phi}}{
\SAt\,|\DAt|^{1/2}\,\wt{\phi}}
\\
& &\label{nora4}
\,-\,\SPb{a(\Delta_\V{h}f)\,\wt{\phi}}{a(\omega\,\Delta_\V{h}f)\,\wt{\phi}}
\\
& &
\;+\,i\SPb{a(\Delta_\V{h}f)\,\wt{\phi}}{\SPn{\omega\,\Delta_\V{h}f}{
\V{g}\cdot\V{x}}\,\wt{\phi}}\,.\nonumber
\end{eqnarray}
When we add this inequality to \eqref{nora1}
the first term on the right hand side of \eqref{nora1} and the
term in \eqref{nora4} cancel each other and we obtain
\begin{eqnarray}
\lefteqn{\nonumber
\SPb{\Delta_{-\V{h}}a(f)\,\wt{\phi}}{\Delta_{-\V{h}}a(\omega\,f)\,\wt{\phi}}
}
\\
&\klg&
\SPb{\Delta_{-\V{h}}a(f)\,\wt{\phi}}{a\big((\Delta_\V{h}\omega)\,
f(\cdot+\V{h})\big)\,\wt{\phi}}
\label{nora5}
\\
& &\;
-\,\SPb{[\SAt\,,\,\Delta_{-\V{h}}a(f)]\,\wt{\phi}}{\valpha\cdot
\SPn{f}{\Delta_{-\V{h}}\wt{\V{g}}_\V{x}}\,
\wt{\phi}}\label{nora6a}
\\
& &\;
-\,\SPb{\valpha\cdot
\SPn{\Delta_{-\V{h}}\wt{\V{g}}_\V{x}}{f}\,
\Delta_{-\V{h}}a(f)\,\SAt\,\wt{\phi}}{\wt{\phi}}\label{nora6b}
\\
& &\;+\,\label{nora7}
\SPb{|\DAt|^{1/2}\,
[\ad(\Delta_\V{h}f)\,,\,\SAt]\,\Delta_{-\V{h}}a(f)\,\wt{\phi}}{
\SAt\,|\DAt|^{1/2}\,\wt{\phi}}
\\
& &
\;+\,i\SPb{\Delta_{-\V{h}}a(f)\,\wt{\phi}}{\SPn{f}{
\Delta_{-\V{h}}(\omega\,\V{g}\cdot\V{x})}\,\wt{\phi}}\,.\label{nora8}
\end{eqnarray}
We replace $f$ by
$
f_\ell=\omega^{-1/2}\,q\,e_\ell
$,
for some orthonormal basis $\{e_\ell:\,\ell\in\NN\}$
of $\HP$ and some $q\in C_0(\cA_m\times\ZZ_2)$,
and sum the previous estimate with respect to $\ell$.
Notice that, apart from the term in \eqref{nora5},
the previous estimate is an analogue of \eqref{nora-99}
with $a$ replaced by $\Delta_{-\V{h}}a$ or $f$
replaced by $\Delta_{\V{h}}f$.
Moreover, employing \eqref{Parceval-Deltaha} we find
\begin{equation}
\sum_{\ell\in\NN}\label{nora9}
\SPb{\Delta_{-\V{h}}a(f_\ell)\,\wt{\phi}}{
\Delta_{-\V{h}}a(\omega\,f_\ell)\,\wt{\phi}}
\,=\,
\int|q(k)|^2\,\|\Delta_{-\V{h}}a(k)\,\wt{\phi}\|^2\,dk\,.
\end{equation}
Furthermore, an analogue of \eqref{katja2b} reads
\begin{equation}\label{nora9b}
\sum_{\ell\in\NN}\SPn{\Delta_{-\V{h}}(\omega\,\V{g}\cdot\V{x})\,e^{-F}}{f_\ell}
\,\Delta_{-\V{h}}a(f_\ell)\,\psi\,=\,
\Delta_{-\V{h}}a\big(\omega^{-1}\,|q|^2\,
\Delta_{-\V{h}}(\omega\,\V{g}\cdot\V{x})\,e^{-F}\big)\,\psi\,,
\end{equation}
for $\psi\in\dom(\Hf^{1/2})$.
Here and henceforth we choose $F$ as described in the paragraph
succeeding Equation \eqref{Parceval-a}.
We may hence use the same line of arguments that led to
\eqref{franz44} in order to deduce that
\begin{align}
\Big|\sum_{\ell\in\NN}&i\SPb{\Delta_{-\V{h}}a(f_\ell)\,\wt{\phi}}{
\SPn{f_\ell}{\Delta_{-\V{h}}(\omega\,\V{g}\cdot\V{x})}\,\wt{\phi}}\,\Big|
\nonumber
\\
&\klg\,\label{kalli1}
\frac{\delta}{4}\,I_q(\V{h})
\,+\,
\frac{C_F}{\delta}
\int\frac{|q(k)|^2}{|\V{k}|^2}\,
\big|\,
\Delta_{-\V{h}}(\omega\,\V{g})(k)\,\big|^2\,dk\,.
\end{align}
Moreover, \eqref{exp-loc-sup}, \eqref{katja1D}, \eqref{fritz3D}, 
and the Cauchy-Schwarz inequality permit to get
\begin{align}
\Big|\sum_{\ell\in\NN}&
\SPb{[\SAt\,,\,\Delta_{-\V{h}}a(f_\ell)]\,\wt{\phi}}{\valpha\cdot
\SPn{f_\ell}{\Delta_{-\V{h}}\wt{\V{g}}_\V{x}\,e^{-F}}\,
e^{F}\wt{\phi}}
\Big|
\,\klg\,
C''\,J_q^{1/2}(\V{h})
\label{kalli2}\,.
\end{align} 
For the term appearing in \eqref{nora5} we find
by means of the soft photon bound (for $\wt{\phi}$ instead
of $\phi$; recall \eqref{spb-wtphi})
\begin{align}
\Big|\sum_{\ell\in\NN}&\nonumber
\SPb{\Delta_{-\V{h}}a(f_\ell)\,\wt{\phi}}{
a\big((\Delta_{\V{h}}\omega)\,f_\ell(\cdot+\V{h})\big)\,\wt{\phi}}\,\Big|
\\
&=\,\nonumber
\Big|\int
(\Delta_{\V{h}}\omega)(k-\V{h})\,\frac{|q(k)|^2}{\omega(k)}
\,\SPb{\Delta_{-\V{h}}a(k)\,\wt{\phi}}{a(k-\V{h})\,\wt{\phi}}\,dk
\,\Big|
\\
&\klg\,
\nonumber
\frac{\delta}{4}
\,I_q(\V{h})
\,+\,
\frac{1}{\delta}
\int|q(k)|^2\,
\frac{(\Delta_{\V{h}}\omega)(k-\V{h})^2}{\omega(k)^2}
\,\|a(k-\V{h})\,\wt{\phi}\|^2\,dk
\\
&\klg\,\label{kalli3}
\frac{\delta}{4}\,I_q(\V{h})\,+\,
\frac{C'''}{\delta}
\int|q(k)|^2\,
\frac{(\Delta_{\V{h}}\omega)(k-\V{h})^2}{|\V{k}|^2
\,|\V{k}-\V{h}|}\,\id_{|\V{k}-\V{h}|\klg\UV}\,dk
\,.
\end{align}
Finally, Lemmata~\ref{le-lora} and~\ref{le-katja1} below together assert that
the terms in \eqref{nora6b} and \eqref{nora7} can be estimated as
\begin{align}
\Big|&\sum_{\ell\in\NN}
\SPb{\valpha\cdot
\SPn{\Delta_{-\V{h}}\wt{\V{g}}_\V{x}}{f_\ell}\,
\Delta_{-\V{h}}a(f_\ell)\,\SAt\,\wt{\phi}}{\wt{\phi}}
\Big|\nonumber
\\
&\quad+\,
\Big|\sum_{\ell\in\NN}\SPb{
|\DAt|^{1/2}\,[\ad(\Delta_\V{h}f_\ell)\,,\,\SAt]\,\Delta_{-\V{h}}a(f_\ell)\,
\wt{\phi}}{\SAt\,|\DAt|^{1/2}\,\wt{\phi}}\,\Big|
\nonumber
\\
&\klg\,\label{kalli4}
\frac{\delta}{2}\,I_q(\V{h})
\,+\,
\frac{C^{(4)}}{\delta}
\,J_q^1(\V{h})\,.
\end{align}
Combining \eqref{nora-88b} and \eqref{nora5}--\eqref{nora8}
with \eqref{kalli1}--\eqref{kalli4} we arrive at
\begin{align}
\frac{1-\delta}{2}&
\int|q(k)|^2\,\|(\Delta_{-\V{h}}a)(k)\,\phi\|^2\,dk\nonumber
\\
&\klg\,\nonumber
\frac{C^{(5)}}{\delta}\int\frac{|q(k)|^2}{|\V{k}|^2}\,
\Big\{|\V{k}|^2\,|\Delta_{-\V{h}}\V{g}(k)|^2+
\big|\Delta_{-\V{h}}(\omega\,\V{g})(k)\big|^2\Big\}\,dk
\\
&\quad+\nonumber
\frac{C^{(5)}}{\delta}\int\frac{|q(k)|^2}{|\V{k}|^2}\,
\frac{(\Delta_{\V{h}}\omega)(k-\V{h})^2}{|k-\V{h}|}
\,\id_{|\V{k}-\V{h}|\klg\UV}\,dk
\\
&\quad+\,\label{lisel1}
\frac{C^{(5)}}{\delta}
\int\Big(1+\frac{1}{|\V{k}|^2}\Big)\,|q(k)|^2\,
\sup_{\V{x}\in\RR^3}\big\{|\Delta_{-\V{h}}\wt{\V{g}}_\V{x}(k)|^2
\,e^{-2F(\V{x})}\big\}\,dk\,,
\end{align}
for every $\delta\in(0,1]$ and $q\in C_0(\cA_m\times\ZZ_2)$.

As in \cite{GLL2001} we now employ the special choice of the
polarization vectors \eqref{pol-vec}
in order to bound the discrete derivatives of the previous estimate.
In fact, set
$\V{y}_\bot:=(y^{(2)},-y^{(1)},0)$
and $\V{y}^\circ:=\V{y}/|\V{y}|$, 
for $\V{y}=(y^{(1)},y^{(2)},y^{(3)})\in\RR^3\setminus\{\V{0}\}$.
Then
\begin{align*}
(\Delta_{-\V{h}}\veps)(\V{k},0)\,&=\,
-|\V{k}_\bot|^{-1}\,\V{h}_\bot+
\big(|(\V{k}-\V{h})_\bot|^{-1}-|\V{k}_\bot|^{-1}\big)\,(\V{k}-\V{h})_\bot\,,
\\
(\Delta_{-\V{h}}\veps)(\V{k},1)
\,&=\,
\big((\V{k}-\V{h})^\circ-{\V{k}}^\circ\big)\wedge\veps(\V{k}-\V{h},0)
+{\V{k}}^\circ\wedge(\Delta_{-\V{h}}\veps)(\V{k},0)\,,
\end{align*}
whence
\begin{align*}
|\Delta_{-\V{h}}\veps(\V{k},0)|\,&\klg\,
2|\V{h}_\bot|/|\V{k}_\bot|\,\klg\,2|\V{h}|/|\V{k}_\bot|\,,
\\
|\Delta_{-\V{h}}\veps(\V{k},1)|\,&\klg\,2|\V{h}|/|\V{k}|
+|\Delta_{-\V{h}}\veps(\V{k},0)|\,\klg\,4|\V{h}|/|\V{k}_\bot|\,.
\end{align*}
In the sequel we re-introduce the reference to $m$ in the notation.
Since $\V{g}_m(k)=|\V{k}|^{-1/2}\,\veps(k)\,\id_{m\klg|\V{k}|\klg\UV}$ and
$|a^{-1/2}-b^{-1/2}|\klg(|a-b|/2)(a^{-3/2}+b^{-3/2})$,
$a,b>0$,
we further have, for $m<|\V{k}|,|\V{k}-\V{h}|<\UV$,
\begin{align*}
\big|\Delta_{-\V{h}}\V{g}_m(k)\big|\,&\klg\,
\frac{4|\V{h}|}{|\V{k}|^{1/2}|\V{k}_\bot|}\,+\,
\frac{|\V{h}|}{2}\Big(\frac{1}{|\V{k}|^{3/2}}\,+
\,\frac{1}{|\V{k}-\V{h}|^{3/2}}\Big)\,,
\\
\frac{1}{|\V{k}|}\,\big|\Delta_{-\V{h}}(\omega\,\V{g}_m)\big|
\,&\klg\,\big|\Delta_{-\V{h}}\V{g}_m(k)\big|
+\,\frac{|\V{h}|}{|\V{k}|\,|\V{k}-\V{h}|^{1/2}}
\,.
\end{align*}
Moreover, since
$\wt{\V{g}}_\V{x}=(e^{i\V{k}\cdot\V{x}}-1)\,\V{g}_m(k)$
and $|e^{i\V{y}\cdot\V{x}}-e^{i\V{z}\cdot\V{x}}|\klg|\V{y}-\V{z}|\,|\V{x}|$,
we find
\begin{align*}
\frac{1}{|\V{k}|}\,\big|\Delta_{-\V{h}}\wt{\V{g}}_\V{x}\big|\,
&\klg\,
|\V{x}|\,\big|\Delta_{-\V{h}}\V{g}_m(k)\big|
+\,\frac{|\V{h}|\,|\V{x}|}{|\V{k}|\,|\V{k}-\V{h}|^{1/2}}\,,
\end{align*}
again for $m<|\V{k}|,|\V{k}-\V{h}|<\UV$.
Furthermore, it is clear that
\begin{align*}
\frac{(\Delta_\V{h}\omega)(k-\V{h})^2}{|\V{k}|^2|\V{k}-\V{h}|}\,
\klg\,\frac{|\V{h}|^2}{|\V{k}|^2\,|\V{k}-\V{h}|}\,.
\end{align*}
Finally, by Young's inequality, 
$$
\frac{|\V{h}|}{|\V{k}|\,|\V{k}-\V{h}|^{1/2}}
\,\klg\,
\frac{|\V{h}|}{3}
\Big(\frac{2}{|\V{k}|^{1/2}|\V{k}_\bot|}\,+\,
\frac{1}{|\V{k}-\V{h}|^{1/2}|(\V{k}-\V{h})_\bot|}\Big)
\,.
$$
Inserting the previous estimates in \eqref{lisel1}
we find some constant, $C'\in(0,\infty)$, such that,
for all $m\in(0,m_1]$, $\delta\in(0,1/2]$, and 
$q\in C_0(\cA_m\times\ZZ_2,\CC)$
\begin{align*}
\int|q(&k)|^2\,\|(\Delta_{-\V{h}}a)(k)\,\phi_m\|^2\,dk
\\
&\klg\,|\V{h}|^2\:
\frac{C(1+\UV^2)}{\delta}\int|q(k)|^2\,
\Big(\frac{1}{|\V{k}|\,|\V{k}_\bot|^2}\,+\,
\frac{1}{|\V{k}-\V{h}|\,|(\V{k}-\V{h})_\bot|^2}\Big)\,dk\,,
\end{align*}
provided that $m<|\V{k}|,|\V{k}-\V{h}|<\UV$ on the support of $q$.
Peaking at some fixed $k\in\cA_m\times\ZZ_2$ 
with $m<|\V{k}|,|\V{k}-\V{h}|<\UV$, $\V{k}_\bot,(\V{k}-\V{h})_\bot\not=0$, 
by inserting a family of approximate $\delta$-functions for $q$
we conclude the proof of Proposition~\ref{prop-pdb}.
\end{proof}

%%%%%%%%%%%%%%%%%%%%%%%%%%%%%%%%%%%%%%%%%%%%%%%%%%%%%%%%%%%%%%%%%%%%%%%%%%

\subsection{Some technical lemmata}
\label{ssec-katja}

\noindent
In this subsection we complete the derivation of the
soft photon and photon derivative bounds by providing
the missing estimates on 
\eqref{franz0b}, \eqref{franz1}, and \eqref{kalli4}. 
Throughout the whole section we drop the subscript $m$ and
one should keep in mind that $\V{g}$ and $\wt{\V{g}}_\V{x}$
are cut-off in the infra-red under this convention.

\begin{lemma}\label{le-lora}
Let $e^2,\UV>0$ and $\gamma\in(0,2/\pi)$.
Then we find constants, $C,C'\in(0,\infty)$,
such that, for all $m\in(0,m_1]$, $\delta\in(0,1]$,
$q\in C_0((\RR^3\setminus\{0\})\times\ZZ_2)$, every orthonormal
basis, $\{e_\ell:\,\ell\in\NN\}$, of $\HP$, and
$f_\ell:=\omega^{-1/2}\,q\,e_\ell$,
\begin{align}
\Big|\sum_{\ell\in\NN}&
\SPb{a(f_\ell)\,\SAt\,\wt{\phi}}{
\valpha\cdot\nonumber
\SPn{f_\ell}{\wt{\V{g}}_\V{x}}\,
\wt{\phi}}\Big|
\\
&\klg\,
\delta\int|q(k)|^2\,\|a(k)\,\wt{\phi}\|^2\,dk
\,+\,\frac{C}{\delta}\,
\int\Big(|\V{k}|+\frac{1}{|\V{k}|}\Big)\,
\id_{|\V{k}|\klg\UV}\,|q(k)|^2\,dk\,.\label{lora-eq1}
\end{align}
and
\begin{align}
\Big|\sum_{\ell\in\NN}&
\SPb{\nonumber
\Delta_{-\V{h}}a(f_\ell)\,\SAt\,\wt{\phi}}{
\valpha\cdot
\SPn{f_\ell}{\Delta_{-\V{h}}\wt{\V{g}}_\V{x}}\,
\wt{\phi}}\Big|
\\
&\klg\,
\delta\int|q(k)|^2\,\big\|\,\Delta_{-\V{h}}a(k)\,\wt{\phi}\,\big\|^2\,dk
\,+\,\frac{C'(1+\UV^2)}{\delta}\,
J^1_q(\V{h})\,.\label{lora-eq2}
\end{align}
Here $m_1>0$ is the parameter appearing in 
Proposition~\ref{prop-exp-loc} and
$J_q^1(\V{h})$ is defined in \eqref{def-Jqvkh}.
\end{lemma}

\begin{proof}
We only prove \eqref{lora-eq2} explicitly
as \eqref{lora-eq1} may be obtained by simply
ignoring the operators $\Delta_{\pm\V{h}}$ in the argument below.
Let 
$F\in C^\infty(\RR^3_\V{x},[0,\infty))$ 
be as in the paragraph preceding~\eqref{nora-88}.
We write
\begin{align}
\sum_{\ell\in\NN}&\nonumber
\SPb{\valpha\cdot
\SPn{\Delta_{-\V{h}}\wt{\V{g}}_\V{x}\,e^{-F}}{\omega^{-1/2}\,q\,e_\ell}\,
\Delta_{-\V{h}}a(\omega^{-1/2}\,q\,e_\ell)\,\SAt\,\wt{\phi}}{e^F\,\wt{\phi}}
\\
&=\,\nonumber
\SPb{\valpha\cdot \Delta_{-\V{h}}a\big(\omega^{-1}\,|q|^2\,
\Delta_{-\V{h}}\wt{\V{g}}_{\V{x}}\,e^{-F}\big)\,\SAt\,\wt{\phi}}{
e^F\,\wt{\phi}}
\\
&=\,
\sum_{\ell\in\NN}\label{lora1} 
\SPb{\valpha\cdot\SPn{\omega^{-1}\,\ol{q}\,
\Delta_{-\V{h}}\wt{\V{g}}_\V{x}\,e^{-F}}{e_\ell}\,
[\Delta_{-\V{h}}a(q\,e_\ell)\,,\,
\SAt]\,\wt{\phi}}{e^F\,\wt{\phi}}
\\
&\qquad\qquad+
\,\sum_{\ell\in\NN}\label{lora2}
\SPb{\valpha\cdot\SPn{\omega^{-1}\,\ol{q}\,
\Delta_{-\V{h}}\wt{\V{g}}_\V{x}\,e^{-F}}{e_\ell}\,\SAt\,
\Delta_{-\V{h}}a(q\,e_\ell)\,\wt{\phi}}{e^F\,\wt{\phi}}
\\
&=:\,\nonumber
\sum_{\ell\in\NN} Q_1(\ell)\,+\,\sum_{\ell\in\NN} Q_2(\ell)
\,.
\end{align}
On account of \eqref{katja1D}, \eqref{fritz3D}, and the Cauchy-Schwarz
inequality the term in
\eqref{lora1} is bounded by
\begin{align*}
\sum_{\ell\in\NN}|Q_1(\ell)|
\,&\klg\,C\,J_q^1(\V{h})^{1/2}\,J_q^0(\V{h})^{1/2}\,\|e^F\,\wt{\phi}\|^2\,.
\end{align*}
Using \eqref{Parceval-Deltaha} we estimate the term in
\eqref{lora2} as
$$
\sum_{\ell\in\NN}|Q_2(\ell)|
\,\klg\,C\,J_q^1(\V{h})^{1/2}\,\|e^F\,\wt{\phi}\|\,
\Big(\int|q(k)|^2\,\big\|\,\Delta_{-\V{h}}a(k)\,\wt{\phi}
\,\big\|^2\,dk\Big)^{1/2}\,.
$$ 
Altogether this implies the second asserted estimate
\eqref{lora-eq2}. (Recall \eqref{exp-loc-sup}.)
When we ignore the operators $\Delta_{\pm\V{h}}$
then we apply \eqref{katja1} and \eqref{fritz3} and
we have to replace each factor $J_q^{\nu}(\V{h})$
by some constant times $\|\omega^{1-\nu}\,\ol{q}\,\V{g}\|^2$.
\end{proof}

\begin{lemma}\label{le-katja1}
Let $e^2,\UV>0$ and $\gamma\in(0,2/\pi)$.
Then there is a constant, $C\in(0,\infty)$, such that,
for every orthonormal basis, $\{e_\ell:\,\ell\in\NN\}$, of $\HP$,
and 
for all $q\in C_0((\RR^3\setminus\{0\})\times\ZZ_2)$, 
$m\in (0,m_1]$, $\delta\in(0,1]$,
and $f_\ell:=\omega^{-1/2}\,q\,e_\ell$,
\begin{align}
\Big|\sum_{\ell\in\NN}&
\SPB{|\DAt|^{1/2}\,[\ad(f_\ell)\,,\,\SAt]\,
a(f_\ell)\,\wt{\phi}}{\SAt\,|\DAt|^{1/2}\,\wt{\phi}}\Big|
\nonumber
\\
&\klg\,
\delta\int|q(k)|^2\,\|a(k)\,\wt{\phi}\|^2\,dk\,+\,
\frac{C}{\delta}\int\Big(|\V{k}|+\frac{1}{|\V{k}|}\Big)
\,|q(k)|^2\,dk\label{katja-2}
\,.
\end{align}
Moreover, for all $\V{h}\in\RR^3$,
\begin{align}
\Big|\sum_{\ell\in\NN}&\SPB{|\DAt|^{1/2}\,\big[\,\ad(\Delta_{\V{h}}f_\ell)
\,,\,\SAt\,\big]\,
\Delta_{-\V{h}}a(f_\ell)\,\wt{\phi}}{
\SAt\,|\DAt|^{1/2}\,\wt{\phi}}\Big|
\nonumber
\\
&\klg\,
\label{katja-2D}
\delta\int|q(k)|^2\,\|\Delta_{-\V{h}}a(k)\,\wt{\phi}\|^2\,dk
\,+\,\frac{C(1+\UV^2)}{\delta}
\,J_q^1(\V{h})\,.
\end{align}
Here $m_1>0$ is the parameter appearing in 
Proposition~\ref{prop-exp-loc} and $J_q^1(\V{h})$ is defined in
\eqref{def-Jqvkh}.
\end{lemma}

\begin{proof}
Let 
$F\in C^\infty(\RR^3_\V{x},[0,\infty))$ 
be as in the paragraph preceding~\eqref{nora-88}
and set $L:=i\valpha\cdot\nabla F$.
Since we do not know whether $\wt{\phi}\in\dom(|\DAt|^{1/2})$ 
and $e^F\,\wt{\phi}\in\dom(|\DAt|^{1/2})$ belong to the
domain of $\DAt$ and since the commutation relation
\eqref{fritz2b} requires exponential weights
in order to control $\wt{\V{g}}_\V{x}$
we have to be careful when doing formal manipulations
in what follows.
Therefore, the arguments in the next paragraphs look 
somewhat elaborate.
Let $f\in\HP$ such that $\omega^{-1/2}\,f\in\HP$ also.
On account of \eqref{sgn} and \eqref{fritz2b} we have
\begin{align}
\SPb{|\DAt|^{1/2}\,[\ad(&f)\,,\,\SAt]\,
a(f)\,\wt{\phi}}{\SAt\,|\DAt|^{1/2}\,\wt{\phi}}\nonumber
\\
&=\,
\SPB{|\DAt|^{1/2}\nonumber
\,\lim_{\tau\to\infty}\int_{-\tau}^\tau\eta(y)\,\frac{dy}{\pi}
}{\SAt\,|\DAt|^{1/2}\,\wt{\phi}}\,,
\end{align}
where
$$
\eta(y)\,:=\,
\RAt{iy}\,\valpha\cdot\SPn{\wt{\V{g}}_\V{x}\,e^{-F}}{f}
\,\R{\wt{\V{A}},L}{iy}\,a(f)\,
e^F\,\wt{\phi}\in\dom(\DAt)\,,\qquad y\in\RR\,.
$$
We recall from Lemma~\ref{le-HfmeF} that $e^F\,\wt{\phi}\in\dom(\Hft^{1/2})$
and, hence, $e^F\,\wt{\phi}\in\dom(a(f))$.
Next, we observe that both Bochner integrals 
$\int_\RR\eta(y)\,dy$ and $\int_\RR|\DAt|^{1/2}\,\eta(y)\,dy$
are absolutely convergent. Since $|\DAt|^{1/2}$ is closed
and the Bochner integral commutes with closed operators
we thus get
\begin{align}
\SPb{&|\DAt|^{1/2}\,[\ad(f)\,,\,\SAt]\,
a(f)\,\wt{\phi}}{\SAt\,|\DAt|^{1/2}\,\wt{\phi}}\nonumber
\\
&=\,\nonumber
\int_\RR
\SPB{e^{-F}\,
\DAt\,\eta(y)}{e^F\,\wt{\phi}}\,\frac{dy}{\pi}
\\
&=\,
\int_\RR
\SPB{
(\DAt-L)\,\R{\wt{\V{A}},-L}{iy}\,\valpha\cdot\SPn{\wt{\V{g}}_\V{x}\,e^{-F}}{f}
\R{\wt{\V{A}}}{iy}\,\nonumber
a(f)\,\wt{\phi}}{e^F\,\wt{\phi}}\,\frac{dy}{\pi}
\\
&=\label{katja4}\,\int_\RR
\SPB{\RAt{iy}\,a(f)\,\wt{\phi}}{
\valpha\cdot\SPn{f}{\wt{\V{g}}_\V{x}\,e^{-F}}\,\wt{\eta}(y)}
\,\frac{dy}{\pi}
\,,
\end{align}
where we applied Lemma~\ref{le-marah} in the second step and
abbreviated
\begin{equation}\label{def-wteta}
\wt{\eta}(y)\,:=\,\big(|\DAt|^{1/2}\,\R{\wt{\V{A}},-L}{iy}\big)^*\,
\SAt\,|\DAt|^{1/2}\,e^F\,\wt{\phi}\,+\,
\R{\wt{\V{A}},L}{-iy}\,L\,e^F\,\wt{\phi}
\,.
\end{equation}
In Lemma~\ref{le-katja2} below we show that there is
some constant, $C\in(0,\infty)$, such that,
for all $m\in[0,m_1]$,
\begin{equation}\label{katja222}
\int_\RR\|\wt{\eta}(y)\|^2\,\frac{dy}{\pi}\,\klg\,C\,.
\end{equation}
In the sequel we only treat \eqref{katja-2D}
explicitly. 
From time to time we indicate what has to be
changed in order to derive \eqref{katja-2} which is obtained
essentially by ignoring the operators $\Delta_{\pm\V{h}}$
below.
Substituting $f$ by $\Delta_\V{h}f_\ell$ in \eqref{katja4}
and employing \eqref{katja1}
we get
\begin{align*}
\int_\RR
\Big|\SPB{&\RAt{iy}\,a(f)\,\wt{\phi}}{
\valpha\cdot\SPn{f}{\wt{\V{g}}_\V{x}\,e^{-F}}\,\wt{\eta}(y)}\Big|
\,\frac{dy}{\pi}
\\
&\klg
\Big(\sum_{\ell\in\NN}\int_\RR
\frac{\|a(\Delta_{\V{h}}f_\ell)\,\wt{\phi}\|^2}{1+y^2}\,dy\Big)^{\frac{1}{2}}
\Big(J_q^{1/2}(\V{h})
\int_\RR\|\wt{\eta}(y)\|^2dy\Big)^{\frac{1}{2}}
\,.
\end{align*}
(When $\Delta_\V{h}$ is dropped then $J_q^{1/2}(\V{h})$ has to be replaced
by $C_F\,\|\omega^{1/2}\,\ol{q}\,\V{g}\|$.)
Here the first integral on the right hand side,
$$
\sum_{\ell\in\NN}\int_\RR
\frac{\|a(\Delta_{\V{h}}f_\ell)\,\wt{\phi}\|^2}{1+y^2}\,dy
\,=\,\pi\int\frac{|q(k)|^2}{\omega(k)}\,
\big\|\,\Delta_{-\V{h}}a(k)\,\wt{\phi}\,\big\|^2\,dk
$$ 
is finite since $q$ has a compact support in 
$(\RR^3\setminus\{0\})\times\ZZ_2$.
When we sum \eqref{katja4} 
with $f=\Delta_\V{h}f_\ell$ (or $f=f_\ell$)
with respect to $\ell$ we may thus
interchange the $dy$-integration with the $\ell$-summation.
Proceeding in this way and commuting $a(\Delta_\V{h}f_\ell)$ through
the resolvent in the last line of \eqref{katja4} and
using \eqref{fritz2} and  
\begin{equation*}
\sum_{\ell\in\NN}\SPn{\wt{\V{g}}_{\V{x}}\,e^{-F}}{\Delta_\V{h}f_\ell}
\,a(\Delta_\V{h}f_\ell)\,\psi
\,=\,
\Delta_{-\V{h}}a\big(\omega^{-1}\,|q|^2\,
\Delta_{-\V{h}}\wt{\V{g}}_\V{x}\,e^{-F}\big)\,\psi\,,
\end{equation*}
for $\psi\in\dom(\Hft^{1/2})$,
we arrive at
\begin{align}
\Big|\sum_{\ell\in\NN}&\SPb{|\DAt|^{1/2}\,[\ad(\Delta_\V{h}f_\ell)\,,\,\SAt]\,
a(\Delta_\V{h}f_\ell)\,\wt{\phi}}{\SAt\,|\DAt|^{1/2}\,\wt{\phi}}
\Big|\nonumber
\\
&\klg\,\Big|\int_\RR\SPB{\valpha\cdot\label{katja5}
\Delta_{-\V{h}}a\big(\omega^{-1}\,|q|^2\,\Delta_{-\V{h}}
\wt{\V{g}}_\V{x}\,e^{-F}\big)\,
\R{\wt{\V{A}}}{iy}\,\wt{\phi}}{\wt{\eta}(y)}\,\frac{dy}{\pi}\Big|
\\
&+\label{katja6}
\Big(\sum_{\ell\in\NN}\int_\RR\|\RAt{iy}\|^2\,\big\|
\valpha\cdot\SPn{\Delta_{\V{h}}f_\ell}{\wt{\V{g}}_\V{x}\,e^{-F}}\,
\R{\wt{\V{A}},L}{iy}\,e^F\,\wt{\phi}\big\|^2\,\frac{dy}{\pi}\Big)^{\frac{1}{2}}
\\
&\quad\cdot\label{katja7}
\Big(\sum_{\ell\in\NN}\int_\RR
\big\|\,\valpha\cdot\SPn{\Delta_\V{h}f_\ell}{
\wt{\V{g}}_\V{x}\,e^{-F}}
\,\wt{\eta}(y)\,\big\|^2\,\frac{dy}{\pi}\Big)^{\frac{1}{2}}.
\end{align} 
Applying \eqref{marah1},
\eqref{katja1D}, and \eqref{katja222} to the terms in \eqref{katja6}
and \eqref{katja7} 
we obtain
\begin{align}
&\big(\textrm{integral in \eqref{katja6}}\big)^{1/2}
\cdot\big(\textrm{integral in \eqref{katja7}}\big)^{1/2}
\,\klg\,\const\,J^{1/2}_q(\V{h})\,\|e^F\,\wt{\phi}\|
\,.\label{katja2008}
\end{align}
When we ignore the operators $\Delta_{\pm\V{h}}$ in the estimates
above then we apply \eqref{katja1} instead of \eqref{katja1D} and
$J_q^{1/2}(\V{h})$ has to be replaced by 
$C_F\,\|\omega^{1/2}\,\ol{q}\,\V{g}\|^2$
in \eqref{katja2008}.
The idea behind the procedure started above is that
we can now write
\begin{align*}
\valpha\cdot
a\big(&\Delta_\V{h}(\omega^{-1}\,|q|^2\,\Delta_{-\V{h}}
\wt{\V{g}}_\V{x})\,e^{-F}\big)\,
\R{\wt{\V{A}}}{iy}\,\wt{\phi}
\\
&=\,
\sum_{\ell\in\NN}\valpha\cdot
\SPn{\omega^{-1}\,\ol{q}\,\Delta_{-\V{h}}\wt{\V{g}}_{\V{x}}\,e^{-F}}{
e_\ell}
\,a\big(\Delta_\V{h}(q\,e_\ell)\big)
\,\R{\wt{\V{A}}}{iy}\,\wt{\phi}
\end{align*}
and commute the $\V{x}$-independent annihilation operator
$a\big(\Delta_\V{h}(q\,e_\ell)\big)$ -- which also contains
no $\omega^{-1/2}$-singularity anymore --
with the resolvent to its right.
As a result we obtain
\begin{align}
\Big|\int_\RR\SPB{&\valpha\cdot
a\big(\Delta_\V{h}(\omega^{-1}\,|q|^2\,\Delta_{-\V{h}}
\wt{\V{g}}_\V{x})\,e^{-F}\big)\,
\R{\wt{\V{A}}}{iy}\,\wt{\phi}}{\wt{\eta}(y)}\,\frac{dy}{\pi}\Big|
\nonumber
\\
&\klg\label{katja9}
\Big|\sum_{\ell\in\NN}
\int_{\RR}\SPB{\R{\wt{\V{A}}}{iy}\,a
\big(\Delta_{\V{h}}(q\,e_\ell)\big)\,\wt{\phi}}{
\valpha\cdot\SPn{e_\ell}{\tfrac{\ol{q}}{\omega}\,
\Delta_{-\V{h}}\wt{\V{g}}_{\V{x}}\,e^{-F}}
\,\wt{\eta}(y)}\,dy\,\Big|
\\
&\;+\,\label{katja10}
\Big(\sum_{\ell\in\NN}\int_\RR
\big\|\,\R{\wt{\V{A}}}{iy}
\,\valpha\cdot\SPn{q\,e_\ell}{\Delta_{-\V{h}}
\wt{\V{g}}_{\V{x}}\,e^{-F}}\,
\R{\wt{\V{A}},L}{iy}\,e^F\,\wt{\phi}\,\big\|^2\,dy
\Big)^{\frac{1}{2}}
\\
&\qquad\qquad\qquad\cdot\label{katja11}
\Big(\sum_{\ell\in\NN}\int_\RR
\big\|\,\valpha\cdot
\SPn{e_\ell}{\omega^{-1}\,\ol{q}\,
\Delta_{-\V{h}}\wt{\V{g}}_{\V{x}}\,e^{-F}}
\,\wt{\eta}(y)\,\big\|^2\,dy
\Big)^{\frac{1}{2}}
\end{align}
By virtue of \eqref{katja1D} we have,
analogously to \eqref{katja2008},
\begin{align}
&\big(\textrm{integral in \eqref{katja10}}\big)^{1/2}
\cdot\big(\textrm{integral in \eqref{katja11}}\big)^{1/2}
\,\klg
\,\const
\,\big(J_q^0(\V{h})\,J_q^{1}(\V{h})\big)^{1/2}\,.\label{katja12}
\end{align}
When we ignore the operators $\Delta_{\pm\V{h}}$, then
the factor $\big(J^0_q(\V{h})\,J_q^{1}(\V{h})\big)^{1/2}$
has to be replaced by 
$C_F\,\|\omega\,q\,\V{g}\|\,\|q\,\V{g}\|$
in the 
previous estimate.
The term in \eqref{katja9}, finally, is estimated as
\begin{align}\nonumber
\big|\,\textrm{integral in \eqref{katja9}}\,\big|
\,&\klg\,\nonumber
\Big(\sum_{\ell\in\NN}\int_\RR
\frac{\|\Delta_{-\V{h}}a
(q\,e_\ell)\,\wt{\phi}\|^2}{1+y^2}\,\frac{dy}{\pi}\Big)^{\frac{1}{2}}
\\
&\quad\cdot\nonumber
\Big(\sum_{\ell\in\NN}\int_\RR
\big\|
\valpha\cdot\SPn{e_\ell}{\omega^{-1}\,\ol{q}\,
\Delta_{-\V{h}}\wt{\V{g}}_{\V{x}}\,e^{-F}}
\,\wt{\eta}(y)\,\big\|^2\,\frac{dy}{\pi}\Big)^{\frac{1}{2}}
\\
&
\klg\,\label{katja12b}
\Big(
\int|q(k)|^2\,\big\|\,\Delta_{-\V{h}}a(k)\,\wt{\phi}\,\big\|^2\,dk
\Big)^{\frac{1}{2}}
\,J_q^1(\V{h})^{1/2}\,
C^{1/2}\,.
\end{align}
Again, the factor $J^1_q(\V{h})^{1/2}$ 
has to be replaced by $C_F\,\|q\,\V{g}\|$,
when the operators $\Delta_{\pm\V{h}}$ are dropped.
Combining \eqref{katja5}--\eqref{katja12b} we arrive at the asserted
estimate \eqref{katja-2D}. Taking also the modifications indicated
above into account we further obtain \eqref{katja-2}.
\end{proof}

\begin{lemma}\label{le-katja2}
The bound \eqref{katja222} holds true, where
$\wt{\eta}$ is defined in \eqref{def-wteta} and
the constant $C$ does not depend on $m\in(0,m_1]$.
\end{lemma}

\begin{proof}
As we did in the whole subsection we drop all subscripts $m$
in this proof.
Writing
$\R{\wt{\V{A}},-L}{iy}=\RAt{iy}\,(\id+L\,\R{\wt{\V{A}},-L}{iy})$
we deduce that
$$
(|\DAt|^{1/2}\,\R{\wt{\V{A}},-L}{iy})^*
\,=\,(\id-\R{\wt{\V{A}},L}{-iy}\,L)\,|\DAt|^{1/2}\,\RAt{-iy}\,.
$$
Abbreviating $\psi_F:=\SAt\,|\DAt|^{1/2}\,e^F\,\wt{\phi}$
and employing Lemma~\ref{le-marah} we thus
get
\begin{align*}
\int\|\eta(y)\|^2\,dy
\,&\klg\,
C'\,\int_\RR\big\|\,|\DAt|^{1/2}\,\RAt{-iy}\,\psi_F\,\big\|^2\,dy
\,+\,C'\int_\RR\frac{\|e^F\,\phi\|^2}{1+y^2}\,dy\,.
\end{align*}
Moreover, the spectral calculus yields
\begin{align*}
\int_\RR\big\|\,|\DAt|^{1/2}\,\RAt{iy}\,\psi_F\,\big\|^2\,dy
&=\,
\int_\RR\int_\RR\frac{|\lambda|}{\lambda^2+y^2}\,dy\,
d\|\id_\lambda(\DAt)\,\psi_F\|^2\,=\,\pi\,\|\psi_F\|^2\,.
\end{align*} 
We conclude by recalling that
$\|\psi_F\|\klg \|\,|\DA|^{1/2}\,e^F\,\phi\|\klg\,C''$
in virtue of Lemma~\ref{le-HfmeF}. 
\end{proof}

%%%%%%%%%%%%%%%%%%%%%%%%%%%%%%%%%%%%%%%%%%%%%%%%%%%%%%%%%%%%%%%%%%%%%%%%%
%%%%%%%%%%%%%%%%%%%%%%%%%%%%%%%%%%%%%%%%%%%%%%%%%%%%%%%%%%%%%%%%%%%%%%%%%
%%%%%%%%%%%%%%%%%%%%%%%%%%%%%%%%%%%%%%%%%%%%%%%%%%%%%%%%%%%%%%%%%%%%%%%%%

\appendix

\section{\`{A}-priori bounds on eigenvectors}
\label{app-a(f)}

\noindent
The purpose of this appendix is to show
that, for every eigenvector, $\phi_m$, of $\PFm{\gamma}$
and every $f\in\HP$ with $\omega^{-1/2}\,f\in\HP$ and $\omega\,f\in\HP$,
the vector $a(f)\,\phi_m$ belongs to the form domain
of $\PFm{\gamma}$. This information is necessary in order
to derive the soft photon and photon derivative bounds.
To prove this result we shall 
essentially proceed along the lines
of \cite[Appendix~B]{GLL2001} in the proof of
Lemma~\ref{le-clara4} later on. Lemma~\ref{le-clara2}
has also been observed in the non-relativistic setting in \cite{GLL2001}.
In order to deal with the difficulties posed
by the non-locality of $\PFmg{\gamma}$ we shall 
derive two additional technical lemmata.

In what follows we set, for $E\grg1$ and $f$ as above,
$$
a_E(f)\,:=\,a(f)\,E\,(\Hf+E)^{-1}\,,\qquad
\ad_E(f)\,:=\,E\,(\Hf+E)^{-1}\,\ad(f)\,.
$$

\begin{lemma}\label{le-clara2}
Let $e^2,\UV,m>0$, $\gamma\in[0,2/\pi)$,
$E\grg1$, and $f\in\HP$ such that $\omega^{-1/2}\,f\in\HP$. Then
$a_E(f)$ and $\ad_E(f)$ are continuous operators on the
Hilbert space $\form(\PFm{\gamma})$ equipped
with the form norm corresponding to $\PFm{\gamma}$.
\end{lemma}

\begin{proof}
It suffices to prove the assertion for $\gamma=0$ since
the form norms of $\PFm{0}$ and $\PFm{\gamma}$ are
equivalent, for $\gamma\in(0,2/\pi)$.
To begin with we recall from Lemma~\ref{le-sres-conv-PF}
that $\form(\PFm{0})\subset\form(\Hf)\subset\dom(a^\sharp(f))$
since $\omega^{-1/2}\,f\in\HP$, where
$a^\sharp$ is $a$ or $\ad$.
Applying Corollary~\ref{cor-pia} 
(with $\wt{\V{A}}$ replaced by $\V{A}_m$ and $\V{A}$ replaced by $\V{0}$)
we find some
constant $C\in(0,\infty)$ such that,
for every $\vp\in\core$,
$$
\SPb{a^\sharp_E(f)\,\vp}{\PFm{0}\,a^\sharp_E(f)\,\vp}
\,\klg\,C\,\big(
\big\|\,|\DO|^{1/2}\,a^\sharp_E(f)\,\vp\,\big\|^2
+
\big\|\,(\Hf+1)^{1/2}\,a_E^\sharp(f)\,\vp\,\big\|^2\big).
$$
In the first term on the right side $\,a_E^\sharp(f)$ commutes with
$|\DO|^{1/2}$. Moreover, the norm 
of $a_E^\sharp(f)$ is bounded by some constant 
depending on $E$.
In the second term the operator $(\Hf+1)^{1/2}\,a_E^\sharp(f)$ is easily
seen to be bounded, for fixed $E$, as well.
Employing Corollary~\ref{cor-pia}
once more (this time with $\wt{\V{A}}=\V{0}$) we thus find, for some
$C_E\in(0,\infty)$,
\begin{equation*}
\SPb{a^\sharp_E(f)\,\vp}{\PFm{0}\,a_E^\sharp(f)\,\vp}
\,\klg\,C_E\,\SPn{\vp}{(\PFm{0}+1)\,\vp}\,.
\end{equation*}
Since $\core$ is a form core for $\PFm{0}$
and $a_E^\sharp(f)$ is bounded the statement becomes evident.
\end{proof}

\begin{lemma}\label{le-friedel}
Assume that $\vo$ and $\V{G}$ satisfy Hypothesis~\ref{hyp-G}.
Then the operator $|\DA|^{1/2}\,E\,(\Hf+E)^{-1}\,|\DA|^{-1/2}$ is
defined on all of
$\HR_4$ and its norm is bounded uniformly in $E\grg1$.
\end{lemma}

\begin{proof}
We use the norm convergent integral representation
\cite[Page~286]{Kato}
\begin{align*}
|\DA|^{-1/2}\,=\,(\DA^2)^{-1/4}\,&=\,
\frac{1}{2^{1/2}\pi}\int_0^\infty\frac{1}{\DA^2+t}\,\frac{dt}{t^{1/4}}
\\
&=\,
\frac{1}{2^{1/2}\pi}
\int_0^\infty\frac{1}{2i}\Big(
\frac{1}{\DA-i\,t^{1/2}}\,-\,\frac{1}{\DA+i\,t^{1/2}}
\Big)\,\frac{dt}{t^{3/4}}
\end{align*}
to get, for $\vp,\psi\in\core$,
\begin{align*}
\SPb{|\DA|^{1/2}&\,\vp}{\big[\,|\DA|^{-1/2}\,,\,E\,(\Hf+E)^{-1}\,\big]\,\psi}
\\
&=\,\frac{1}{2^{3/2}\pi i}
\sum_{\vk=\pm1}\vk\int_0^\infty
\SPB{\vp}{\frac{|\DA|^{1/2}}{\DA-\vk\, i\,t^{1/2}}\,
\big\{E\,(\Hf+E)^{-1}\big\}\,\times
\\
&\qquad\qquad\qquad\times\,\big\{[\valpha\cdot\V{A}\,,\,\Hf]\,
(\Hf+E)^{-1}\big\}\,\frac{1}{\DA-\vk\,i\,t^{1/2}}\,\psi}\,\frac{dt}{t^{3/4}}\,.
\end{align*}
On account of \eqref{eq-friedel} below it is obvious 
that both operators in the curly
brackets $\{\cdots\}$ are bounded uniformly in $E\grg1$.
Taking also \eqref{ralf3} into account we readily infer that
the commutator of $|\DA|^{-1/2}$ and $E\,(\Hf+E)^{-1}$
maps $\HR_4$ into the domain of $|\DA|^{1/2}$ and that
$$
\sup_{E\grg1}\big\|\,
|\DA|^{1/2}\,\big[\,|\DA|^{-1/2}\,,\,E\,(\Hf+E)^{-1}\,\big]\,\big\|\,<\,\infty\,.
$$ 
Now the assertion is obvious.
\end{proof}

\begin{lemma}\label{le-friedel2}
Assume that $\vo$ and $\V{G}$ fulfill Hypothesis~\ref{hyp-G}.
Then the operator $(\Hf+E)^{-1/2}\,[\Hf\,,\,\SA]\,|\DA|^{1/2}$
is well-defined on $\core$, bounded, and its norm is
bounded uniformly in $E\grg1+(2d_1)^2$.
\end{lemma}

\begin{proof}
It is well-known that
$$
\valpha\cdot\V{E}\,:=\,
[\Hf\,,\,\valpha\cdot\V{A}]\,=\,\sum_{\vs=1,2,3,4}\int_{\RR^3}^\oplus
\valpha\cdot\big(\ad(\omega\,e^{i\V{k}\cdot\V{x}}\,\V{g})
-a(\omega\,e^{i\V{k}\cdot\V{x}}\,\V{g})\big)\,d^3\V{x}\,,
$$
where $\V{E}$ is the electric field, and that,
consequently,
\begin{equation}\label{eq-friedel}
\big\|\,(\Hf+E)^{-1/2}\,[\Hf\,,\,\valpha\cdot\V{A}]\,\big\|\,\klg\,
(d_2^2+2d_1^2)^{1/2}\,.
\end{equation}
Employing Formula \eqref{sgn}, the intertwining
relation \eqref{eva3}, and \eqref{ralf3}
we thus get, for all
$\vp,\psi\in\core$,
\begin{align*}
\big|&\SPb{\vp}{(\Hf+E)^{-1/2}\,[\SA\,,\,\Hf]\,|\DA|^{1/2}\,\psi}\big|
\\
&\klg\,
\int_\RR
\Big|\SPB{\vp}{\Xi_{1/2,0}(iy)\RA{iy}\,(\Hf+E)^{-1/2}\,
[\Hf\,,\,\valpha\cdot\V{A}]\,\RA{iy}\,|\DA|^{1/2}\,\psi}\Big|
\,\frac{dy}{\pi}
\\
&\klg\,
C\,\big\{\sup_{y\in\RR}\|\Xi_{1/2,0}(iy)\|\big\}
\,(d_2^2+2d_1^2)^{1/2}\,\|\vp\|\,\|\psi\|\,.
\end{align*}
According to \eqref{verona} and
\eqref{bd-Xi} the supremum in the last line is less than or
equal to $(1-2d_1/E^{1/2})^{-1}$, which is uniformly bounded,
for $E\grg1+(2d_1)^2$.
\end{proof}

\begin{lemma}\label{le-clara4}
Let $\gamma\in(0,2/\pi)$, $m>0$, $f\in \HP$ 
such that $\omega^{-1/2}\,f,\omega\,f\in\HP$,
and assume that $\phi_m$ is an eigenvector of $\PFm{\gamma}$.
Then it follows that 
$a(f)\,\phi_m\in\form(\PFm{\gamma})$.
\end{lemma}

\begin{proof}
This proof proceeds along the lines of an argument
in \cite[Appendix~B]{GLL2001}.

To begin with we observe that $\phi_m\in\dom(a(f))$ since
$\dom(\PFm{\gamma})\subset\dom(\Hf^{1/2})\subset\dom(a(f))$ 
by Corollary~\ref{cor-pia}.
Moreover, using that $\omega\,f\in\HP$, it is easily verified that
$a_E(f)\,\phi_m\to a(f)\,\phi_m$, as $E$ tends to infinity.
To prove the lemma is thus
suffices to show that there is some $E$-independent
constant, $C\in(0,\infty)$, such that
\begin{equation}\label{adam41}
\SPb{a_E(f)\,\phi_m}{\PFm{\gamma}\,a_E(f)\,\phi_m}\,\klg\,C\,,
\end{equation}
for all sufficiently large values of $E>0$.
In fact, Lemma~\ref{le-clara2} ensures that the left hand
side of \eqref{adam41} is well-defined and \eqref{adam41} itself
implies that the functional 
$$
u(\eta)\,:=\,\SPn{a(f)\,\phi_m}{(\PFm{\gamma})^{1/2}\,\eta}
\,=\,\lim_{E\to\infty}\SPn{(\PFm{\gamma})^{1/2}\,a_E(f)\,\phi_m}{\eta}\,,
$$
for all $\eta\in\dom((\PFm{\gamma})^{1/2})$,
is bounded with $\|u\|\klg C^{1/2}$,
whence 
$a(f)\,\phi_m$ belongs to
$\dom((\PFm{\gamma})^{1/2\,*})=\form(\PFm{0})$.

In order to prove \eqref{adam41} we pick some $\vp\in\core$
and write
\begin{align}
[\PFm{\gamma}\,,\,a_E(f)]\,\vp
\,&=\,\label{adam42}
[\SA\,,\,a(f)]\,\DA\,E\,(\Hf+E)^{-1}\,\vp
\\
&\;\;\;\;-\,\valpha\cdot\SPn{f}{\id_{\cA_m}\,e^{-i\V{k}\cdot\V{x}}\,\V{g}}\,
E\,(\Hf+E)^{-1}\,\vp\label{adam42b}
\\
&\;\;\;\;-\,\label{adam43}
a(\omega\,f)\,E\,(\Hf+E)^{-1}\,\vp
\\
&\;\;\;\;+\,\label{adam45}
a(f)\,[\PFm{\gamma}\,,\,E\,(\Hf+E)^{-1}]\,\vp
\\
&=:\,\nonumber
Y_1\,\vp\,+\,Y_2\,\vp\,+\,Y_3\,\vp\,+\,Y_4\,\vp
\,.
\end{align}
First, we discuss the terms in \eqref{adam42}--\eqref{adam43}.
We recall that both operators 
$[\SA\,,\,a(f)]\,|\DA|^{1/2}$
and $(\Hf+1)^{-1/2}\,a(\omega\,f)$ are bounded.
(Here we use that $\omega^{1/2}\,f\in L^2$.)
From these remarks we readily infer that
\begin{align*}
\big|\SPb{&a_E(f)\,\vp}{(Y_1+Y_2+Y_3)\,\vp}\big|
\\
&\klg\,C\,
\|a_E(f)\,\vp\|\,\big\|\,|\DA|^{1/2}\,E\,(\Hf+E)^{-1}\,|\DA|^{-1/2}\,\big\|
\,\big\|\,|\DA|^{1/2}\,\vp\,\big\|
\\
&\;\;\;+\,
C\,\|a_E(f)\,\vp\|\,\big\|\,E\,(\Hf+E)^{-1}\,\big\|\,\|\vp\|
\\
&\;\;\;+\,
C\,\big\|\,(\Hf+1)^{1/2}\,a_E(f)\,\vp\,\big\|\,
\big\|\,E\,(\Hf+E)^{-1}\,\big\|\,\|\vp\|\,,
\end{align*}
where the constant is independent of $E$.
Thanks to Lemma~\ref{le-friedel} we know that
the second norm in the second line is bounded uniformly
in $E\grg0$. Consequently, we find, for every $\ve\in(0,1]$,
some $C_\ve\in(0,\infty)$ such that, for all
$E\grg1$ and $\vp\in\core$,
\begin{align}
\big|\SPb{a_E(f)\,\vp}{&(Y_1+Y_2+Y_3)\,\vp}\big|\nonumber
\\
&\klg\,\ve\,\SPb{a_E(f)\,\vp}{(\Hf+1)\,a_E(f)\,\vp}+
C_\ve\,\SPb{\vp}{|\DA|\,\vp}\,.\label{walter1}
\end{align}
In order to treat the term $Y_4\,\vp$ in \eqref{adam45} we write
\begin{align*}
Y_4\,\vp\,&=\,
a(f)\,\big[\,|\DA|\,,\,E\,(\Hf+E)^{-1}\,\big]\,\vp
\\
&=\,
\big\{a(f)\,(\Hf+E)^{-1/2}\big\}\big\{(\Hf+E)^{-1/2}\,
\big[\,\Hf\,,\,|\DA|\,\big]\,
|\DA|^{-1/2}\big\}\,\times
\\
&\qquad\times\,
\big\{|\DA|^{1/2}\,E\,(\Hf+E)^{-1}\,|\DA|^{-1/2}\big\}\,|\DA|^{1/2}\,\vp\,.
\end{align*}
Here the first and the third curly brackets $\{\cdots\}$
are bounded operators on $\HR_4$ whose norms are uniformly
bounded in $E\grg1$ due to a well-known estimate and
Lemma~\ref{le-friedel}, respectively.
We write the operator in the second curly bracket as
\begin{align}
(\Hf&+E)^{-1/2}\,
\big[\,\Hf\,,\,|\DA|\,\big]\,
|\DA|^{-1/2}\nonumber
\\
&=\,
(\Hf+E)^{-1/2}\,\big[\Hf\,,\,\SA\big]\,|\DA|^{1/2}\,\SA\label{hansi1}
\\
&\qquad+\;\label{hansi2}
\big\{(\Hf+E)^{-1/2}\,\SA\,(\Hf+E)^{1/2}\big\}\,\times
\\
&\qquad\qquad\quad\times\,\label{hansi3}
(\Hf+E)^{-1/2}\,[\Hf\,,\,\valpha\cdot\V{A}]\,|\DA|^{-1/2}\,.
\end{align}
Here the operators in
\eqref{hansi1} and \eqref{hansi2}
are bounded uniformly in $E\grg1+(2d_1)^2$
as we know from Lemma~\ref{le-friedel2} and
\eqref{SA-eF-Hnu}, respectively.
The operator
in \eqref{hansi3} is bounded uniformly in $E\grg1$
according to \eqref{eq-friedel}.
Altogether it follows that $Y_4$ is a bounded
operator with domain $\core$ whose norm is uniformly
bounded, for $E\grg1+(2d_1)^2$.
Combining this result with 
\eqref{walter1} we find, for every $\ve>0$, two
constants $C',C_\ve'\in(0,\infty)$
such that, for all $E\grg1+(2d_1)^2$ and $\vp\in\core$,
\begin{align*}
\big|\SPb{a_E(f)\,&\vp}{[\PFm{\gamma}\,,\,a_E(f)]\,\vp}\big|
\\
&\klg\,
\ve\,\SPb{a_E(f)\,\vp}{(\PFm{\gamma}+C')\,a_E(f)\,\vp}
\,+\,C_\ve'\,
\SPb{\vp}{(\PFm{\gamma}+C')\,\vp}\,.
\end{align*}
Here we also applied Corollary~\ref{cor-pia}.

Now, we conclude as follows.
Since $a_E(f)$ and $\ad_E(f)$ are bounded operators
on $\form(\PFm{\gamma})$ and $\core$ is a form core
for $\PFm{\gamma}$
it follows from the previous estimate that 
\begin{align*}
\SPb{a_E&(f)\,\phi_m}{\PFm{\gamma}\,a_E(f)\,\phi_m}
\\
&=\,
E_m\,\SPb{a_E(f)\,\phi_m}{a_E(f)\,\phi_m}
\,+\,
\SPb{a_E(f)\,\phi_m}{[\PFm{\gamma}\,,\,a_E(f)]\,\phi_m}
\\
&\klg\,
\ve\,\SPb{a_E(f)\,\phi_m}{\PFm{\gamma}\,a_E(f)\,\phi_m}
\\
&\quad
+\,
(E_m+\ve\,C')\,\|\omega^{-1/2}\,f\|^2\,\SPn{\phi_m}{\Hf\,\phi_m}
\,+\,C_\ve'\,\SPn{\phi_m}{(\PFm{\gamma}+C')\,\phi_m}\,.
\end{align*}
Choosing some
$\ve<1$ and applying Corollary~\ref{cor-pia} once more 
we arrive at the desired bound \eqref{adam41}.
\end{proof}

%%%%%%%%%%%%%%%%%%%%%%%%%%%%%%%%%%%%%%%%%%%%%%%%%%%%%%%%%%%%%%%%%%%%%%%%%
%%%%%%%%%%%%%%%%%%%%%%%%%%%%%%%%%%%%%%%%%%%%%%%%%%%%%%%%%%%%%%%%%%%%%%%%%
%%%%%%%%%%%%%%%%%%%%%%%%%%%%%%%%%%%%%%%%%%%%%%%%%%%%%%%%%%%%%%%%%%%%%%%%%

\section{Operators acting in Fock space}\label{app-Fock}

\noindent
In this appendix we recall some standard definitions.
Let $(\cM,\mathfrak{A},\mu)$ be some measure space.
Then the bosonic Fock space modeled over
the one particle Hilbert space $L^2(\mu)$ is given as
a countable direct sum
$$
\Fock[L^2(\mu)]\,:=\,
\bigoplus_{n=0}^\infty\Fock^{(n)}[L^2(\mu)]
\,\ni\psi\,=\,(\psi^{(0)},\psi^{(1)},\psi^{(2)},\ldots\;)
\,,
$$
where $\Fock^{(0)}[L^2(\mu)]:=\CC$ and
$\Fock^{(n)}[L^2(\mu)]$
is the subspace of all  
$\otimes_1^n\mu$-square integrable functions
$\psi^{(n)}:\cM^n\to\CC$ such that
$$
\psi^{(n)}(k_{\pi(1)},\dots,k_{\pi(n)})\,=\,
\psi^{(n)}(k_1,\dots,k_n)\,,
$$ 
$\otimes_1^n\mu$-almost everywhere, for every
permutation $\pi:\{1,\dots,n\}\circlearrowleft$.
The vector $\Omega=(1,0,0,\dots)\in\Fock[L^2(\mu)]$
is called the vacuum in $\Fock[L^2(\mu)]$.
The second quantization of the multiplication
operator with a measurable function $q:\cM\to\RR$
is the self-adjoint operator defined by
\begin{align*}
\dom(d\Gamma(q))\,=\,
\Big\{\,&(\psi^{(n)})_{n=0}^\infty\in\Fock[L^2(\mu)]\::
\\
&
\sum_{n=1}^\infty\int\Big|\sum_{j=1}^nq(k_j)\,\psi^{(n)}(k_1,\dots,k_n)
\Big|^2\,d\mu(k_1)\dots d\mu(k_n)\,<\,\infty
\,\Big\}\,,
\end{align*}
and
$(d\Gamma(q)\,\psi)^{(0)}=0$ and 
$$
(d\Gamma(q)\,\psi)^{(n)}(k_1,\dots,k_n)\,=
\sum_{j=1}^nq(k_j)\,\psi^{(n)}(k_1,\dots,k_n)\,,\quad n\in\NN\,,
\;\psi\in\dom(d\Gamma(q)).
$$
By symmetry and Fubini's theorem we find, for non-negative $q$,
\begin{equation}\label{Hf=dGamma}
\SPb{d\Gamma(q)^{1/2}\,\phi}{d\Gamma(q)^{1/2}\,\psi}\,=\,\int q(k)\,
\SPn{a(k)\,\phi}{a(k)\,\psi}\,d\mu(k)\,,
\end{equation}
for all $\phi,\psi\in\dom(d\Gamma(q)^{1/2})$
where we use the notation 
\begin{equation}\label{def-a(k)}
(a(k)\,\psi)^{(n)}(k_1,\dots,k_n)\,=\,
(n+1)^{1/2}\,\psi^{(n+1)}(k,k_1,\dots,k_n)\,,\quad n\in\NN_0\,,
\end{equation}
almost everywhere,
and $a(k)\,\Omega=0$.
We further recall that the creation
and the annihilation operators
of a boson $f\in L^2(\mu)$ are given
by
\begin{align*}
(\ad(f)\,\psi)^{(n)}(k_1,\ldots,k_n)\,&=\,
n^{-\frac{1}{2}}\sum_{j=1}^nf(k_j)\,\psi^{(n-1)}(\ldots,k_{j-1},k_{j+1},\ldots),
\,\,\,\quad n\in\NN,
\\
(a(f)\,\psi)^{(n)}(k_1,\ldots,k_n)\,&=\,
(n+1)^{\frac{1}{2}}\int\ol{f}(k)\,\psi^{(n+1)}(k,k_1,\ldots,k_n)\,d\mu(k),
\, n\in\NN_0,
\end{align*}
and $(\ad(f)\,\psi)^{(0)}=0$, $a(f)\,\Omega=0$.
We define $\ad(f)$ and $a(f)$ on their 
maximal domains.
The following canonical commutation relations hold true
on $\dom(d\Gamma(1)^2)$,
$$
[a(f)\,,\,a(g)]\,=\,[\ad(f)\,,\,\ad(g)]\,=\,0\,,\qquad
[a(f)\,,\,\ad(g)]\,=\,\SPn{f}{g}\,\id\,,
$$
where $f,g\in\HP$. Moreover, we have
$\SPn{a(f)\,\phi}{\psi}=\SPn{\phi}{\ad(f)\,\psi}$, 
and, by definition, $a(f)\,\phi=\int \ol{f}(k)\,a(k)\,\phi\,d\mu(k)$,
for $\phi,\psi\in\dom(d\Gamma(1))$.

%%%%%%%%%%%%%%%%%%%%%%%%%%%%%%%%%%%%%%%%%%%%%%%%%%%%%%%%%%%%%%%%%%%%%%%%%%
%%%%%%%%%%%%%%%%%%%%%%%%%%%%%%%%%%%%%%%%%%%%%%%%%%%%%%%%%%%%%%%%%%%%%%%%%%
\bigskip

\noindent
{\bf Acknowledgement.}
This work has been partially supported by the DFG (SFB/TR12).

%%%%%%%%%%%%%%%%%%%%%%%%%%%%%%%%%%%%%%%%%%%%%%%%%%%%%%%%%%%%%%%%%%%%%%%%%%
%%%%%%%%%%%%%%%%%%%%%%%%%%%%%%%%%%%%%%%%%%%%%%%%%%%%%%%%%%%%%%%%%%%%%%%%%%
%%%%%%%%%%%%%%%%%%%%%%%%%%%%%%%%%%%%%%%%%%%%%%%%%%%%%%%%%%%%%%%%%%%%%%%%%%
%%%%%%%%%%%%%%%%%%%%%%%%%%%%%%%%%%%%%%%%%%%%%%%%%%%%%%%%%%%%%%%%%%%%%%%%%%

\bibliographystyle{plain} 

%\bibliography{bibliothek}

\def\cprime{$'$} \def\cprime{$'$} \def\cprime{$'$} \def\cprime{$'$}
  \def\cprime{$'$}

\end{document}